\pdfoutput=1
\documentclass[11pt]{article}

\ifx\comments\undefined
  \usepackage[letterpaper,margin=.7in,includefoot]{geometry}
  \newcommand{\XXXcomment}[1]{}
\else
  \usepackage[letterpaper,margin=.5in,right=2.5in,marginpar=2in,includefoot]{geometry}
  \newcommand{\XXXcomment}[1]{\marginpar{\color{blue}{\footnotesize #1}}}
\fi

\usepackage{amssymb,amsthm,latexsym,amsmath,mathrsfs,stmaryrd,graphicx,fancyhdr,paralist,mathtools,paralist}
\allowdisplaybreaks

\usepackage[pdftex,
            bookmarks=true,
            bookmarksnumbered=true,
            pdfborder={0 0 0},
            plainpages=false,
            pdfpagelabels]{hyperref}

\usepackage{float}
\floatstyle{boxed}
\restylefloat{figure}
\restylefloat{table}

\usepackage{caption,float}
\usepackage{tikz}
\usetikzlibrary{trees,arrows,automata,shapes,decorations,decorations.pathmorphing}

\newcommand{\mystretch}{\renewcommand{\arraystretch}{1.5}}
\newcommand{\normalstretch}{\renewcommand{\arraystretch}{1}}

\usepackage{bussproofs}

\EnableBpAbbreviations

\newcommand{\RLs}[1]{\RightLabel{\small #1}}

\usepackage{changepage}

\newenvironment{pquote}
{\begin{adjustwidth}{.03\textwidth}{}\setlength{\parindent}{0em}\setlength{\parskip}{.7em}\vspace{-.5em}}
  {\end{adjustwidth}}

\newcommand{\imp}{\to}                         

\newcommand{\down}{{\downarrow}}               %

\newcommand{\col}{\,{:}\,}                     
\newcommand{\colu}[1]{\,{:}^{#1}}              
\newcommand{\pcolu}[1]{{:}^{#1}}               

\newcommand{\sem}[1]{\llbracket{#1}\rrbracket} 

\newcommand{\Cn}{\mathsf{Cn}}

\newcommand{\sub}{\mathsf{sub}}                

\newcommand{\Nat}{\mathbb{N}}                  
\newcommand{\Prop}{\mathscr{P}}                
\newcommand{\Term}{\mathscr{T}}                
\newcommand{\Lang}{\mathscr{L}}                

\newcommand{\calS}{\mathcal{S}}

\newcommand{\CPL}{{\mathsf{CPL}}}       
\newcommand{\CDL}{{\mathsf{CDL}}}       
\newcommand{\JCDL}{{\mathsf{JCDL}}}     

\newcommand{\BRSIC}{{\mathsf{BRSIC}}} 

\newcommand{\cc}{{\mathsf{cc}}}

\theoremstyle{definition} 
\newtheorem{definition}{Definition}[section]
\newtheorem{theorem}[definition]{Theorem}

\newtheorem{lemma}[definition]{Lemma}

\newtheorem{remark}[definition]{Remark}

\newcommand{\ourtitle}{Revisable Justified Belief: Preliminary Report}

\newcommand{\alexandru}{Alexandru Baltag}

\newcommand{\bryan}{Bryan Renne}

\newcommand{\bryanFunding}{Funded by an Innovational Research
  Incentives Scheme Veni grant from the Netherlands Organisation for
  Scientific Research (NWO) hosted by the Institute for Logic,
  Language, Information and Computation (ILLC) at the University of
  Amsterdam.}

\newcommand{\sonja}{Sonja Smets}

\newcommand{\sonjaFunding}{Funded in part by an Innovational Research
  Incentives Scheme Vidi grant from the Netherlands Organisation for
  Scientific Research (NWO) and by the European Research Council under
  the European Community's Seventh Framework Programme (FP7/2007-2013)
  / ERC Grant agreement no.~283963. Both grants were hosted by the
  Institute for Logic, Language, Information and Computation (ILLC) at
  the University of Amsterdam.}

\hypersetup{ 
  pdfauthor={\alexandru{}, \bryan{}, \sonja{}}, 
  pdftitle={\ourtitle{}}
}

\title{\ourtitle{}}

\author{\alexandru{} \and \bryan{}\thanks{\bryanFunding{}} \and
  \sonja{}\thanks{\sonjaFunding{}}}

\date{}

\begin{document}
\maketitle
\thispagestyle{fancy}

\begin{abstract}
  The theory $\CDL$ of \emph{Conditional Doxastic Logic} is the
  single-agent version of Board's multi-agent theory $\BRSIC$ of
  conditional belief. $\CDL$ may be viewed as a version of AGM belief
  revision theory in which Boolean combinations of revisions are
  expressible in the language. We introduce a theory $\JCDL$ of
  \emph{Justified Conditional Doxastic Logic} that replaces
  conditional belief formulas $B^\psi\varphi$ by expressions
  $t\colu\psi\varphi$ made up of a term $t$ whose syntactic structure
  suggests a derivation of the belief $\varphi$ after revision by
  $\psi$. This allows us to think of terms $t$ as reasons justifying a
  belief in various formulas after a revision takes place.  We show
  that $\JCDL$-theorems are the exact analogs of $\CDL$-theorems, and
  that this result holds the other way around as well. This allows us
  to think of $\JCDL$ as a theory of revisable justified belief.
\end{abstract}

\section{Introduction} 

\emph{Conditional Doxastic Logic} is Baltag and Smets'
\cite{BalSme08:LOFT} name for a single-agent version of Board's
\cite{Boa04:GEB} multi-agent theory of conditional belief $\BRSIC$.
$\CDL$ has formulas $B^\psi\varphi$ to express that the agent believes
$\varphi$ conditional on $\psi$.  As we will see, $\CDL$ has a certain
relationship with the ``AGM theory'' of belief revision due to due to
Alchourr\'on, G\"ardenfors, and Makinson \cite{AGM}. So we may also
think of $B^\psi\varphi$ as say that the agent will believe $\varphi$
after revising her belief state by successfully incorporating the
information that $\psi$ is true. As with AGM theory, $\CDL$ assumes
that conditionalization (i.e., revision) is always successful: the
agent is to assume that the incoming information $\psi$ is completely
trustworthy and therefore update her belief state by consistently
incorporating this incoming information. If before the revision she
holds beliefs that imply $\lnot\psi$, then she must give up these
beliefs so that she will come to believe $\psi$ after the revision
takes place. The question then is how to do this in general.

The models of $\CDL$ are ``plausibility models.''  These are very
close to Grove's ``system of spheres'' for AGM theory
\cite{Gro88:JPL}. As we will see, we can view the
conditionalization/revision process either from the semantic
perspective, as a definite operation on plausibility models, or from
the syntactic perspective, as an axiomatic formulation analogous to
the postulate-based approach of AGM.  However, it is perhaps simplest
to start with the semantic perspective.


\section{Plausibility models}

\begin{definition}[Plausibility models]
  \label{definition:plausibility-model}
  Let $\Prop$ be a set of propositional letters.  A \emph{plausibility
    model} is a structure $M={(W,\leq,V)}$ consisting of a nonempty
  set $W$ of ``worlds,'' a preorder (i.e., a reflexive and transitive
  binary relation) $\leq$ on $W$, and a propositional valuation
  $V:W\to\wp(\Prop)$ mapping each world $w$ to the set $V(w)$ of
  propositional letters true at $w$.  We call $\leq$ a
  \emph{plausibility relation} on $W$.  In terms of $\leq$, we define
  the converse relation $\geq$, the strict version $<$, the strict
  converse relation $>$, and the various negations of these (denoted
  by writing a slash through the symbol to be negated) as usual.
  $\simeq$ denotes the relation ${\simeq}\coloneqq({\geq}\cap{\leq})$
  of \emph{equi-plausibility}.  A \emph{pointed plausibility model} is
  a pair $(M,w)$ consisting of a plausibility model $M$ and the
  \emph{point} $w$, itself a world in $M$.  Notation: for each
  $w\in W$, we define the set
  \[
  w^\down\coloneqq\{x\in W\mid x\leq w\}\enspace.
  \]
\end{definition}

$x\leq y$ read, ``$x$ is no less plausible than $y$.''  According to
this reading, if we think of $\leq$ as a ``less than or equal to''
relation, then it is the ``lesser'' elements that are \emph{more}
plausible.  Therefore, if $\leq$ is a well-order and $S$ is a nonempty
set of worlds, $\min S$ is the set of worlds that are the most
plausble in $S$.  While it may at first seem counterintuitive to the
uninitiated, this convention of ``lesser is more plausible'' is
nevertheless standard in Belief Revision Theory.\footnote{This
  convention stems from the notion of ``Grove spheres''
  \cite{Gro88:JPL}: given a well-order, worlds are arranged so that
  they sit on the surface of a number of concentric spheres.  Worlds
  of strictly greater plausibility are assigned to spheres with
  strictly shorter radii, and equi-plausible worlds are assigned to
  the same sphere.  In this way, the most plausible worlds sit on the
  surface of the innermost sphere, which has the \emph{minimum}
  radius.  Similarly, if we restrict attention to a nonempty set $S$
  of worlds, then we ``recenter'' the sphere around $S$. By this we
  mean that we create a new system of spheres consisting of just those
  worlds in $S$.  After doing so, the most plausible worlds again sit
  on the surface of the innermost sphere, which again has the minimum
  radius.}

We think of the plausibility relation as describing the judgments of
an unnamed agent: for each pair of worlds $(x,y)$, she either judges
one world to be more plausible than the other or the two to be of
equal plausibility.  Plausibility models for multiple agents have a
number of plausibility relations, one for each agent.  For present
purposes, we restrict attention to the single-agent case, though we
say more about the multi-agent situation later.

\begin{definition}[Plausibility model terminology]
  \label{definition:smooth}
  Let $M=(W,\leq,V)$ be a plausibility model.
  \begin{itemize}
  \item To say $M$ is \emph{finite} means $W$ is finite.

  \item To say $M$ is \emph{connected} means that for each $x\in W$,
    we have $\cc(x)=W$, where  
    \[
    \cc(x)\coloneqq \{y\in x\mid x({\geq}\cup{\leq})^+y\}\enspace
    \]
    is the \emph{connected component of $x$} and
    $({\geq}\cup{\leq})^+$ is the transitive closure of
    ${\geq}\cup{\leq}$.\footnote{The \emph{transitive closure} of a
      binary relation $R$ is the smallest extension $R^+\supseteq R$
      satisfying the property that $xR^+y$ and $yR^+z$ together imply
      $xR^+z$.} A \emph{connected component} is a subset
    $S\subseteq W$ for which there exists an $x\in W$ such that
    $\cc(x)=S$.

  \item To say $M$ is \emph{well-founded} means $\leq$ is
    well-founded: for each nonempty $S\subseteq W$, the set
    \[
    \min S\coloneqq\{x\in S\mid\forall y\in S:y\not<x\}
    \]
    of \emph{minimal elements of $S$} is nonempty.

  \item To say that a set $S\subseteq W$ of worlds is \emph{smooth} in
    $M$ means that for each world $x\in S$, either $x\in\min(S)$ or
    there exists $y\in\min(S)$ such that $y<x$. Given a collection
    $\Gamma\subseteq\wp(W)$ of sets of worlds, to say that $M$ is
    \emph{smooth with respect to} $\Gamma$ means that every
    $S\in\Gamma$ is smooth in $M$.  To say that $M$ is \emph{smooth}
    means that $M$ is smooth with respect to $\wp(W)$.

  \item To say that $M$ is \emph{total} means that $\leq$ is total on
    $W$: for each $(x,y)\in W\times W$, we have $x\leq y$ or
    $y\leq x$.

  \item To say $M$ is \emph{well-ordered} (equivalently, that $M$ is a
    \emph{well-order}) means that $\leq$ is well-ordered (i.e., it is
    total and well-founded).

  \item To say that $M$ is \emph{locally total} means that $\leq$ is
    \emph{total on each connected component\/}: for each $w\in W$ and
    $(x,y)\in\cc(w)\times\cc(w)$, we have $x\leq y$ or $y\leq x$.

  \item To say $M$ is \emph{locally well-ordered} (equivalently, that
    $M$ is a \emph{local well-order}) means that $\leq$ is locally
    well-ordered (i.e., it is well-founded and total on each connected
    component).
  \end{itemize}
  To say that a pointed plausibility model $(M,w)$ satisfies
  one of the model-applicable adjectives above means that
  $M$ itself satisfies the adjective in question.
\end{definition}

\begin{theorem}[Relationships between terminology]
  \label{theorem:wf-smooth}
  Let $M=(W,\leq,V)$ be a plausibility model and $S\subseteq W$.
  \begin{enumerate}
  \item \label{i:finite} If $M$ is finite, then $M$ is well-founded.

  \item \label{i:lwo} If $M$ is locally well-ordered, and $S$ is a
    connected component, then
    \[
    \min S=\{x\in S\mid\forall y\in S:x\leq y\}\enspace.
    \]

  \item \label{i:wo} If $M$ is well-ordered, then
    $\min S=\{x\in S\mid\forall y\in S:x\leq y\}$.

  \item \label{i:wf-smooth} $M$ is well-founded iff $M$ is smooth.

  \item \label{i:wo-smooth} $M$ is well-ordered iff $M$ is smooth and
    total.

  \item \label{i:lwo-smooth} $M$ is locally well-ordered iff $M$ is
    smooth and locally total.
  \end{enumerate}
\end{theorem}
\begin{proof}
  See the appendix.
\end{proof}

From now on, in this paper we will restrict ourselves to locally
well-ordered plausibility models, unless otherwise specified.

\section{Conditional Doxastic Logic}

\subsection{Language and semantics}

\begin{definition}[$\Lang_\CDL$]
  \label{definition:lang-CDL}
  Let $\Prop$ be a fixed set of propositional letters.  The language
  of \emph{Conditional Doxastic Logic} consists of the set of
  $\Lang_\CDL$ formulas $\varphi$ formed by the following grammar:
  \[
  \varphi \Coloneqq \bot \mid p \mid (\varphi\to\varphi) \mid B^\varphi\varphi
  \qquad\text{\small $p\in\Prop$}
  \]
  The logical constant $\top$ (for truth) and the various familiar
  Boolean connectives are defined by the usual abbreviations.  Other
  important abbreviations: $B(\varphi|\psi)$ denotes $B^\psi\varphi$,
  and $B\varphi$ denotes $B^\top\varphi$.
\end{definition}

The formula $B^\psi\varphi$ is read, ``Conditional on $\psi$, the
agent believes $\varphi$.''  Intuitively, this means that each of the
most plausible $\psi$-worlds satisfies $\varphi$.  The forthcoming
semantics will clarify this further.  The basic idea is that a belief
conditional on $\psi$ is a belief the agent would hold were she to
minimally revise her beliefs so that she comes to believe $\psi$.

\begin{definition}[$\Lang_\CDL$-truth]
  \label{definition:CDL-truth} 
  Let $M=(W,\leq,V)$ be a locally well-ordered plausibility model.  We
  define a binary satisfaction relation $\models$ between locally
  well-ordered pointed plausibility models $(M,w)$ (written without
  surrounding parentheses) and $\Lang_\CDL$-formulas and we define a
  function $\sem{-}:\Lang_\CDL\to\wp(W)$ as follows.
  \begin{itemize}
  \item $\sem{\varphi}_M\coloneqq\{v\in W\mid
    M,v\models\varphi\}$. The subscript $M$ may be suppressed.

  \item $M,w\not\models\bot$.
    
  \item $M,w\models p$ iff $p\in V(w)$ for $p\in\Prop$.

  \item $M,w\models\varphi\imp\psi$ iff $M,w\not\models\varphi$ or
    $M,w\models\psi$.

  \item $M,w\models B^\psi\varphi$ iff for all $x\in\cc(w)$, we have
    \[
    x^\down\cap\sem{\psi}=\emptyset
    \quad\text{or}\quad
    \exists y\in x^\down\cap\sem{\psi}:
    y^\down\cap\sem{\psi}\subseteq\sem{\varphi}
    \enspace.
    \]

    $B^\psi\varphi$ holds at $w$ iff for every world connected to $w$
    that has an equally or more plausible $\psi$-world $y$, the
    $\psi$-worlds that are equally or more plausible than $y$ satisfy
    $\varphi$.
  \end{itemize} 
  We extend the above so that we may have sets $S\subseteq\Lang_\CDL$
  of formulas on the right-hand side:
  \begin{center}
    $M,w\models S$ \quad means\quad $M,w\models\varphi$ for each
    $\varphi\in S$\enspace.
  \end{center}
  Also, we will have occasion to use the following notion of
  \emph{local consequence\/}: given a set
  $S\cup\{\varphi\}\subseteq\Lang_\CDL$ of formulas and writing
  $\mathfrak{P}_*$ to denote the class of pointed plausibility models,
  \begin{center}
    $S\models_\ell\varphi$ \quad means\quad
    $\forall(M,w)\in\mathfrak{P}_*$ : $M,w\models S$ implies
    $M,w\models\varphi$ \enspace.
  \end{center}
  Finally, we write $M\models\varphi$ to mean that $M,v\models\varphi$
  for each world $v$ in $M$ (``$\varphi$ is valid within $M$'').
\end{definition}

\begin{remark}[Knowledge]
  \label{remark:knowledge}
  Baltag and Smets \cite{BalSme08:LOFT} read the abbreviation
  \[
  K\varphi \quad\coloneqq\quad B^{\lnot\varphi}\bot
  \]
  as ``the agent knows $\varphi$.''  This notion of ``knowledge'' is
  based on the rejection of a proposed belief revision.  In
  particular, $K\varphi=B^{\lnot\varphi}\bot$ says that the most
  plausible $\lnot\varphi$-worlds are $\bot$-worlds.  The
  propositional constant $\bot$ for falsehood is true nowhere, so this
  amounts to us saying that the agent does not consider any
  $\lnot\varphi$-worlds possible. Hence all the worlds she considers
  possible are $\varphi$-worlds.  It is in this sense that we say she
  ``knows'' that $\varphi$ is true: she will not revise her beliefs by
  $\lnot\varphi$ (on pain of contradiction).  It is easy to see that
  the semantics ensures that $K$ so-defined is an $\mathsf{S5}$ modal
  operator: knowledge is closed under classical implication, what is
  known is true, it is known what is known, it is known what is not
  known, and all validities are known.
\end{remark}

In well-founded plausibility models, belief in $\varphi$ conditional
on $\psi$ is equivalent to having $\varphi$ true at the most plausible
$\psi$-worlds that are within the connected component of the actual
world. And if the models are well-ordered, then we can omit mention of
the connected component.

\begin{theorem}[$\Lang_\CDL$-truth in well-founded models]
  \label{theorem:CDL-wf}\label{theorem:CDL-wo}
  Let $M=(W,\leq,V)$ be a plausibility model.
  \begin{enumerate}[\;\;(a)]
  \item \label{i:CDL-wf} If $M$ is well-founded:
    $M,w\models B^\psi\varphi$ $\;\;\Leftrightarrow\;\;$
    $\min\sem{\psi}\cap\cc(w)\subseteq\sem{\varphi}$.

  \item \label{i:CDL-wo} If $M$ is well-ordered:
    $M,w\models B^\psi\varphi$ $\;\;\Leftrightarrow\;\;$
    $\min\sem{\psi}\subseteq\sem{\varphi}$.
  \end{enumerate}
\end{theorem}
\begin{proof}
  See the appendix for the proof of \eqref{i:CDL-wf}.  For
  \eqref{i:CDL-wo}, if $M$ is well-ordered, then $\leq$ is total and
  therefore $\cc(w)=W$ for each $w\in W$.  Apply \eqref{i:CDL-wf}.
\end{proof}

The intended models for $\CDL$ are the well-ordered (and hence
well-founded) plausibility models.

\begin{definition}[$\Lang_\CDL$-validity]
  \label{definition:CDL-validity}
  To say that a $\Lang_\CDL$-formula $\varphi$ is \emph{valid},
  written $\models\varphi$, means that we have $M\models\varphi$
  for each well-ordered plausibility model $M$.
\end{definition}

As per Theorem~\ref{theorem:CDL-wo}, restricting validity to the
well-orders allows us to read $B^\psi\varphi$ as follows: ``the most
plausible $\psi$-worlds satisfy $\varphi$.'' While the intended models
for $\CDL$ are well-ordered, and validity is defined accordingly (as
per Definition~\ref{definition:CDL-validity}), the following theorem
shows that locally well-ordered plausibility models would suffice.

\begin{theorem}[$\Lang_\CDL$-validity with respect to local
  well-orders]
  \label{theorem:CDL-lwo}
  Let $\mathfrak{P}_L$ be the class of locally
  well-ordered plausibility models.  For each $\varphi\in\Lang_\CDL$,
  we have:
  \[
  \models\varphi
  \qquad\text{iff}\qquad
  \forall M\in\mathfrak{P}_L,\;
  M\models\varphi\enspace.
  \]
\end{theorem}
\begin{proof}
  Right to left (``if''): obvious. Left to right (``only if''): assume
  $\models\varphi$ and take $M=(W,\leq,V)\in\mathfrak{P}_L$ and a
  world $w\in W$. Let $M'$ be the sub-model of $M=(W',\leq',V')$
  obtained by restricting to $\cc(w)$:
  \[
  W'=\cc(w),\quad
  \leq'=\leq\cap(W'\times W'),\quad
  V'(v)=V(v) \text{ for } v\in W'.
  \]
  Since $M\in\mathfrak{P}_L$, it follows that $M'$ is well-ordered and
  therefore, since $\models\varphi$, we have
  $M',w\models\varphi$.  It follows by a straightforward induction on
  the construction of $\Lang_\CDL$-formulas $\theta$ that
  $M',w\models\theta$ iff $M,w\models\theta$. Hence
  $M,w\models\varphi$.  Since $w\in W$ and $M\in\mathfrak{P}_L$ were
  chosen arbitrarily, we conclude that $M\models\varphi$ for each
  $M\in\mathfrak{P}_L$.
\end{proof}

\subsection{Board's theory \texorpdfstring{$\BRSIC$}{BRSIC} and
  \texorpdfstring{$\CDL_0$}{CDL0}}

The Hilbert theory of Conditional Doxastic Logic was first studied by
Board \cite{Boa04:GEB} under the name $\BRSIC$.  Baltag and Smets
\cite{BalSme08:LOFT} subsequently developed various alternative
axiomatizations and extensions and introduced the name
\emph{Conditional Doxastic Logic}.  The single-agent version of
Board's theory $\BRSIC$ is equivalent to what we call $\CDL_0$.

\begin{definition}[$\CDL_0$ theory]
  $\CDL_0$ is defined in Table~\ref{table:CDL0}.
\end{definition}

\begin{table}
  \begin{center}
    \textsc{Axiom Schemes}
    \\[.3em]
    \renewcommand{\arraystretch}{1.3}
    \begin{tabular}{rl}
      (CL) & 
      Schemes for Classical Propositional Logic
      \\
      (K) &
      $B^\psi(\varphi_1\imp\varphi_2)\to
      (B^\psi\varphi_1\to B^\psi\varphi_2)$
      \\
      (Succ) &
      $B^\psi\psi$
      \\
      (IEa) &
      $B^\psi\varphi\to
      (B^{\psi\land\varphi}\chi\leftrightarrow
      B^\psi\chi)$
      \\
      (IEb) &
      $\lnot B^\psi\lnot\varphi\to
      (B^{\psi\land\varphi}\chi\leftrightarrow
      B^\psi(\varphi\to\chi))$
      \\
      (PI) &
      $B^\psi\chi\to B^\varphi B^\psi\chi$
      \\
      (NI) &
      $\lnot B^\psi\chi\to B^\varphi\lnot B^\psi\chi$
      \\
      (WCon) & 
      $B^\psi\bot\to\lnot\psi$
    \end{tabular}
    \\[1em]
    \textsc{Rules}
    \\[.5em]
    \AXC{$\varphi\to\psi$}
    \AXC{$\varphi$}
    \RLs{(MP)}
    \BIC{$\psi$}
    \DP
    \qquad
    \AXC{$\varphi$}
    \RLs{(MN)}
    \UIC{$B^\psi\varphi$}
    \DP
    \qquad
    \AXC{$\psi\leftrightarrow\psi'$}
    \RLs{(LE)}
    \UIC{$B^\psi\varphi\leftrightarrow B^{\psi'}\varphi$}
    \DP
  \end{center}
  \caption{The theory $\CDL_0$, a single-agent variant of Board's
    theory $\BRSIC$ \cite{Boa04:GEB}}
  \label{table:CDL0}
\end{table}

$\CDL_0$ is actually a simplification of $\BRSIC$.  In particular,
$\BRSIC$ is a multi-agent theory for a nonempty set $A$ of agents
using a language similar to $\Lang_\CDL$ except that it has as
primitives both conditional belief $B^\psi_a\varphi$ for each agent
$a\in A$ and unconditional belief $B_a\varphi$ for each agent
$a\in A$.  Since the $\BRSIC$ axiom
$B_a\varphi\leftrightarrow B^\top_a\varphi$ (``Triv'') requires that
unconditional belief be equivalent to conditional belief based on a
tautological conditional, we have decided upon a streamlined language
that contains conditional belief only.  This allowed us to define away
Board's axiom Triv in the following way (in
Definition~\ref{definition:lang-CDL}): let $B_a\varphi$ abbreviate
$B_a^\top\varphi$.  We have also renamed some of Board's axioms and
rules: his Taut is now called (CL), his Dist is now called (K), his
IE(a) is now called (IEa), his IE(b) is now called (IEb), his TPI is
now called (PI), his NPI is now called (NI), his RE is now called
(MN), and all of his other axiom names have been enclosed in
parenthesis.  Finally, what we call (WCon) is the contrapositive of
what Board called WCon. Restricting to a single-agent setting and
thereby dropping subscripted agent names, we obtain the theory
$\CDL_0$.

\begin{remark}[Multi-agent $\CDL_0$]
  \label{remark:multi-agent-CDL0}
  A multi-agent version of $\CDL_0$ is obtained by making trivial
  modifications to the language, axiomatization, and semantics of
  $\CDL_0$.  In particular, for a nonempty set $A$ of agents, the
  multi-agent language $\Lang_\CDL^A$ is like the single-agent
  language $\Lang_\CDL$ except that each conditional belief operator
  $B^\psi$ is replaced by a number of operators $B^\psi_a$, one for
  each agent $a\in A$. The multi-agent theory $\CDL_0^A$ is obtained
  by adding a metavariable agent subscript $a$ to each of the belief
  operators in Table~\ref{table:CDL0}. The models of $\CDL_0^A$ are
  \emph{multi-agent plausibility models}: these are like single-agent
  plausibility models presented above (in
  Definition~\ref{definition:plausibility-model}) except that the
  preorder $\leq$ is replaced by a preorder $\leq_a$ for each agent
  $a\in A$.  The definition of truth for $\Lang_\CDL^A$ on these
  models is like that in Definition~\ref{definition:CDL-truth} except
  that the meaning of $M,w\models B^\psi_a\varphi$ is changed so as to
  refer to the preorder $\leq_a$. Validity is defined with respect to
  the class of multi-agent plausibility models satisfying the property
  that each $\leq_a$ is locally well-ordered.

  So, in essence, the multi-agent version consists of multiple single
  agent versions such that each agent's conditional beliefs are always
  restricted to the worlds connected (for that agent) to the given
  world $w$ currently under consideration. It is clear that
  restricting to one agent $a$ yields a framework that is equivalent
  to the version of $\CDL_0$ we have presented here.

  We note that the multi-agent version allows us to describe what one
  agent conditionally believes about what another agent conditionally
  believes. This is feature of interest in a wide variety of
  applications.  However, from the technical perspective, the
  difference between the single- and multi-agent frameworks does not
  amount to too much in the way of mathematical shenanigans.  It
  therefore suffices to indicate, as we have here, how the multi-agent
  version is obtained from the single-agent version and then restrict
  our study to the single-agent version. Of course, one may consult
  Board \cite{Boa04:GEB} for the fully specified account of the
  multi-agent theory $\BRSIC$.
\end{remark}

\subsection{The theory \texorpdfstring{$\CDL$}{CDL}}

It will be our task in this paper to develop a version of Conditional
Doxastic Logic with justifications in the tradition of Justification
Logic \cite{ArtFit12:SEP}. We will say more about this later, but for
now it suffices to say that justifications in this tradition are meant
to encode the individual reasoning steps that the agent uses to
support her belief in one statement based on justifications she has
for beliefs in other statements.  In this way, justifications are
supposed to present a stepwise explanation for how the agent derives
complex beliefs from more basic ones.  It is in this sense that
justifications are ``proof-like.'' 

In order to make this precise, we require an axiomatization of
Conditional Doxastic Logic that is more perspicious than is $\CDL_0$
with regard to the ways in which conditional beliefs obtain.  In
particular, (IEa), (IEb), and (LE) are powerful principles that in fact
encode a number of more basic principles and, as such, these powerful
principles compress a number of reasoning steps into a small number of
postulates. This is especially obvious with (LE): a belief conditional
$\psi$ may be replaced by a provably equivalent conditional $\psi'$ in
one step, which does not reflect the complexity of the derivation that
was used to prove the equivalence $\psi\leftrightarrow\psi'$.  From
the point of view of the Justification Logic tradition, wherein
justifications should explain in a stepwise fashion how one
conditional belief follows from another, this is
undesirable. Intuitively, if the agent believes $\varphi$ conditional
on $\psi$, then the reason she believes $\varphi$ conditional on a
provably equivalent $\psi'$ depends crucially on the reasoning as to
why $\psi'$ is in fact equivalent to $\psi$. If $\pi_1$ and a more
complex $\pi_2$ are proofs of this equivalence, then an agent who
bases her belief on the more complex $\pi_2$ should have a
correspondingly more complex justification witnessing her belief.  We
therefore require an alternative but equivalent axiomatization of the
theory $\CDL_0$ that makes such stepwise reasoning operations more
explicit.  The exact criteria we seek for such a theory are not
precisely defined but are based on the authors' experience in working
in the Justification Logic tradition.  We call the theory we have
settled upon $\CDL$, and later we will explain how this theory gives
rise to a theory of Conditional Doxastic Logic with justifications.

\begin{definition}[$\CDL$ theory]
  $\CDL$ is defined in Table~\ref{table:CDL}.
\end{definition}

\begin{table}
  \begin{center}
    \textsc{Axiom Schemes}
    \\[.3em]
    \renewcommand{\arraystretch}{1.3}
    \begin{tabular}{rl}
      (CL) & 
      Schemes for Classical Propositional Logic
      \\
      (K) &
      $B^\psi(\varphi_1\imp\varphi_2)\to
      (B^\psi\varphi_1\to B^\psi\varphi_2)$
      \\
      (Succ) &
      $B^\psi\psi$
      \\
      (KM) &
      $B^\psi\bot\to B^{\psi\land\varphi}\bot$
      \\
      (RM) &
      $\lnot B^\psi\lnot\varphi\to(
      B^\psi\chi\to B^{\psi\land\varphi}\chi)$
      \\
      (Inc) &
      $B^{\psi\land\varphi}\chi\to B^\psi(\varphi\to\chi)$
      \\
      (Comm) &
      $B^{\psi\land\varphi}\chi\to B^{\varphi\land\psi}\chi$
      \\
      (PI) &
      $B^\psi\chi\to B^\varphi B^\psi\chi$
      \\
      (NI) &
      $\lnot B^\psi\chi\to B^\varphi\lnot B^\psi\chi$
      \\
      (WCon) & 
      $B^\psi\bot\to\lnot\psi$
    \end{tabular}
    \\[1em]
    \textsc{Rules}
    \\[.5em]
    \AXC{$\varphi\to\psi$}
    \AXC{$\varphi$}
    \RLs{(MP)}
    \BIC{$\psi$}
    \DP
    \qquad
    \AXC{$\varphi$}
    \RLs{(MN)}
    \UIC{$B^\psi\varphi$}
    \DP
  \end{center}
  \caption{The theory $\CDL$}
  \label{table:CDL}
\end{table}


The scheme (CL) of \emph{Classical Logic} and the rule (MP) of
\emph{Modus Ponens} tell us that $\CDL$ is an extension of Classical
Propositional Logic.  The rule (MN) of \emph{Modal Necessitation}
tells us that derivable formulas hold in all conditional belief
states.

Scheme (K) is just Kripke's axiom for our conditional belief operator
$B^\psi$. The scheme (Succ) of \emph{Success} says that every belief
revision is always successful: if the agent revises her belief based
on the information that $\psi$, then she will always arrive in a
belief state in which $\psi$ is one of her beliefs.

Making use of the definition of knowledge
$K\varphi\coloneqq B^{\lnot\varphi}\bot$ from
Remark~\ref{remark:knowledge}, we can look at the following special
case of the scheme (KM) of \emph{Knowledge Monotonicity\/}:
\[
B^{\lnot\psi}\bot\to B^{\lnot\psi\land\lnot\varphi}\bot\enspace.
\]
Since $\lnot\lnot\psi$ is equivalent to $\psi$ and
$\lnot(\lnot\psi\land\lnot\varphi)$ is equivalent to
$\psi\lor\varphi$, we may interpret the above instance of (KM) as
telling us that knowledge is closed under disjunction: if $\psi$ is
known, then so is $\psi\lor\varphi$.  But another interpretation
perhaps better explains the word ``Monotonicity'' in the name of this
scheme.  Returning now to the official formulation
\[
B^\psi\bot\to B^{\psi\land\varphi}\bot
\]
of (KM), this scheme tells us that if we can conclude that a belief
state conditional on $\psi$ is contradictory, then conjunctively
adding any further information $\varphi$ yields a belief state
conditional on $\psi\land\varphi$ that is still
contradictory. Accordingly, the belief state is unchanged by the
conjunctive addition of any further conditional information, and so
the belief state is trivially ``monotonic'' in the conjunctive
addition of conditional information.

The scheme (RM) of \emph{Rational Monotonicity} permits a more subtle
kind of conjunctive addition.  This scheme,
\[
\lnot B^\psi\lnot\varphi\to(B^\psi\chi\to
B^{\psi\land\varphi}\chi)\enspace,
\]
says that if $\varphi$ is \emph{consistent} with the belief state
conditional on $\psi$, then we may conjunctively add $\varphi$ to our
conditional without losing any beliefs from the original belief state.
This is a non-trivial monotonicity: incorporating the information
$\varphi$ by forming the conditional $\psi\land\varphi$ yields a
belief state that includes all the beliefs from the belief state
conditional on $\psi$, but it may also include more.

The scheme (Inc) of \emph{Inclusion} says that a belief state
conditional on $\psi$ includes every $\chi$ implied by $\varphi$
whenever the belief state conditional on the conjunction
$\psi\land\varphi$ includes $\chi$.  The scheme (Comm) of
\emph{Commutativity} says that the belief state conditional on a
conjunction is invariant to the ordering of the conjuncts. The schemes
(PI) of \emph{Positive Introspection} and (NI) or \emph{Negative
  Introspection} tell us that conditional beliefs are identical in all
belief states.  The scheme (WCon) of \emph{Weak Consistency} tells us
that belief revision is consistent with the actual state of affairs:
if a revision by $\psi$ yields a contradictory belief state, then
$\psi$ cannot be true.

\begin{remark}[Classical reasoning (CR), modal reasoning (MR)]
  When discussing derivation in $\CDL$, we will often suppress
  elementary reasoning steps familiar from the study of normal modal
  logics.  Toward this end, ``classical reasoning,'' which may be
  denoted by (CR), refers to a derivation with one or more steps that
  makes use solely of (CL) and (MP).  ``Modal reasoning,'' which may
  be denoted by (MR), refers to a derivation with one or more steps
  that makes use solely of (CL), (K), (MP), and (MN).
\end{remark}

\begin{theorem}[$\CDL$-theorems]
  \label{theorem:CDL-theorems}
  The following schemes of (Cut), \emph{Cautious Monotonicity} (CM),
  (Taut), (And), (Or), \emph{Positive Reduction} (PR), and
  \emph{Negative Reduction} (NR) are all derivable in $\CDL$:
  \begin{align*}
    \text{(Cut)} 
    &\quad
      B^\psi\varphi\to(B^{\psi\land\varphi}\chi\to B^\psi\chi)
    \\
    \text{(CM)} 
    &\quad
      B^\psi\varphi\to(B^\psi\chi\to B^{\psi\land\varphi}\chi)
    \\
    \text{(Taut)}
    &\quad
      B\varphi \leftrightarrow B^\top\varphi
    \\
    \text{(And)}
    &\quad
      B^\psi\varphi_1\to
      (B^\psi\varphi_2\to B^\psi(\varphi_1\land\varphi_2))
    \\
    \text{(Or)}
    &\quad
      B^{\psi_1}\varphi\to
      (B^{\psi_2}\varphi\to B^{\psi_1\lor\psi_2}\varphi)
    \\
    \text{(PR)}
    &\quad
      B^\varphi B^\psi\chi
      \leftrightarrow
      ( B^\varphi\bot\lor B^\psi\chi)
    \\
    \text{(NR)}
    &\quad
      B^\varphi\lnot B^\psi\chi
      \leftrightarrow
      ( B^\varphi\bot\lor\lnot B^\psi\chi)
  \end{align*}
  Also, the following rules of \emph{(Left) Logical Equivalence} (LE),
  \emph{Right Weakening} (RW), and \emph{Supraclassicality} (SC) are
  all derivable in $\CDL$:
  \begin{center}
    \AXC{$\psi\leftrightarrow\psi'$}
    \RLs{(LE)}
    \UIC{$B^{\psi}\chi\leftrightarrow B^{\psi'}\chi$} 
    \DP  
    \qquad
    \AXC{$\chi\to\chi'$}
    \RLs{(RW)}
    \UIC{$B^\psi\chi\to B^\psi\chi'$} 
    \DP  
    \qquad
    \AXC{$\psi\to\chi$}
    \RLs{(SC)}
    \UIC{$B^\psi\chi$}
    \DP
  \end{center}
\end{theorem}
\begin{proof}
  See the appendix.
\end{proof}

The following result shows that $\CDL$ and $\CDL_0$ derive the same
theorems and therefore that these theories are identical.

\begin{theorem}[$\CDL$-$\CDL_0$ equivalence]
  \label{theorem:CDL-equivalence}
  For each $\varphi\in\Lang_\CDL$:
  \[
  \vdash_\CDL\varphi \quad\text{iff}\quad 
  \vdash_{\CDL_0}\varphi \enspace.
  \]
\end{theorem}
\begin{proof}
  See the appendix.
\end{proof}

That $\CDL$ is sound and complete with respect to the class of
well-ordered plausibility follows by
Theorem~\ref{theorem:CDL-equivalence} and the results in Board's work
\cite{Boa04:GEB}. However, we provide a full proof of this in the
appendix because the details will be useful when we consider a
justified version of $\CDL$.

\begin{theorem}[$\CDL$ soundness and completeness; \cite{Boa04:GEB}]
  \label{theorem:CDL-completeness}
  For each $\varphi\in\Lang_\CDL$:
  \[
  \vdash_\CDL\varphi \quad\text{iff}\quad \models\varphi
  \enspace.
  \]
\end{theorem}
\begin{proof}
  See the appendix.
\end{proof}




\section{AGM Belief Revision}

The most influential theory of belief change is due to Alchourr\'on,
G\"ardenfors, and Makinson \cite{AGM}. Their theory, commonly called
to the ``AGM theory,'' takes the view that an agent's belief state (or
``database'') is represented by a deductively closed set of sentences
$T$ called a ``belief set.''  The agent is understood to believe
exactly those sentences in her belief set $T$, and various operators
on $T$ are used to describe various kinds of changes in her belief
state.  Of particular interest is the \emph{revision} operator, now
often denoted using the symbol ``$\,*\,$''. This operator takes new
information in the form of a sentence $\psi$ and produces another
belief set $T*\psi$ that contains $\psi$. Intuitively, the revision
operation assumes that the incoming information $\psi$ is completely
trustworthy and so it should be incorporated into the
database. However, simply adding $\psi$ and taking the deductive
closure, forming the \emph{expansion}
\begin{equation}
  T+\psi \coloneqq \Cn(T\cup\{\psi\})
  \label{eq:expansion}
\end{equation}
using an assumed consequence-closure operator $\Cn(-)$ underlying the
setting, might lead to an inconsistent belief set. By this it is meant
that $T+\psi$ might be logically inconsistent according to the logic
governing $\Cn(-)$. Therefore, we cannot simply equate revision with
expansion but must do something more clever so that the revised belief
set $T*\psi$ not only contains $\psi$ but is also consistent whenever
$\psi$ is consistent.

Instead of providing an exact procedure for computing revision, the
AGM approach is ``postulate based'': a number of axiomatic postulates
are provided, some intuitive justification is given as to why a
revision operator should satisfy each of the postulates, and any
operation on belief sets that satisfies all of the postulates is said
to be an \emph{AGM revision operator}. So in principle, there are many
revision operators, and each is to be studied from an axiomatic point
of view using the AGM revision postulates.

Following the exposition of AGM theory from \cite{Ove11:SEP} (but with
some minor modifications), we begin with a set $\Prop$ of
propositional letters (usually countable).  The set $\Lang_\CPL$ of
formulas (the ``propositional formulas'' or ``formulas of Classical
Propositional Logic'') consists of those expressions that can be built
up from the propositional letters and the Boolean constants $\bot$
(falsehood) and $\top$ (truth) using the usual Boolean connectives. A
deductive theory is assumed, and this theory is specified in terms of
a Tarskian consequence operator: for any set $S$ of formulas, $\Cn(S)$
is the set of logical consequences of $S$. It is assumed that $\Cn(-)$
satisfies the following conditions:
\begin{itemize}
\item Inclusion: $S\subseteq\Cn(S)$,

\item Monotony (also sometimes called ``Monotonicity''):
  $S\subseteq S'$ implies $\Cn(S)\subseteq\Cn(S')$,

\item Iteration: $\Cn(S)=\Cn(\Cn(S))$, and

\item Supraclassicality: $\Cn(S)$ contains each classical tautology in
  $\Lang_\CPL$.
\end{itemize}
It is usually assumed that $\Cn(-)$ also satisfies the following
conditions:
\begin{itemize}
\item Deductive Consistency: $\bot\notin\Cn(\emptyset)$; and

\item Compactness: $\varphi\in\Cn(S)$ iff there exists a finite
  $S'\subseteq S$ such that $\varphi\in\Cn(S')$.
\end{itemize}
The alternative notation $S\vdash\varphi$ is used to express
$\varphi\in\Cn(S)$.  A \emph{belief base} is a set of formulas in the
language, and a \emph{belief set} is a deductively closed belief base
(i.e., $\Cn(S)=S$). Note that the phrase ``belief state'' (a.k.a.,
``database'') is an intuitive notion meant to describe the agent's
situation with regard to her beliefs.  This intuitive notion is
formalized either by a belief base (not necessarily deductively
closed) or a belief set (necessarily deductively closed).  A belief
base $T'$ gives rise to a belief set $T$ by applying the consequence
operator: $T\coloneqq\Cn(T')$.

A revision operator is meant to take an existing belief set $T$ and
some incoming information $\psi$ and produce a new belief set $T*\psi$
that incorporates the incoming information $\psi$ (i.e.,
$\psi\in T*\psi$), is consistent whenever the incoming information
$\psi$ is consistent (i.e., $T*\psi\nvdash\bot$ if $\psi$ is
consistent), and is obtained from $T$ by way of a ``minimal change.''
The latter is an intuitive (and non-formalized) guiding principle that
is used to persuade the reader that certain proposed postulates are
desirable.  From a formal perspective, it can be safely ignored.

The AGM revision postulates are reproduced in Table~\ref{table:AGM}.
Postulates 1--6 are called the \emph{G\"ardenfors postulates} (or,
more elaborately, the ``basic G\"ardenfors postulates for revision'').
Postulates 7--8 are called the \emph{supplementary postulates}.

\begin{table}
  \begin{center}
    \textsc{Postulates of AGM Belief Revision}
    \\[.3em]
    \renewcommand{\arraystretch}{1.3}
    \begin{tabular}{ll}
      1. Closure: & 
                    $T*\psi=\Cn(T*\psi)$
      \\
      2. Success: &
                    $\psi\in T*\psi$
      \\
      3. Inclusion: &
                      $T* \psi\subseteq T+\psi$
      \\
      4. Vacuity: &
                    if $\lnot\psi\notin T$, then $T*\psi=T+\psi$
      \\
      5. Consistency: &
                        if $\lnot\psi\notin\Cn(\emptyset)$, then
                        $\bot\notin\Cn(T*\psi)$
      \\
      6. Extensionality: &
                           if $(\psi\leftrightarrow\psi')\in\Cn(\emptyset)$, 
                           then
                           $T*\psi=T*\psi'$
      \\
      7. Superexpansion: &
                           $T*(\psi\land\varphi)\subseteq 
                           (T*\psi)+\varphi$
      \\
      8. Subexpansion: &
                         if $\lnot \varphi\notin\Cn(T*\psi)$, then
                         $(T*\psi)+\varphi\subseteq 
                         T*(\psi\land\varphi)$
    \end{tabular}

    \bigskip\small Notes:
    \begin{tabular}[t]{l}
      $S+\psi\coloneqq\Cn(S\cup\{\psi\})$; \\
      $\Cn(-)$ satisfies Inclusion, Monotony, Iteration, and Supraclassicality; \\
      Vacuity may be called ``Preservation.''
    \end{tabular}
  \end{center}
  \caption{The AGM revision postulates (as presented in \cite{Ove11:SEP})}
  \label{table:AGM}
\end{table}

Grove \cite{Gro88:JPL} proposed a possible worlds modeling of the AGM
postulates. Modulo certain details we gloss over, his proposal
essentially amounts to this: represent the agent's belief set using
the minimal worlds of a well-ordered plausibility model and define
revision in terms of belief conditionalization.  

\begin{theorem}[AGM revision and Grove spheres; adapted from
  \cite{Gro88:JPL}]
  \label{theorem:AGM-wo}
  Let $\Cn(-)$ be the consequence function $\CPL(-)$ of Classical
  Propositional Logic over the language $\Lang_\CPL$.  For each
  well-ordered plausibility model $M$, each belief set
  $T\subseteq\Lang_\CPL$, and each propositional formula
  $\psi\in\Lang_\CPL$, define
  \begin{align*}
    M^\down
    &\;\;\coloneqq\;\;
      \{\varphi\in\Lang_\CPL\mid
      \min(W)\subseteq\sem{\varphi}_M\}\enspace,
    \\
    T*_M\psi
    &\;\;\coloneqq\;\;
      \{\varphi\in\Lang_\CPL\mid
      \min\sem{\psi}_M\subseteq\sem{\varphi}_M\}\enspace,
    \\
    \mathfrak{B}_\CPL
    &\;\;\coloneqq\;\;
      \{L\subseteq\Lang_\CPL\mid
      \CPL(L)=L\}\enspace.
  \end{align*}
  Note that $\mathfrak{B}_\CPL$ is the collection of all propositional
  belief sets.  For a propositional formula or set of propositional
  formulas $X\in\Lang_\CPL\cup\wp(\Lang_\CPL)$, to say $X$ is
  \emph{consistent} means $\bot\notin\CPL(X)$ (or equivalently, that
  $\CPL(X)\neq\Lang_\CPL$), and to say $X$ is \emph{inconsistent}
  means it is not consistent.  To say that a plausibility model $M$ is
  a \emph{system of spheres} means that $M$ is well-ordered and the
  function
  \[
  M^\down*_M(-):\Lang_\CPL\to\mathfrak{B}_\CPL
  \]
  mapping propositional formulas $\psi\in\Lang_\CPL$ to belief sets
  $M^\down*_M\psi\in\mathfrak{B}_\CPL$ satisfies the AGM revision
  postulates. To say that $M$ is a \emph{Grove system} for a belief
  set $T\in\mathfrak{B}_\CPL$ means that $M$ is a system of spheres
  and $M^\down=T$.
  \begin{enumerate}[\;\;(a)]
  \item \label{i:cbs-base} Each consistent belief set
    $T\in\mathfrak{B}_\CPL$ has a Grove system.

  \item \label{i:*=>AGM} Suppose for a function
    $*:\mathfrak{B}_\CPL\times\Lang_\CPL\to\mathfrak{B}_\CPL$ and each
    $(T,\psi)\in\mathfrak{B}_\CPL\times\Lang_\CPL$ we have:
    \begin{enumerate}[\;\;(i)]
    \item if $T$ is inconsistent, then $T*\psi=M_*^\down*_{M_*}\psi$ for
      some fixed system of spheres $M_*$; and

    \item if $T$ is consistent, then $T*\psi=T*_{M_T}\psi$ for some
      fixed Grove system $M_T$ for $T$.
    \end{enumerate}
    It follows that $*$ is an AGM revision operator. 
  \end{enumerate}
\end{theorem}
\begin{proof}
  See the appendix.
\end{proof}

This suggests that we may view $\CDL$ as a version of AGM belief
revision in which the revision process itself can be described in the
language \cite{Boa04:GEB,BalSme08:LOFT}.  In particular, for
propositional formulas $\varphi$ and $\psi$, the formula $B\varphi$,
which is our abbreviation for $B^\top\varphi$, says that the agent
believes $\varphi$ before the revision takes place; and the formula
$B^\psi\varphi$ says that the agent believes $\varphi$ after revision
by $\psi$.  So by restricting to propositional $\varphi$ and $\psi$,
we can use conditional belief formulas $B^\psi\varphi$ to describe a
version of the AGM revision process directly in the language of
$\Lang_\CDL$.  This leads us to the following overview of what the
axiomatic theory $\CDL$ (Table~\ref{table:CDL}) has to say about the
$\CDL$-based version of AGM revision.
\begin{itemize}
\item (CL) and (MP) indicate that we use a classical meta-theory.

\item The principles we have grouped together under the name ``modal
  reasoning''---(CL), (MP), (K), and (MN)---together correspond to AGM
  Closure and AGM Extensionality.  In addition, it follows by modal
  reasoning that our underlying ``consequence operator'' satisfies
  Inclusion, Monotony, Iteration, and Supraclassicality.

\item (Succ) corresponds to AGM Success.

\item Under the assumption of (Succ), scheme (KM) corresponds to
  consequence of AGM Consistency: inconsistency of revision by $\psi$
  implies inconsistency of $\psi$ and therefore of $\psi\land\varphi$,
  and so inconsistency of revision by $\psi\land\varphi$ follows via
  (Succ).

\item Under the assumption of (Succ), scheme (RM) corresponds to AGM
  Subexpansion.  And if $\varphi=\top$, then (Succ) and (RM) together
  correspond to AGM Vacuity.

\item (Inc) corresponds to AGM Subexpansion. And, if $\varphi=\top$,
  then (Inc) corresponds to AGM Inclusion.

\item (Comm) corresponds to a special case of AGM Extensionality
  (i.e., commutativity of conjunction).

\item (PI) and (NI) do not corresponds to principles in the AGM
  setting (belief sets are subset of $\Lang_\CPL$).

\item (WCon) corresponds to AGM Consistency.
\end{itemize}

\section{Justified Conditional Doxastic Logic}

Though $\CDL$ may be viewed as a version of AGM revision in which
Boolean combinations of revisions are expressible in the language, one
key aspect that is missing: the \emph{reasons} as to why revisions
result in a state in which certain formulas are believed.  This is a
deficit present also in AGM revision: while a revision by $\psi$ may
lead to a belief set that includes $\varphi$, it is not immediately
clear why it is that $\varphi$ ought to obtain.  What is missing is
some language-describable reason that explains how it is $\varphi$
came about as a result of the revision.  It shall be our task in this
section to study how we might ``fill in'' these reasons in a theory
based on $\CDL$.  Our approach follows the general paradigm of
Justification Logic \cite{ArtFit12:SEP}, where modal operators are
replaced by syntactically structured objects called
\emph{terms}. Terms are meant to suggest ``reasons'' in the sense that
the syntactic structure of a term accords with certain derivational
principles in the underlying logic.  Our goal is to adapt this
methodology to $\CDL$. In particular, for a term $t$, we introduce new
formulas
\[
t\colu\psi\varphi
\]
with the intended meaning that, whenever the agent revises her belief
state by incorporating the formula $\psi$, then $t$ will be a reason
justifying her belief of $\varphi$ in the resulting belief state.
Intuitively, the formula $t\colu\psi\varphi$ tells us two things.
First, it tells us that the agent believes $\varphi$ conditional on
$\psi$, which was the information conveyed in $\CDL$ by the formula
$B^\psi\varphi$.  Second, something new: the formula
$t\colu\psi\varphi$ tells us that the reason encoded by $t$ supports
$\varphi$. Taken together, $t\colu\psi\varphi$ tells us that the agent
has a reason-based belief of $\varphi$ after revising by $\psi$.

Terms will be built up using a simple grammar. At the base of this
grammar are the ``certificates,'' which are terms of the form
$c_\varphi$ for some formula $\varphi$.  Intuitively, whenever we have
one or more reasons in support of $\varphi$, the certificate
$c_\varphi$ picks out the ``best'' one. So whenever $\varphi$ has any
support at all, we are always guaranteed that $c_\varphi$ names a
particular reason in support of $\varphi$. As such, the theory we
eventually define will derive the principle
\[
t\colu\psi\varphi \to c_\varphi\colu\psi\varphi\enspace,
\]
which says that the certificate $c_\varphi$ supports $\varphi$ after a
revision by $\psi$ whenever there is some reason $t$ that support
$\varphi$ after the same revision. In essence, certificates allow us
to ``forget'' the details of a complex argument in support of some
assertion, remembering only that we at some point found such an
argument.

Though a certificate $c_\varphi$ must support the formula $\varphi$ it
certifies, as per the above-mentioned derivable principle, we do not
prevent $c_\varphi$ from supporting other formulas as well.  For
example, it is consistent with the theory we will develop for us to
have $c_\varphi\colu\psi\chi$ for some $\chi\neq\varphi$.  As such,
though we require certificates to provide ``best evidence'' for the
formulas they certify, we do not require that this be the only
evidence that they provide.

Other terms are formed from certificates using one of two operators.
The first is the \emph{Application operator} ``$\,\cdot\,$'' from
Justification Logic. This operator is used to indicate that terms are
to be combined using a single step of the rule of \emph{Modus Ponens}.
In particular, the logic we will develop will derive the following
principle:
\[
t\colu\psi(\varphi_1\to\varphi_2)\to(s\colu\psi\varphi_1\to(t\cdot
s)\colu\psi\varphi_2)\enspace.
\]
This says that if $t$ supports an implication after the revision by
$\psi$ and $s$ supports the antecedent after the same revision, then
the combination $t\cdot s$ supports the consequence after that
revision. This is the reason-explicit version of the principle (K) of
$\CDL$.  The difference is that the present version tells us something
about how we obtained the consequent: the form of $t\cdot s$, with $t$
to the left and $s$ to the right, indicates that $t$ supports an
implication and $s$ supports the antecedent, and hence we were able to
derive the consequent via one step of Modus Ponens in virtue of the
fact that we use a single instance of the Application operator
``$\,\cdot\,$'' to combine $t$ with $s$ to form $t\cdot s$.

The second term-combining operator we introduce is the \emph{Sum
  operator} ``$+$'' from Justification Logic. This operator allows us
to combine to reasons in a way that preserves support.  In particular,
the logic we develop will derive the following principle:
\[
(t\colu\psi\varphi\lor s\colu\psi\varphi)\to
(t+s)\colu\psi\varphi\enspace.
\]
This says that $t+s$ supports $\varphi$ whenever at least one of $t$
or $s$ does so.  As such, the sum $t+s$ combines the supported
statements of $t$ and of $s$ without performing logical inference.

Formulas of the language will be built up from the language of
Classical Propositional Logic (based on propositional letters, the
constant $\bot$ for falsehood, and material implication) by adding
formulas of the form $t\colu\psi\varphi$, where $t$ is a term and
$\varphi$ and $\psi$ are other formulas.  Intuitively,
$t\colu\psi\varphi$ says that $t$ supports the agent's belief of
$\varphi$ after she revises her beliefs by incorporating $\psi$.

The theory we shall define is called \emph{Justified Conditional
  Doxastic Logic} or $\JCDL$. In addition to the derivable principles
mentioned above, $\JCDL$ has a number of additional principles that
make explicit the revision-based principles of $\CDL$.  In particular,
we will see that every $\CDL$-principle gives rise to a corresponding
$\JCDL$-principle, and the other way around as well.  We explain this
in more detail after the main definitions are in place.

\subsection{Language and axiomatics}

\begin{definition}[$\Lang_{\JCDL}$]
  \label{definition:lang-JCDL}
  Let $\Prop$ be a fixed set of propositional letters. The language of
  \emph{Justified Conditional Doxastic Logic} consists of the set
  $\Lang_{\JCDL}$ of formulas $\varphi$ and set $\Term_\JCDL$ of terms
  $t$ formed by the following grammar:
  \[
  \begin{array}{lll@{\qquad}l}
    \varphi & \Coloneqq & 
    p \mid \bot \mid (\varphi\to\varphi)
    \mid t\colu\varphi\varphi
    &
    \text{\small
      $p\in\Prop$
    }
    \\[.8em]
    t & \Coloneqq &
    c_\varphi \mid
    (t\cdot t) \mid
    (t+t)
  \end{array}
  \]
  Standard abbreviations for Boolean constants and connectives are
  used, and parentheses are dropped when doing so will cause no
  confusion. We adopt the following key abbreviation:
  \begin{center}
    $\dot B^\psi\varphi$ \quad{}denotes\quad
    $c_\varphi\colu\psi\varphi$\enspace.
  \end{center}
  We may write $c(\varphi)$ as an abbreviation for $c_\varphi$ when
  convenient.  Also, we let $t\col\varphi$ abbreviate
  $t\colu\emptyset\varphi$, and we let $\dot B\varphi$ abbreviate
  $\dot B^\emptyset\varphi$.
\end{definition}

Roughly speaking, it will be useful to think of $\dot B^\psi\varphi$
as the $\JCDL$-analog of the $\CDL$-expression $B^\psi\varphi$. We
will see that every $\CDL$-principle gives rise to a $\JCDL$-principle
obtained by replacing modal operators ``$\,B^\psi\,$'' by reason-based
operators ``$\,t\colu\psi\,$''. Accordingly, if $B^\psi\varphi$ is
$\CDL$-derivable, then a corresponding $t\colu{\psi'}\varphi'$ will be
derivable, where $\psi'$ corresponds to $\psi$ and $\varphi'$
corresponds to $\varphi$). So using the certificate $c_{\varphi'}$ for
$\varphi'$, it will follow that $\dot B^{\psi'}\varphi'$ is derivable
as well. It is in this sense that $\dot B^{\psi'}\varphi'$ may be
thought of as the $\JCDL$-analog of the $\CDL$-principle
$B^\psi\varphi$.

\begin{definition}[$\JCDL$ theory]
  The theory $\JCDL$ is defined in Table~\ref{table:JCDL}.
\end{definition}

\begin{table}
  \begin{center}
    \textsc{Axiom Schemes}\\[.4em]
    \renewcommand{\arraystretch}{1.4}
    \begin{tabular}[t]{rl}
      (CL) &
             Schemes for Classical Propositional Logic
      \\
      (eCert) &
                $t\colu\psi\varphi\to\dot B^\psi\varphi$
      \\
      (eK) &
             $t\colu\psi(\varphi_1\imp\varphi_2)\imp
             (s\colu\psi\varphi_1\imp(t\cdot s)\pcolu\psi\varphi_2)$
      \\
      (eSum) &
               $(t\colu\psi\varphi\lor s\colu\psi\varphi)\imp
               (t+s)\colu\psi\varphi$
      \\
      (eSucc) &
                $\dot B^\psi\psi$
      \\
      (eKM) &
              $t\colu\psi\bot\to t\colu{\psi\land\varphi}\bot$
      \\
      (eRM) &
              $\lnot\dot B^\psi\lnot\varphi\to
              (t\colu\psi\chi\to t\colu{\psi\land\varphi}\chi)$
      \\
      (eInc) &
               $t\colu{\psi\land\varphi}\chi\to
               \dot B^\psi(\varphi\to\chi)$
      \\
      (eComm) &
                $t\colu{\psi\land\varphi}\chi\to
                t\colu{\varphi\land\psi}\chi$
      \\
      (ePI) &
              $t\colu\psi\chi\to\dot B^\varphi(t\colu\psi\chi)$
      \\
      (eNI) &
              $\lnot t\colu\psi\chi\to\dot B^\varphi(\lnot t\colu\psi\chi)$
      \\
      (eWCon) &
                $t\colu\psi\bot\to\lnot\psi$
      \\
      (eA) &
             $t\colu\psi\varphi\to(\dot B^\chi\varphi\to t\colu\chi\varphi)$
    \end{tabular}
    \\[2em]
    \textsc{Rules}
    \\[.5em]
    \AXC{$\varphi\to\psi$}
    \AXC{$\varphi$}
    \RLs{(MP)}
    \BIC{$\psi$}
    \DP
    \qquad
    \AXC{$\varphi$}
    \RLs{(eMN)}
    \UIC{$\dot B^\psi\varphi$}
    \DP
    \\[1em]
    Note:
    $\dot B^\psi\varphi$ abbreviates
    $c_\varphi\colu\psi\varphi$.
  \end{center}
  \caption{The theory $\JCDL$}
  \label{table:JCDL}
\end{table}

(CL) and (MP) tell us that $\JCDL$ is an extension of Classical
Propositional Logic.

(eK) is our reason-explicit analog of the $\CDL$-scheme (K); it tells
us that reason support is closed under the rule of Modus Ponens using
the Application operator ``$\,\cdot\,$''.  (eSucc) is our
reason-explicit analog of $\CDL$ (Succ); it tells us that certificates
are used to certify the success of belief revisions.  (eKM) is our
reason-explicit version of (KM); it tells us that a reason supporting
a contradiction conditional on some $\psi$ continues to do so no
matter what additional information we conjunctively add to the
conditional. 

(eRM) is our reason-explicit version of (RM).  The antecedent
$\lnot\dot B^\psi\lnot\varphi$ is just
$lnot c_{\lnot\varphi}\colu\psi\lnot\varphi$.  Since the certificate
$c_{\lnot\varphi}$ will always stand in for some argument relevant to
$\lnot\varphi$, the antecedent $\lnot\dot B^\psi\lnot\varphi$ tells us
that $\varphi$ is consistent with the belief state obtained after
revision by $\psi$.  And (eRM) tells us that if this is so and $t$ is
a reason to believe $\chi$ after revision by $\psi$, then it follows
that $t$ is still a reason to believe $\chi$ even after we revise by
the conjunction $\psi\land\varphi$. Notice that the reason $t$ for
$\chi$ does not change; it is only the revision formula itself that
changes.

(eInc) is our reason-explicit version of (Inc); it tells us that that
a belief of $\varphi\to\psi$ is certified after revising by $\psi$ so
long as there is a reason $t$ for believing $\psi$ after revising by
$\psi\land\varphi$.  (eComm) tells us that reasons are invariant to
the order of conjuncts in revisions; this is the reason-explicit
analog of (Comm).  (ePI) and (eNI) are the reason-explicit versions of
(PI) and (NI), respectively; these tell us that all support and
negated support statements are certified. (eWCon) is the
reason-explicit analog of (WCon); it tells us that if $t$ is a reason
supporting a belief in $\bot$ after revision by $\psi$, then
$\lnot\psi$ must have been true.

(eA) is particular to $\JCDL$. This scheme tells us that if $t$
supports a belief in $\varphi$ conditional on $\psi$ and $\varphi$ is
also certified conditional on some other $\chi$, then $t$ must itself
support $\varphi$ conditional on $\chi$ as well.  This tells us that
if a reason supports a belief of $\varphi$ after some revision, it
does so no matter the particulars of the revision. Said another way,
reason support depends only on the statement supported (and not on the
revision).

(eMN) says that every derivable formula is ``certified'' (i.e.,
supported by its certificate). This corresponds to the
$\CDL$-principle (MN). But while all derivable principles are
certified, intuitively such certification omits a great deal of
information; in particular, it is not clear from which axioms a given
principle follows and how it is that it follows by way of the rules of
the theory.  Toward this end, it will be useful to consider a
restriction of (eMN) in which we apply this rule only to axioms,
possibly multiple times in a row. An axiom to which we apply (eMN)
zero or more times in a row will be called a ``possibly necessitated
axiom.'' We will see that if we remove (eMN) from the theory, then all
$\JCDL$-theorems can be derived from the possibly necessitated axioms
using (MP) as the only rule. The trick to this will be to eliminate
from $\JCDL$-derivations uses of (eMN) that are ``troublesome
necessitations'': these are the derivable consequences of (eMN) that
are not themselves possibly necessitated axioms. If we can show that
all such ``troublesome necessitations'' can be eliminated, then the
result follows.  This will be our task now.

The terms that can be formed from certificates of possibly
necessitated axioms using the Application operator ``$\,\cdot\,$''
will be called the ``logical terms.'' These terms play a special role:
every $\JCDL$-theorem $\varphi$ gives rise to a logical term $t$ that
supports $\varphi$ (i.e., $t\colu\psi\varphi$ is derivable for each
$\psi$).

\begin{definition}[Necessitations, logical terms]
  A \emph{necessitation} is a $\Lang_\JCDL$-formula of the form
  \begin{equation}
    \underbrace{\dot B^{\psi_n}\dot B^{\psi_{n-1}}
      \dot B^{\psi_{n-2}}\cdots\dot B^{\psi_1}}_{\text{zero or
        more of these}}\varphi
    \label{eq:N-axiom}
  \end{equation}
  for some integer $n\geq 0$. A \emph{possibly necessitated axiom} is
  a $\Lang_\JCDL$-formula of the form \eqref{eq:N-axiom} for which
  $n\geq 0$ and $\varphi$ is a $\JCDL$-axiom.  The set $\Term_\JCDL^L$
  of \emph{logical terms} is the smallest set that contains
  certificates $c_\varphi$ for each possibly necessitated axiom
  $\varphi$ and is closed under the term-forming operation
  $t,s\mapsto t\cdot s$.
\end{definition}

\begin{definition}[Troublesome necessitations, notation
  $\pi\vdash_\JCDL^n\varphi$]
  A $\JCDL$-\emph{derivation} is finite nonempty sequence of
  $\Lang_\JCDL$-formulas, each of which is either a $\JCDL$-axiom or
  follows by a $\JCDL$-rule from formulas occurring earlier in the
  sequence.  A \emph{line} of a $\JCDL$-derivation $\pi$ is an element
  of the sequence $\pi$.  A \emph{troublesome necessitation} is a line
  $\varphi$ of a $\JCDL$-derivation that neither is a possibly
  necessitated axiom nor follows from earlier lines by (MP).  Clearly,
  a troublesome necessitation must follow by applying (eMN) to an
  earlier line that is not itself a possibly necessitated axiom.  For
  $n\in\Nat$, we write $\pi\vdash_\JCDL^n\varphi$ to mean that $\pi$
  is a $\JCDL$-derivation that contains at most $n$ troublesome
  necessitations and whose last line is $\varphi$. For $n\in\Nat$, we
  write $\vdash_\JCDL^n\varphi$ to mean that there exists a
  $\JCDL$-derivation $\pi$ such that
  $\pi\vdash_\JCDL^n\varphi$. Obviously, $\vdash_\JCDL\varphi$ implies
  $\vdash_\JCDL^n\varphi$ for some $n\in\Nat$.
\end{definition}

The following lemma shows that every $\JCDL$-theorem can be derived
from the possibly necessitated axioms using (MP) as the only rule of
inference.

\begin{lemma}[Elimination of troublesome necessitations]
  \label{lemma:elim-t-necess}
  For each $\varphi\in\Lang_\JCDL$, we have:
  \begin{equation}
    \pi\vdash_\JCDL^n\varphi
    \quad\Rightarrow\quad
    \exists \pi_*\supseteq\pi,\;\;
    \pi_*\vdash_\JCDL^0\varphi\enspace.
  \end{equation}
\end{lemma}
\begin{proof}
  See the appendix.
\end{proof}

``Theorem Internalization'' is a property of Justification Logics
whereby every theorem $\varphi$ of the logic is witnessed by
supporting term. Since certificates trivially support $\JCDL$-theorems
by (eMN), the usual formulation of Theorem Internalization is
trivialized in our setting. However, we can prove a stronger variant:
every $\JCDL$-theorem $\varphi$ is witnessed by a \emph{logical}
supporting term. This stronger version tells us that all
$\JCDL$-theorems are witnessed by terms that refer only to possibly
necessitated axioms and combinations of these using rule (MP).

\begin{theorem}[$\JCDL$ Theorem Internalization]
  \label{theorem:internalization}
  For each $\psi\in\Lang_{\JCDL}$, we have:
  \begin{center}
    $\vdash_{\JCDL}\varphi$ \quad$\Rightarrow$\quad
    $\exists t\in\Term_\JCDL^L$,
    ${}\vdash_{\JCDL}t\colu\psi\varphi$\enspace.
  \end{center}
\end{theorem}
\begin{proof}
  See the appendix.
\end{proof}

\subsection{Relationship to \texorpdfstring{$\CDL$}{CDL}}

In our motivation of $\JCDL$, we have described the formula
$t\colu\psi\varphi$ as an analog of a corresponding
$\Lang_\CDL$-formula $B^\psi\varphi$. Up to this point, the idea was
mere intuition. We now make this intuition precise by defining two
mappings.  The first, called ``forgetful projection,'' maps
$\Lang_\JCDL$-formulas to $\Lang_\CDL$-formulas by replacing each
``$\,t\colu\psi\,$'' prefix by the prefix ``$\,B^\psi\,$''.  The
second, called ``trivial realization,'' maps $\Lang_\CDL$-formulas to
$\Lang_\JCDL$-formulas by replacing each prefix ``$\,B^\psi\,$'' by
the prefix ``$\,\dot B^\psi\,$''. We will see that these operations
preserve derivability of schemes.

\begin{definition}[The forgetful projection]
  \label{definition:projection}
  The \emph{forgetful projection} is the function
  \[
  (-)^\circ:\Lang_{\JCDL}\to\Lang_{\CDL}
  \]
  defined by:
  \begin{align*}
    q^\circ 
    &\coloneqq 
      q \quad\text{for } q\in\Prop\cup\{\bot\} 
    \\
    (\varphi\to\psi)^\circ 
    &\coloneqq 
      \varphi^\circ\to\psi^\circ 
    \\
    (t\colu\psi\varphi)^\circ 
    &\coloneqq
      B^{\psi^\circ}\varphi^\circ
  \end{align*}
  The \emph{forgetful projection} of $\psi\in\Lang_\JCDL$ is the
  $\psi^\circ\in\Lang_\CDL$. Extend the function $(-)^\circ$ to sets
  $\Gamma\subseteq\Lang_\JCDL$ by
  $\Gamma^\circ\coloneqq\{\varphi^\circ\mid\varphi\in\Gamma\}$.  We
  further extend the function $(-)^\circ$ to schemes.  In particular,
  let $\calS$ be a fixed set of schematic variables (i.e.,
  ``metavariables'' or placeholders for formulas) that includes all
  schematic variables used in this paper and that has cardinality
  $\min\{|\Prop|,\omega\}$.  For each of our languages
  $\Lang\in\{\Lang_\CDL,\Lang_\JCDL\}$, let $\Lang(\calS)$ be the set
  of formula schemes that can be formed using the formula formation
  grammar of $\Lang$ but with schematic variables in $\calS$ used in
  place of propositional letters in $\Prop$.  Define
  $\Term_\JCDL(\calS)$ similarly. Using $\Phi$ and $\Psi$ as
  metavariables ranging over members of $\Lang_\JCDL(\calS)$ and $T$
  as a metavariable ranging over members of $\Term_\JCDL(\calS)$, let
  \begin{align*}
    X^\circ 
    &\coloneqq 
      X \quad\text{for } X\in\calS\cup\{\bot\} 
    \\
    (\Phi\to\Psi)^\circ 
    &\coloneqq 
      \Phi^\circ\to\Psi^\circ 
    \\
    (T\colu\Psi\Phi)^\circ 
    &\coloneqq
      B^{\Psi^\circ}\Phi^\circ
  \end{align*}
  Extend the function $(-)^\circ$ to sets
  $\Gamma\subseteq\Lang_\JCDL(\calS)$ by
  $\Gamma^\circ\coloneqq\{\Phi^\circ\mid\Phi\in\Gamma\}$.
\end{definition}

\begin{definition}[Realizations and the trivial realization]
  \label{definition:realization}
  A \emph{realization} of a formula $\varphi\in\Lang_{\CDL}$ is a
  formula $\psi\in\Lang_{\JCDL}$ for which $\psi^\circ=\varphi$ and
  $\vdash_\JCDL\psi$.  The \emph{trivial realization} is the function
  $(-)^t:\Lang_{\CDL}\to\Lang_{\JCDL}$ defined by:
  \begin{align*}
    q^t 
    &\coloneqq 
      q \quad\text{ for } q\in\Prop\cup\{\bot\} 
    \\
    (\varphi\to\psi)^t 
    &\coloneqq 
      \varphi^t\to\psi^t 
    \\
    (B^\psi\varphi)^t 
    &\coloneqq
      \dot B^{\psi^t}\varphi^t
  \end{align*}
  We extend the function $(-)^t$ to sets $\Gamma\subseteq\Lang_\CDL$
  by $\Gamma^t\coloneqq\{\varphi^t\mid\varphi\in\Gamma\}$.  As in
  Definition~\ref{definition:projection} and using the notation from
  that definition, we extend the function $(-)^t$ to schemes:
    \begin{align*}
    X^t 
    &\coloneqq 
      X \quad\text{ for } X\in\calS\cup\{\bot\} 
    \\
    (\Phi\to\Psi)^t 
    &\coloneqq 
      \Phi^t\to\Psi^t 
    \\
    (B^\Psi\Phi)^t 
    &\coloneqq
      \dot B^{\Psi^t}\Phi^t
  \end{align*}
  Finally, we apply $(-)^t$ to sets $\Gamma\subseteq\Lang_\CDL(\calS)$
  by defining $\Gamma^t\coloneqq\{\Phi^t\mid \Phi\in\Gamma\}$.  For
  some object $Z$ in the domain of $(-)^t$ we say that $Z^t$ is the
  \emph{trivial realization} of $Z$.
\end{definition}

The trivial realization of a formula or scheme of $\CDL$ is obtained
by replacing each ``$B\,$'' with ``$\dot B\,$''.  Note that while the word
``realization'' in the phrase ``trivial realization'' suggests that
the trivial realization of a formula or scheme is indeed a realization
(i.e., we obtain something derivable in $\JCDL$), this does not come
automatically (i.e., by definition) because the trivial realization is
a mere syntactic translation and so it must be proved that this
translation satisfies the requisite property before the conclusion can
be drawn.  However, the following theorem guarantees that the trivial
realization is indeed a realization. The theorem also tells us that
the forgetful projection maps $\JCDL$-theorems to $\CDL$-theorems.

\begin{theorem}[Projection and realization for schemes]
  \label{theorem:realization}
  We use the notation from Definitions~\ref{definition:projection} and
  \ref{definition:realization}.  For each
  $\Gamma\cup\{\varphi\}\subseteq\Lang_\JCDL(\calS)$ and each
  $\Delta\cup\{\psi\}\subseteq\Lang_\CDL(\calS)$, we have:
  \begin{enumerate}
  \item \label{i:projection} $\Gamma\vdash_\JCDL\varphi$ implies
    $\Gamma^\circ\vdash_{\CDL}\varphi^\circ$, and
    
  \item \label{i:t-realization} $\Delta\vdash_{\CDL}\psi$ implies
    $\Delta^t\vdash_\JCDL\psi^t$.
  \end{enumerate}
\end{theorem}
\begin{proof}
  See the appendix.
\end{proof}

Theorem~\ref{theorem:realization} tells us that $\JCDL$ really is an
explicit analog of $\CDL$: every $\JCDL$-derivable statement gives
rise to a $\CDL$-derivable statement (by forgetful projection), and
every $\CDL$-derivable statement gives rise to a $\JCDL$-derivable
statement (by trivial realization).  This link makes precise our
intuition that the $\Lang_\CDL$-formula $B^\psi\varphi$ should
correspond to the $\Lang_\JCDL$-formula $\dot B^\psi\varphi$.

Though Theorem~\ref{theorem:realization} is stated with respect to
schemes, we have the analogous result for formulas as well.

\begin{theorem}[Projection and realization for formulas]
  \label{theorem:realization-formulas}
  We use the notation from Definitions~\ref{definition:projection} and
  \ref{definition:realization}.  For each
  $\Gamma\cup\{\varphi\}\subseteq\Lang_\JCDL$ and each
  $\Delta\cup\{\psi\}\subseteq\Lang_\CDL$, we have:
  \begin{enumerate}
  \item \label{i:projection-form} $\Gamma\vdash_\JCDL\varphi$ implies
    $\Gamma^\circ\vdash_{\CDL}\varphi^\circ$, and
    
  \item \label{i:t-realization-form} $\Delta\vdash_{\CDL}\psi$ implies
    $\Delta^t\vdash_\JCDL\psi^t$.
  \end{enumerate}
\end{theorem}
\begin{proof}
  Replace each distinct propositional letter with a distinct schematic
  variable, apply Theorem~\ref{theorem:realization}, and take the
  instances of the resulting derivable schemes obtained by
  substituting the original propositional variables back into their
  corresponding positions.
\end{proof}

Using Theorem~\ref{theorem:realization}, the $\CDL$-principles from
Theorem~\ref{theorem:CDL-theorems} translate into $\JCDL$-principles.

\begin{theorem}[$\JCDL$-theorems]
  \label{theorem:JCDL-theorems}
  The following schemes of (eCut), \emph{Explicit Cautious
    Monotonicity} (eCM), (eTaut), (eAnd), (eOr), \emph{Explicit
    Positive Reduction} (ePR), and \emph{Explicit Negative Reduction}
  (eNR) are all derivable in $\JCDL$:
  \begin{align*}
    \text{(eCut)} 
    &\quad
      \dot B^\psi\varphi\to(\dot B^{\psi\land\varphi}\chi\to \dot B^\psi\chi)
    \\
    \text{(eCM)} 
    &\quad
      \dot B^\psi\varphi\to(\dot B^\psi\chi\to \dot B^{\psi\land\varphi}\chi)
    \\
    \text{(eTaut)}
    &\quad
      \dot B\varphi \leftrightarrow \dot B^\top\varphi
    \\
    \text{(eAnd)}
    &\quad
      \dot B^\psi\varphi_1\to
      (\dot B^\psi\varphi_2\to \dot B^\psi(\varphi_1\land\varphi_2))
    \\
    \text{(eOr)}
    &\quad
      \dot B^{\psi_1}\varphi\to
      (\dot B^{\psi_2}\varphi\to \dot B^{\psi_1\lor\psi_2}\varphi)
    \\
    \text{(ePR)}
    &\quad
      \dot B^\varphi \dot B^\psi\chi
      \leftrightarrow
      ( \dot B^\psi\bot\lor \dot B^\psi\chi)
    \\
    \text{(eNR)}
    &\quad
      \dot B^\varphi\lnot \dot B^\psi\chi
      \leftrightarrow
      ( \dot B^\varphi\bot\lor\lnot \dot B^\psi\chi)
  \end{align*}
  Also, the following rules of \emph{Explicit (Left) Logical
    Equivalence} (eLE), \emph{Explicit Right Weakening} (eRW), and
  \emph{Explicit Supraclassicality} (eSC) are all derivable in $\JCDL$:
  \begin{center}
    \AXC{$\psi\leftrightarrow\psi'$}
    \RLs{(eLE)}
    \UIC{$\dot B^{\psi}\chi\leftrightarrow \dot B^{\psi'}\chi$} 
    \DP  
    \qquad
    \AXC{$\chi\to\chi'$}
    \RLs{(eRW)}
    \UIC{$\dot B^\psi\chi\to \dot B^\psi\chi'$} 
    \DP  
    \qquad
    \AXC{$\psi\to\chi$}
    \RLs{(eSC)}
    \UIC{$\dot B^\psi\chi$}
    \DP
  \end{center}
\end{theorem}
\begin{proof}
  Apply Theorems~\ref{theorem:CDL-theorems} and
  \ref{theorem:realization}\eqref{i:t-realization}.
\end{proof}

\subsection{Semantics}

One of the main possible worlds semantics for Justification Logic is
the semantics due to Fitting (see \cite{ArtFit12:SEP} for details).
Here we adapt the traditional Fitting semantics for use in our
language $\Lang_\JCDL$.

\begin{definition}[Fitting models]
  \label{definition:F-model}
  A \emph{Fitting model} is a structure $M={(W,\leq,V,A)}$ for which
  $(W,\leq,V)$ is a locally well-ordered plausibility model and $A$ is
  an \emph{admissibility function\/}: a function
  \[
  A:(\Term_\JCDL\times\Lang_\JCDL)\to\wp(W)
  \]
  that maps each term-formula pair $(t,\varphi)$ to a set
  $A(t,\varphi)\subseteq W$ of worlds subject to the following
  restrictions:
  \begin{itemize}
  \item Certification: $A(c_\varphi,\varphi)=W$,

    which says formulas are certified by their certificates;

  \item Application:
    $A(t,\varphi_1\to\varphi_2)\cap A(s,\varphi_1) \subseteq
    A(t\cdot s,\varphi_2)$,

    which says the Application operator encodes instances of (MP);

  \item Sum:
    $A(t,\varphi)\cup A(s,\varphi)\subseteq
    A(t+s,\varphi)$,

    which says the Sum operator encodes support aggregation without
    logical inference; and

  \item Admissibility Indefeasibility: if $x\in A(t,\varphi)$ and
    $y\in\cc(x)$, then $y\in A(t,\varphi)$,

    which says admissibility is constant within each connected
    component. Using the notion of ``knowledge'' from
    Remark~\ref{remark:knowledge}, this tells us that the agent knows
    her admissibility function.
  \end{itemize}
  The notion of \emph{pointed Fitting model} is similar to the
  corresponding definition found in
  Definition~\ref{definition:smooth}; we also apply the terminology
  from that definition to Fitting models in the obvious way.
\end{definition}

If $A$ is an admissibility function, then $w\in A(t,\varphi)$ says
that, from the perspective of world $w$, term $t$ has the proper
``syntactic shape'' to be a reason in support of $\varphi$. This does
not, however, guarantee that $t$ does indeed support $\varphi$.  For
this we shall require something more.

\begin{definition}[$\Lang_\JCDL$-truth]
  \label{definition:JCDL-truth}
  Let $M=(W,\leq,V,A)$ be a Fitting model.  We extend the binary
  satisfaction relation $\models$ from
  Definition~\ref{definition:CDL-truth} to one between pointed Fitting
  models $(M,w)$ (written without surrounding parentheses) and
  $\Lang_\JCDL$-formulas and we extend the function $\sem{-}$ from
  Definition~\ref{definition:CDL-truth} to include a function
  $\sem{-}:\Lang_\JCDL\to\wp(W)$ as follows.
  \begin{itemize}
  \item $\sem{\varphi}_M\coloneqq\{v\in W\mid
    M,v\models\varphi\}$. The subscript $M$ may be suppressed.

  \item $M,w\not\models\bot$.
    
  \item $M,w\models p$ iff $p\in V(w)$ for $p\in\Prop$.

  \item $M,w\models\varphi\imp\psi$ iff $M,w\not\models\varphi$ or
    $M,w\models\psi$.

  \item $M,w\models t\colu\psi\varphi$ iff $w\in A(t,\varphi)$
    and
    \begin{equation}
      \forall x\in\cc(w):\quad
      x^\down\cap\sem{\psi}=\emptyset
      \quad\text{or}\quad
      \exists y\in x^\down\cap\sem{\psi}:
      y^\down\cap\sem{\psi}\subseteq\sem{\varphi}
      \enspace.
      \label{eq:JCDL-truth}
    \end{equation}
  \end{itemize}
\end{definition}

So to have $t\colu\psi\varphi$ true at a world $w$, we must have two
things.  First, from the perspective of $w$, term $t$ must have the
correct ``syntactic shape'' for an argument in support of $\varphi$;
that is, we must have $w\in A(t,\varphi)$.  Second, we must satisfy
the condition \eqref{eq:JCDL-truth}, which is the same condition we
had for truth of a $\Lang_\CDL$-formula $B^\psi\varphi$.  So, taken
together, to have $t\colu\psi\varphi$ true at world $w$ means that $t$
has the ``shape'' of an argument for $\varphi$ and the agent believes
$\varphi$ after revising her beliefs by $\psi$. The following theorem
states this precisely.

\begin{theorem}[$\JCDL$-truth in terms of belief formulas]
  \label{theorem:JCDL-truth-B}
  Let $M=(W,\leq,V,A)$ be a Fitting model.
  \begin{align*}
    M,w\models t\colu\psi\varphi 
    &\qquad\text{iff}\qquad 
      w\in A(t,\varphi) 
      \text{ and } 
      M,w\models\dot B^\psi\varphi
      \enspace.
  \end{align*}
\end{theorem}
\begin{proof}
  Left to right (``only if''): assume $M,w\models t\colu\psi\varphi$.
  By Definition~\ref{definition:JCDL-truth}, we have
  $w\in A(t,\varphi)$ and \eqref{eq:JCDL-truth}.  Since we have
  $w\in A(c_\varphi,\varphi)$ by the Certification property of
  admissibility functions and we have
  $\dot B^\psi\varphi=c_\varphi\colu\psi\varphi$ by definition, it
  follows from by \eqref{eq:JCDL-truth} and
  $w\in A(c_\varphi,\varphi)$ by
  Definition~\ref{definition:JCDL-truth} that
  $M,w\models\dot B^\psi\varphi$.

  Right to left (``if''): assume $w\in A(t,\varphi)$ and
  $M,w\models\dot B^\psi\varphi$.  Applying
  Definition~\ref{definition:JCDL-truth} and the definition
  $\dot B^\psi\varphi=c_\varphi\colu\psi\varphi$, it follows that
  \eqref{eq:JCDL-truth}. So since we have $w\in A(t,\varphi)$ and
  \eqref{eq:JCDL-truth}, it follows by
  Definition~\ref{definition:JCDL-truth} that
  $M,w\models t\colu\psi\varphi$.
\end{proof}

In well-founded Fitting models, belief in $\varphi$ conditional on
$\psi$ is equivalent to having $\varphi$ true at the most plausible
$\psi$-worlds that are within the connected component of the actual
world.

\begin{theorem}[$\JCDL$-truth in well-founded models]
  \label{theorem:JCDL-wf}\label{theorem:JCDL-wo}
  Let $M=(W,\leq,V,A)$ be a Fitting model.
  \begin{enumerate}[\;\;(a)]
  \item \label{i:JCDL-wf} If $M$ is well-founded:
    $M,w\models t\colu\psi\varphi$ $\;\;\Leftrightarrow\;\;$
    $w\in A(t,\varphi)$ and
    $\min\sem{\psi}\cap\cc(w)\subseteq\sem{\varphi}$.

  \item \label{i:JCDL-wo} If $M$ is well-ordered:
    $M,w\models t\colu\psi\varphi$ $\;\;\Leftrightarrow\;\;$
    $w\in A(t,\varphi)$ and $\min\sem{\psi}\subseteq\sem{\varphi}$.
  \end{enumerate}
\end{theorem}
\begin{proof}
  \eqref{i:JCDL-wf}, left to right (``only if''): assume
  $M,w\models t\colu\psi\varphi$.  By
  Definition~\ref{definition:JCDL-truth}, we have $w\in A(t,\varphi)$
  and \eqref{eq:JCDL-truth}.  Use the argument in the proof of
  Theorem~\ref{theorem:CDL-wf} to conclude that
  $\min\sem{\psi}\cap\cc(w)\subseteq\sem{\varphi}$.
  \eqref{i:JCDL-wf}, right to left (``if''): assume
  $w\in A(t,\varphi)$ and
  $\min\sem{\psi}\cap\cc(w)\subseteq\sem{\varphi}$.  Use the argument
  in the proof of Theorem~\ref{theorem:CDL-wf} to conclude that
  \eqref{eq:JCDL-truth}. Applying
  Definition~\ref{definition:JCDL-truth}, it follows that
  $M,w\models t\colu\psi\varphi$.

  \eqref{i:JCDL-wo}: $M$ is well-ordered, then $\cc(w)=W$ and $M$ is
  well-founded. Apply \eqref{i:JCDL-wf}.
\end{proof}

And so in well-founded Fitting models, we can see that formulas
$\dot B^\psi\varphi$ really do play the semantic analog of
$\Lang_\CDL$-belief formulas.

\begin{theorem}[$\JCDL$-truth in well-founded models in terms of
  belief formulas]
  \label{theorem:JCDL-wf-B}\label{theorem:JCDL-wo-B}
  Let $M$ be a Fitting model.
  \begin{enumerate}[\;\;(a)]
  \item \label{i:JCDL-wf-B} If $M$ is well-founded:
    $M,w\models \dot B^\psi\varphi$ $\;\;\Leftrightarrow\;\;$
    $\min\sem{\psi}\cap\cc(w)\subseteq\sem{\varphi}$.
    
  \item \label{i:JCDL-wo-B} If $M$ is well-ordered:
    $M,w\models \dot B^\psi\varphi$ $\;\;\Leftrightarrow\;\;$
    $\min\sem{\psi}\subseteq\sem{\varphi}$.
  \end{enumerate}
\end{theorem}
\begin{proof}
  By the Certification (Definition~\ref{definition:F-model}) and
  Theorems~\ref{theorem:JCDL-truth-B} and \ref{theorem:JCDL-wf}.
\end{proof}

Similar to $\Lang_\CDL$, the intended semantic objects for
$\Lang_\JCDL$ are the well-ordered models of the appropriate type (in
this case Fitting models, as opposed to simple plausibility models).

\begin{definition}[$\Lang_\JCDL$-validity]
  \label{definition:JCDL-validity}
  To say that a $\Lang_\JCDL$-formula $\varphi$ is \emph{valid},
  written $\models\varphi$, means that $M\models\varphi$ for each
  well-ordered Fitting model $M$. Though we use the same symbol
  ``$\models$'' for $\Lang_\JCDL$-validity as we did for
  $\Lang_\CDL$-validity, it will be clear from context which notion is
  meant.
\end{definition}

And as for $\CDL$, locally well-ordered plausibility models would
suffice.

\begin{theorem}[$\Lang_\JCDL$-validity with respect to local well-orders]
  \label{theorem:JCDL-lwo}
  Let $\mathfrak{F}_L$ be the class of locally well-ordered Fitting
  models.  For each $\varphi\in\Lang_\JCDL$, we have:
  \[
  \models\varphi
  \qquad\text{iff}\qquad
  \forall M\in\mathfrak{F}_L,\;
  M\models\varphi\enspace.
  \]
\end{theorem}
\begin{proof}
  As in the proof of Theorem~\ref{theorem:CDL-lwo}.
\end{proof}

Soundness and completeness of $\JCDL$ with respect to its intended
semantics (i.e., well-ordered Fitting models) makes use of many
components of the proof of Theorem~\ref{theorem:CDL-completeness},
itself essentially due to \cite{Boa04:GEB}.

\begin{theorem}[$\JCDL$ soundness and completeness]
  \label{theorem:JCDL-determinacy}
  For each $\varphi\in\Lang_\JCDL$:
  \[
  \vdash_\JCDL\varphi \quad\text{iff}\quad \models\varphi
  \enspace.
  \]
\end{theorem}
\begin{proof}
  See the appendix.
\end{proof}



\section{Conclusion}

We saw earlier that $\CDL$ is a version of AGM revision in which
Boolean combinations of revisions are expressible in the language.
Since $\JCDL$ is a reason-explicit analog of $\CDL$ (as per
Theorem~\ref{theorem:realization}), we are led to the following
suggestion: $\JCDL$ is a version of AGM revision in which Boolean
combinations of \emph{reason-explicit revisions} are expressible in
the language. In essence, a formula $\varphi$ that is part of the
belief state after revision by $\psi$ may be witnessed by a specific
reason $t$ whose syntactic structure tracks the genesis of $\varphi$
stepwise from basic principles.  This suggests we think of $\JCDL$ as
a theory of \emph{revisable justified belief}.  It would be
interesting to see if there is some explicit version of the AGM
revision principles that matches up with $\JCDL$ in the way that
standard AGM matches up with $\CDL$.  However, we leave this issue for
future work.



\appendix

\section{Technical results}

\subsection{Results for plausibility models}

\begin{proof}[Proof of Theorem~\ref{theorem:wf-smooth}]
  Item~\ref{i:finite} is obvious.  Item~\ref{i:lwo} is a well-known
  result from order theory, but we reprove it anyway for completeness
  purposes.  So assume $M$ is locally well-ordered and
  $\emptyset\neq S=\cc(w)\subseteq W$ for some $w\in W$. Define
  \[
  S'\coloneqq\{x\in S\mid\forall y\in S:x\leq y\}\enspace.
  \]
  We wish to prove that $\min S=S'$.  Proceeding, take $x\in\min S$.
  If $y\in S$ as well, then it follows from $x\in\min S$ by the
  definition of $\min S$ that $y\not<x$, from which we obtain
  $x\leq y$ because $S$ is a connected component and $\leq$ is total
  on each connected component.  Since $y\in S$ was chosen arbitrarily,
  it follows that $x\in S'$.  Hence $\min S\subseteq S'$.  To show the
  inclusion holds in the other direction, take $x\in S'$. If $y\in S$,
  then it follows from $x\in S'$ by the definition of $S'$ that
  $x\leq y$, from which we obtain $y\not<x$ by the definition of
  $\not<$. Since $y\in S$ was chosen arbitrarily, it follows that
  $x\in\min S$.  Hence $S'\subseteq\min S$.

  For Item~\ref{i:wo}, it follows from the fact that $\leq$ is
  well-ordered that $W$ is a connected component.  Further, since
  $\leq$ is well-ordered, it is also locally well-ordered.  The result
  therefore follows by Item~\ref{i:lwo}.

  For Item~\ref{i:wf-smooth}, let us first assume that $M$ is
  well-founded. We wish to prove that each $S\in\wp(W)$ is smooth in
  $M$.  So take $S\in\wp(W)$.  Since $\emptyset$ is smooth in $M$, let
  us assume further that $S\neq\emptyset$.  Now take $x\in S$.  Since
  $x\in x^\down\cap S$ by the reflexivity of $\leq$, it follows that
  $x^\down\cap S\neq\emptyset$.  Therefore, since $M$ is well-founded,
  it follows that $\min(x^\down\cap S)\neq\emptyset$.  That is, there
  exists $y\in\min(x^\down\cap S)$. But then $y\leq x$ and $x\not<y$,
  from which it follows that $y\simeq x$ or $y<x$.  And if
  $y\simeq x$, then it follows from $y\in\min(x^\down\cap S)$ by the
  transitivity of $\leq$ that $x\in\min(x^\down\cap S)$.  So either we
  have $x\in\min(x^\down\cap S)$ or we have $y\in\min(x^\down\cap S)$
  and $y<x$.  Further, for each $m\in\min(x^\down\cap S)$ and
  $z\in S-(x^\down\cap S)$, we have by the transitivity of $\leq$ that
  $z\not<m$. And for each $m\in\min(x^\down\cap S)$ and
  $z\in(x^\down\cap S)$, we have by the definition of
  $\min(x^\down\cap S)$ that $z\not<m$.  But then
  $m\in\min(x^\down\cap S)$ implies $m\in\min S$.  Taken together, we
  have shown that for each $x\in S$, either $x\in\min S$ or there
  exists $y\in\min S$ such that $y<x$.  It follows that $S$ is smooth
  in $M$.  Since we have shown that every $S\in\wp(W)$ is smooth in
  $M$, it follows that $M$ is smooth.
  
  For the converse of Item~\ref{i:wf-smooth}, we assume that $M$ is
  smooth.  We wish to prove that $M$ is well-founded.  So take a
  nonempty $S\subseteq W$.  Since $S$ is nonempty, we have $x\in S$.
  But $M$ is smooth and so $S$ is smooth in $M$, and so it follows
  that $x\in\min S$ or there exists $y\in\min S$ such that $y<x$.  In
  either case, we have $\min S\neq\emptyset$.  So $M$ is well-founded.

  Items~\ref{i:wo-smooth} and \ref{i:lwo-smooth} follow from
  Item~\ref{i:wf-smooth} by Definition~\ref{definition:smooth}.
\end{proof}

\subsection{Results for \texorpdfstring{$\CDL$}{CDL}}

\begin{proof}[Proof of Theorem~\ref{theorem:CDL-wf}\eqref{i:CDL-wf}]
  Assume $M$ is well-founded and $M,w\models B^\psi\varphi$.  The
  latter means
  \begin{equation}
    \forall x\in\cc(w):\quad
    x^\down\cap\sem{\psi}=\emptyset
    \quad\text{or}\quad
    \exists y\in x^\down\cap\sem{\psi}:
    y^\down\cap\sem{\psi}\subseteq\sem{\varphi}
    \enspace.
    \label{eq:theorem:CDL-wf}
  \end{equation}
  If $\min\sem{\psi}\cap\cc(w)=\emptyset$, then
  $\min\sem{\psi}\cap\cc(w)\subseteq\sem{\varphi}$.  So let us assume
  further that $\min\sem{\psi}\cap\cc(w)\neq\emptyset$.  Take
  $z\in\min\sem{\psi}\cap\cc(w)$.  Since we then have
  $z\in z^\downarrow\cap\sem{\psi}$ by the reflexivity of $\leq$ and
  the definition of $\min\sem{\psi}$, it follows by
  \eqref{eq:theorem:CDL-wf} that
  \[
  \exists y\in z^\down\cap\sem{\psi}:
  y^\down\cap\sem{\psi}\subseteq\sem{\varphi} \enspace.
  \]
  Now $z\in\cc(w)$ and $y\in z^\down\cap\sem{\psi}$, so it follows
  that $y\in\cc(w)$ and therefore that
  $y\in\sem{\psi}\cap\cc(w)$. From this we obtain by
  $z\in\min\sem{\psi}\cap\cc(w)$ that $y\not<z$.  But $y\in z^\down$,
  and therefore we have $y\leq z$ and $z\leq y$.  As a result,
  $z\in y^\downarrow\cap\sem{\psi}\subseteq\sem{\varphi}$.  Since
  $z\in\min\sem{\psi}\cap\cc(w)$ was chosen arbitrarily, we have
  proved that $\min\sem{\psi}\cap\cc(w)\subseteq\sem{\varphi}$.

  Conversely, assume $M$ is well-founded and
  $\min\sem{\psi}\cap\cc(w)\subseteq\sem{\varphi}$.  To show that we
  have $M,w\models B^\psi\varphi$, we must show that
  \eqref{eq:theorem:CDL-wf} obtains.  For this it suffices for us to
  take $x\in\cc(w)$ satisfying $x^\down\cap\sem{\psi}\neq\emptyset$
  and prove that
  \begin{equation}
    \exists y\in x^\down\cap\sem{\psi}:
    y^\down\cap\sem{\psi}\subseteq\sem{\varphi}\enspace.
    \label{eq:theorem:CDL-wf2}
  \end{equation}
  Proceeding, since $x^\down\cap\sem{\psi}\neq\emptyset$ and $M$ is
  well-founded, it follows that there exists
  $y\in\min(x^\down\cap\sem{\psi})$.  Since
  $\min\sem{\psi}\cap\cc(w)\subseteq\sem{\varphi}$, if we can show
  that for every $z\in y^\down\cap\sem{\psi}$ we have
  $z\in\min\sem{\psi}\cap\cc(w)$, then it would follow that
  $y^\down\cap\sem{\psi}\subseteq\sem{\varphi}$ and therefore that
  \eqref{eq:theorem:CDL-wf2}, completing the argument.  So take
  $z\in y^\down\cap\sem{\psi}$.  It follows by $z\in y^\down$,
  $y\in x^\down$, and $x\in\cc(w)$ that $z\in\cc(w)$.  So to show that
  $z\in\min\sem{\psi}\cap\cc(w)$, all that remains is to prove that
  $z\in\min\sem{\psi}$, and for this it suffices to prove that
  $u\in z^\down\cap\sem{\psi}$ implies $u\not<z$. So take
  $u\in z^\down\cap\sem{\psi}$.  Since $u\in z^\down$, $z\in y^\down$,
  and $y\in x^\down$, we have
  \[
  u\leq z\leq y\leq x\enspace.
  \]
  By the transitivity of $\leq$, it follows that $u\leq y$ and
  $u\leq x$.  Hence $u\in x^\down\cap\sem{\psi}$.  Since
  $y\in\min( x^\down\cap\sem{\psi})$, it follows that $u\not<y$.  But
  from $u\not<y$ and $u\leq y$, it follows by the definition of
  $\not<$ that $y\leq u$. And from $z\leq y\leq u$, it follows by the
  transitivity of $\leq$ that $z\leq u$.  Applying the definition of
  $\not<$, we obtain $u\not<z$.
\end{proof}

\begin{proof}[Proof of Theorem~\ref{theorem:CDL-theorems}]
  We reason in $\CDL$.  For (Cut), assume $B^\psi\varphi$ and
  $B^{\psi\land\varphi}\chi$.  It follows from
  $B^{\psi\land\varphi}\chi$ by (Inc) that $B^\psi(\varphi\to\chi)$.
  But from $B^\psi(\varphi\to\chi)$ and $B^\psi\varphi$ we obtain by
  (MR) that $B^\psi\chi$.
  
  For (CM), we reason by cases under the assumption $B^\psi\varphi$.
  First: if $\lnot B^\psi\lnot\varphi$, then we obtain
  $B^\psi\chi\to B^{\psi\land\varphi}\chi$ by (RM) and (CR).  Second:
  if $B^\psi\lnot\varphi$, then it follows by our assumption
  $B^\psi\varphi$ and (MR) that $B^\psi\bot$; applying (KM) and (CR)
  yields $B^{\psi\land\varphi}\bot$, from which we obtain
  $B^{\psi\land\varphi}\chi$ by (MR).

  (Taut) obtains by (CR) since $B\varphi=B^\top\varphi$.  (And)
  obtains by (MR).

  For (Or), we assume $B^{\psi_1}\varphi$ and $B^{\psi_2}\varphi$.
  Take $i\in\{1,2\}$.  By (Succ) and (MR) we obtain
  $B^{\psi_i}(\psi_1\lor\psi_2)$.  From $B^{\psi_i}(\psi_1\lor\psi_2)$
  and our assumption $B^{\psi_i}\varphi$, we obtain by (CM) and (Comm)
  that $B^{(\psi_1\lor\psi_2)\land\psi_i}\varphi$. From this it
  follows by (Inc) that $B^{\psi_1\lor\psi_2}(\psi_i\to\varphi)$.
  Since we have this for each $i\in\{1,2\}$, it follows by (MR) that
  $B^{\psi_1\lor\psi_2}((\psi_1\lor\psi_2)\to\varphi)$.  Since
  $B^{\psi_1\lor\psi_2}(\psi_1\lor\psi_2)$ by (Succ), we obtain the
  result $B^{\psi_1\lor\psi_2}\varphi$ by (MR).

  For (PR), we have the following:
  \begin{align*}
    1.
    &\;\;
       B^\varphi\lnot B^\psi\chi\to
      ( B^\varphi B^\psi\chi\to B^\varphi\bot)
    & \text{(MR)}
    \\
    2.
    &\;\;
      ( B^\varphi B^\psi\chi\land
      \lnot  B^\varphi\bot)\to
      \lnot B^\varphi\lnot B^\psi\chi
    & \text{(CR), 1}
    \\
    3.
    &\;\;
      \lnot B^\psi\chi \to
       B^\varphi \lnot B^\psi\chi
    & \text{(NI)}
    \\
    4.
    &\;\;
      \lnot B^\varphi \lnot B^\psi\chi\to
       B^\psi\chi
    & \text{(CR), 3}
    \\
    5. 
    &\;\;
      ( B^\varphi B^\psi\chi\land
      \lnot  B^\varphi\bot)\to
       B^\psi\chi
    & \text{(CR), 2, 4}
    \\
    6.
    &\;\;
       B^\varphi B^\psi\chi
      \to
      ( B^\varphi\bot\lor  B^\psi\chi)
    & \text{(CR), 5}
    \\
    7.
    &\;\;
       B^\psi\chi\to B^\varphi B^\psi\chi
    & \text{(PI)}
    \\
    8. 
    &\;\;
       B^\varphi\bot\to B^\varphi B^\psi\chi
    & \text{(MR)}
    \\
    9.
    &\;\;
      ( B^\varphi\bot\lor  B^\psi\chi) \to
       B^\varphi B^\psi\chi
    & \text{(CR), 7, 8}
    \\
    10. 
    &\;\;
       B^\varphi B^\psi\chi \leftrightarrow
      ( B^\varphi\bot\lor  B^\psi\chi)
    & \text{(CR), 6, 9}
  \end{align*}
  To obtain the proof for (NR), replace each occurrence of
  $ B^\psi\chi$ in the above proof with $\lnot B^\psi\chi$ and change
  the reason for line 7 from (PI) to (NI).  It is straightforward to
  verify that this operation yields a derivation of (NR).

  For (LE), assume $\psi\leftrightarrow \psi'$.  We have $B^\psi\psi$
  by (Succ).  Applying (MR) to our assumption, we obtain
  $B^{\psi}(\psi\to\psi')$. Hence $B^{\psi}\psi'$ by (MR).  By similar
  reasoning, we obtain $B^{\psi'}\psi$.  Now by (CM), $B^{\psi}\psi'$,
  (Comm), and (CR), we obtain
  $B^{\psi}\chi\to B^{\psi'\land\psi}\chi$. By $B^{\psi'}\psi$, (Cut),
  and (CR), we obtain $B^{\psi'\land\psi}\chi\to B^{\psi'}\chi$.  But
  then it follows by (CR) that $B^{\psi}\chi\to B^{\psi'}\chi$.  A
  similar argument shows that $B^{\psi'}\chi\to B^{\psi}\chi$.  By
  (CR), we conclude that $B^{\psi}\chi\leftrightarrow B^{\psi'}\chi$.

  (RW) follows by (MR).  For (SC), from $\psi\to\varphi$ we obtain
  $B^\psi(\psi\to\varphi)$ by (MN); however, we have $B^\psi\psi$ by
  (Succ), and so it follows by (MR) that $B^\psi\varphi$.
\end{proof}

\begin{proof}[Proof of Theorem~\ref{theorem:CDL-equivalence}]
  Left to right: it suffices to show that $\CDL$ derives (IEa), (IEb),
  and (LE).
  \begin{itemize}
  \item (IEa): $\vdash_\CDL B^\psi\varphi\to
    (B^{\psi\land\varphi}\chi\leftrightarrow B^\psi\chi)$.

    By (CM), (Cut), and (CR).

  \item (IEb): $\vdash_\CDL \lnot B^\psi\lnot\varphi\to
    (B^{\psi\land\varphi}\chi\leftrightarrow B^\psi(\varphi\to\chi))$.

    Reasoning in $\CDL$, by (Inc) and (CR) we have
    \[
    \lnot B^\psi\lnot\varphi\to (B^{\psi\land\varphi}\chi\to
    B^\psi(\varphi\to\chi)) \enspace,
    \]
    and so it suffices by (CR) to prove that
    \begin{equation}
      \lnot B^\psi\lnot\varphi\to
      (B^\psi(\varphi\to\chi)\to B^{\psi\land\varphi}\chi)
      \enspace.
      \label{eq:IEb-LtoR}
    \end{equation}
    Proceeding, we have
    \begin{equation}
      \lnot B^\psi\lnot\varphi\to(B^\psi(\varphi\to\chi)\to
      B^{\psi\land\varphi}(\varphi\to\chi))
      \label{eq:IEb-RM}
    \end{equation}
    by (RM).  We also have $B^{\psi\land\varphi}(\psi\land\varphi)$ by
    (Succ) and therefore $B^{\psi\land\varphi}\varphi$ by (MR).  But
    then we obtain \eqref{eq:IEb-LtoR} from \eqref{eq:IEb-RM} and
    $B^{\psi\land\varphi}\varphi$ by (MR). The result follows.

  \item (LE): if $\vdash_\CDL\psi\leftrightarrow\psi'$, then $\vdash_\CDL
    B^\psi\chi\leftrightarrow B^{\psi'}\chi$.

    By Theorem~\ref{theorem:CDL-theorems}.
  \end{itemize}
  This completes the left-to-right direction.  Right to left: it
  suffices to show that $\CDL_0$ derives (KM), (RM), (Inc), and
  (Comm).
  \begin{itemize}
  \item (KM): $\vdash_{\CDL_0} B^\psi\bot\to B^{\psi\land\varphi}\bot$.

    We reason in $\CDL_0$.  We have
    $B^\psi\bot\to(B^\psi\varphi\land B^\psi\bot)$ by (MR). Applying
    (IEa) and (CR), we obtain
    $B^\varphi\bot\to B^{\psi\land\varphi}\bot$.
  
  \item (RM): $\vdash_{\CDL_0} \lnot B^\psi\lnot\varphi\to( B^\psi\chi\to
    B^{\psi\land\varphi}\chi)$.
   
    Reasoning in $\CDL_0$, we assume $\lnot B^\psi\lnot\varphi$ and
    $B^\psi\chi$. It follows from $B^\psi\chi$ by (MR) that
    $B^\psi(\varphi\to\chi)$. From this and our assumption
    $\lnot B^\psi\lnot\varphi$, we obtain $B^{\psi\land\varphi}\chi$
    by (IEb) and (CR).

  \item (Inc): $\vdash_{\CDL_0} B^{\psi\land\varphi}\chi\to
    B^\psi(\varphi\to\chi)$.
   
    We reason in $\CDL_0$ by cases. First: if
    $\lnot B^\psi\lnot\varphi$, then we obtain
    $B^{\psi\land\varphi}\chi\to B^\psi(\varphi\to\chi)$ by (IEb).
    Second: if $B^\psi\lnot\varphi$, then we obtain
    $B^\psi(\varphi\to\chi)$ by (MR) and therefore
    $B^{\psi\land\varphi}\chi\to B^\psi(\varphi\to\chi)$ by (CR).
  
  \item (Comm): $\vdash_{\CDL_0} B^{\psi\land\varphi}\chi\to
    B^{\varphi\land\psi}\chi$.

    By (LE) and (CR). \qedhere
  \end{itemize}
\end{proof}
 
\begin{proof}[Proof of Theorem~\ref{theorem:CDL-completeness}]
  We use the notation and concepts from
  Remark~\ref{remark:multi-agent-CDL0}. Let
  \[
  \Lang_{\CDL_a}\coloneqq\Lang^{\{a\}}_\CDL \quad\text{and}\quad
  \CDL_a \coloneqq \CDL_0^{\{a\}}\enspace.
  \]
  We write $\chi^a$ for the $\Lang_{\CDL_a}$-formula obtained from the
  $\Lang_{\CDL_0}$-formula $\chi$ by recursively replacing each
  occurrence of a modal operator $B^\theta$ in $\chi$ by $B^\theta_a$.
  Obviously, $(\chi^a)'=\chi$ and $(\theta')^a=\theta$.

  It was shown by Board \cite{Boa04:GEB} that we have
  $\vdash_{\CDL_a}\chi$ iff $\models^{\{a\}}_{\CDL_a}\chi$.  By
  Remark~\ref{remark:multi-agent-CDL0}, this is equivalent to the
  statement that
  \begin{equation}
    \forall\chi\in\Lang_{\CDL_a}:\quad
    \vdash_{\CDL_a}\chi \quad\text{iff}\quad
    \models\chi' \enspace.
    \label{eq:CDL-completeness}
  \end{equation}
  By induction on derivation length, it is easy to see that the
  operation $\chi\mapsto\chi^a$ maps $\CDL_0$-theorems to
  $\CDL_a$-theorems and the operation $\chi\mapsto\chi'$ maps
  $\CDL_a$-theorems to $\CDL_0$-theorems.  So $\vdash_{\CDL_0}\varphi$
  iff $\vdash_{\CDL_0}\varphi^a$.  Applying
  \eqref{eq:CDL-completeness} and the fact that
  $(\varphi^a)'=\varphi$, we obtain $\vdash_{\CDL_0}\varphi$ iff
  $\models\varphi$.  Applying
  Theorem~\ref{theorem:CDL-equivalence}, we obtain
  $\vdash_\CDL\varphi$ iff $\models\varphi$.
\end{proof}

\begin{proof}[Proof of Theorem~\ref{theorem:CDL-completeness}]
  The argument can be obtained by combining the ideas from the various
  proofs in \cite{Boa04:GEB}. However, this requires restricting to
  the single-agent case and combining multiple arguments, so it is not
  so transparent how the argument should go. In the interest of making
  the argument clear and so that we have some constructions available
  for us later when we turn to the $\JCDL$ case, we provide a full
  proof here. However, the argument is truly due to \cite{Boa04:GEB}.

  For soundness, we proceed by induction on the length of
  derivation. In the induction base, we must show that each axiom
  scheme is wf-valid. (CL) is straightforward, so we proceed with the
  remaining schemes.  Let $M$ be an arbitrary well-founded
  plausibility model. We make tacit use of
  Theorem~\ref{theorem:CDL-wf}.
  \begin{itemize}
  \item (K) is valid:
    $\models B^\psi(\varphi_1\imp\varphi_2)\to
    (B^\psi\varphi_1\to B^\psi\varphi_2)$.

    Assume $(M,w)$ satisfies $B^\psi(\varphi_1\imp\varphi_2)$ and
    $B^\psi\varphi_1$.  Then
    $\min\sem{\psi}\subseteq\sem{\varphi_1\imp\varphi_2}$ and
    $\min\sem{\psi}\subseteq\sem{\varphi_1}$.  Hence
    $\min\sem{\psi}\subseteq\sem{\varphi_2}$.  So $(M,w)$ satisfies
    $ B^\psi\varphi_2$.
    
  \item (Succ) is valid: $\models B^\psi\psi$.

    We have $\min\sem{\psi}\subseteq\sem{\psi}$ by the definition of
    $\min\sem{\psi}$.  So $(M,w)$ satisfies $B^\psi\psi$.
    
  \item (KM) is valid:
    $\models B^\psi\bot\to B^{\psi\land\varphi}\bot$.

    Suppose $(M,w)$ satisfies $B^\psi\bot$.  Then
    $\min\sem{\psi}\subseteq\sem{\bot}=\emptyset$.  Since $M$ is
    well-founded, it follows that $\sem{\psi}=\emptyset$.  But then
    $\sem{\psi\land\varphi}=\emptyset$, from which it follows that
    $\min\sem{\psi\land\varphi}=\emptyset\subseteq\sem{\bot}$.  So
    $(M,w)$ satisfies $B^{\psi\land\varphi}\bot$.

  \item (RM) is valid:
    $\models \lnot B^\psi\lnot\varphi\to( B^\psi\chi\to
    B^{\psi\land\varphi}\chi)$.

    Suppose $(M,w)$ satisfies $\lnot B^\psi\lnot\varphi$ and
    $B^\psi\chi$.  It follows that
    $\min\sem{\psi}\nsubseteq\sem{\lnot\varphi}$ and
    $\min\sem{\psi}\subseteq\sem{\chi}$.  Hence
    $\min\sem{\psi}\cap\sem{\varphi}\neq\emptyset$.  We prove that
    $\min\sem{\psi\land\varphi}\subseteq
    \min\sem{\psi}\cap\sem{\varphi}$.
    Proceeding, take $x\in\min\sem{\psi\land\varphi}$.  Since
    $\min\sem{\psi}\cap\sem{\varphi}\neq\emptyset$, there exists
    $y\in\min\sem{\psi}\cap\sem{\varphi}$.  Hence
    $y\in\sem{\psi\land\varphi}$.  Since
    $x\in\min\sem{\psi\land\varphi}$ and $\leq$ is total, we have
    $x\leq y$. But $x\in\sem{\psi}$ and $y\in\min\sem{\psi}$, and
    therefore it follows from $x\leq y$ that $x\in\min\sem{\psi}$ as
    well.  Since $x\in\sem{\varphi}$, we have
    $x\in\min\sem{\psi}\cap\sem{\varphi}$.  Conclusion:
    $\min\sem{\psi\land\varphi}\subseteq\min\sem{\psi}\cap\sem{\varphi}$.
    So since $\min\sem{\psi}\subseteq\sem{\chi}$, it follows that
    \[
    \min\sem{\psi\land\varphi} \subseteq
    \min\sem{\psi}\cap\sem{\varphi}
    \subseteq\sem{\chi}\cap\sem{\varphi} \subseteq\sem{\chi}\enspace.
    \]
    That is, $\min\sem{\psi\land\varphi}\subseteq\sem{\chi}$.  So
    $(M,w)$ satisfies $B^{\psi\land\varphi}\chi$.

  \item (Inc) is valid:
    $\models B^{\psi\land\varphi}\chi\to
    B^\psi(\varphi\to\chi)$.

    Suppose $(M,w)$ satisfies $B^{\psi\land\varphi}\chi$.  Then
    $\min\sem{\psi\land\varphi}\subseteq\sem{\chi}$. We prove that
    $\min\sem{\psi}\subseteq\sem{\varphi\to\chi}$. Proceeding, take
    $x\in\min\sem{\psi}$.  If $x\notin\sem{\varphi}$, then
    $x\in\sem{\varphi\to\chi}$. So let us assume further that
    $x\in\sem{\varphi}$ and therefore that
    $x\in\sem{\psi\land\varphi}$. Now take any
    $y\in\sem{\psi\land\varphi}$. Since $y\in\sem{\psi}$, if we had
    $y<x$, then it would follow that $x\notin\min\sem{\psi}$,
    contradicting our choice of $x$. Hence
    $y\in\sem{\psi\land\varphi}$ implies $y\not<x$, from which it
    follows by $x\in\sem{\psi\land\varphi}$ that
    $x\in\min\sem{\psi\land\varphi}$.  But we have
    $\min\sem{\psi\land\varphi}\subseteq\sem{\chi}$ and hence
    $x\in\sem{\chi}$, from which we obtain $x\in\sem{\varphi\to\chi}$.
    Conclusion: $(M,w)$ satisfies $B^\psi(\varphi\to\chi)$.
    
  \item (Comm) is valid:
    $\models B^{\psi\land\varphi}\chi\to
    B^{\varphi\land\psi}\chi$.

    Suppose $(M,w)$ satisfies $B^{\psi\land\varphi}\chi$.  Then
    $\min\sem{\psi\land\varphi}\subseteq\sem{\chi}$.  Since
    $\sem{\psi\land\varphi}=\sem{\varphi\land\psi}$, it follows that
    $\min\sem{\varphi\land\psi}\subseteq\sem{\chi}$.  But then $(M,w)$
    satisfies $B^{\varphi\land\psi}\chi$.
    
  \item (PI) is valid:
    $\models B^\psi\chi\to B^\varphi B^\psi\chi$.

    Suppose $(M,w)$ satisfies $B^\psi\chi$.  Hence
    $\min\sem{\psi}\subseteq\sem{\chi}$, which implies that $(M,v)$
    satisfies $B^\psi\chi$ for any given $v\in\cc(w)=W$.  That is,
    $\sem{B^\psi\chi}=W$.  Therefore,
    $\min\sem{\varphi}\subseteq\sem{B^\psi\chi}$.  Conclusion: $(M,w)$
    satisfies $ B^\varphi B^\psi\chi$.

  \item (NI) is valid:
    $\models \lnot B^\psi\chi\to B^\varphi\lnot B^\psi\chi$.

    Suppose $(M,w)$ satisfies $\lnot B^\psi\chi$. Then
    $\min\sem{\psi}\nsubseteq\sem{\chi}$.  It follows that
    $M,v\models\lnot B^\psi\chi$ for each $v\in\cc(w)=W$.  Therefore,
    $\sem{\lnot B^\psi\chi}=W$, from which it follows that
    $\min\sem{\varphi}\subseteq\sem{\lnot B^\psi\chi}$.  
    Conclusion:
    $(M,w)$ satisfies $ B^\varphi\lnot B^\psi\chi$.
    
  \item (WCon) is valid: $\models B^\psi\bot\to\lnot\psi$.
    
    Suppose $(M,w)$ satisfies $B^\psi\bot$.  It follows that
    $\min\sem{\psi}\subseteq\sem{\bot}=\emptyset$.  Since $\leq$ is
    well-founded, it follows that $\sem{\psi}=\emptyset$. But this
    implies $\sem{\lnot\psi}=W$.  So $(M,w)$ satisfies $\lnot\psi$.
  \end{itemize}
  This completes the induction base. For the induction step, we must
  show that validity is preserved under the rules of (MP) and (MN).
  The argument for (MP) is standard, so let us focus on (MN).  We
  assume $\models\varphi$ for the $\CDL$-derivable $\varphi$
  (this is the ``induction hypothesis''), and we prove that
  $\models B^\psi\varphi$. Proceeding, since we have
  $\models\varphi$ by the induction hypothesis, it follows that
  $\sem{\varphi}=W$ and hence that
  $\min\sem{\psi}\subseteq\sem{\varphi}$.  But then $(M,w)$ satisfies
  $B^\psi\varphi$.  Soundness has been proved.

  Since $\CDL$ is sound with respect to the class of well-ordered
  plausibility models we note that $\CDL$ is consistent (i.e.,
  $\nvdash_\CDL\bot$).  In particular, take any pointed plausibility
  model $(M,w)$ containing only the single world $w$.  Since there is
  only one world, $M$ is well-ordered.  Further, by soundness, we have
  that $\vdash_\CDL\varphi$ to $M,w\models\varphi$. Therefore, since
  $M,w\not\models\bot$ by Definition~\ref{definition:CDL-truth}, it
  follows that $\nvdash_\CDL\bot$. That is, $\CDL$ is consistent.  We
  make use of this fact tacitly in what follows.

  For completeness, we write $\vdash$ without any subscript in the
  remainder of this proof as an abbreviation for $\vdash_\CDL$. Take
  $\theta$ such that $\nvdash \lnot\theta$.  We shall prove that
  $\theta$ is satisfiable at a locally well-ordered pointed
  plausibility model and then apply Theorem~\ref{theorem:CDL-lwo}
  to draw the desired conclusion.  Our construction is based on the
  completeness constructions given in \cite{Boa04:GEB}. Proceeding,
  provability will always be taken with respect to $\CDL$, the
  language is assumed to be $\Lang_\CDL$, and we make tacit use of
  Theorem~\ref{theorem:CDL-theorems}. To say that a set of formulas is
  \emph{consistent} means that for no finite subset does the
  conjunction provably imply $\bot$. For sets $S$ and $S'$ of
  formulas, to say that $S$ is \emph{maximal consistent in} (or
  ``maxcons in'') $S'$ means that $S\subseteq S'$, $S$ is consistent,
  and extending $S$ by adding any formula in $S'$ not already present
  would yield a set that is \emph{inconsistent} (i.e., not
  consistent).  Given a formula $\varphi$, we write $\sub(\varphi)$ to
  denote the set of sub-formulas of $\varphi$, including $\varphi$
  itself.  We extend this definition to sets of formulas: for a set
  $S$ of formulas,
  $\sub(S)\coloneqq\bigcup_{\varphi\in S}\sub(\varphi)$.  Given a set
  $S$ of formulas, we write ${\oplus}S$ to denote the \emph{Boolean
    closure} of $S$ (with respect to the language): this is the
  smallest extension of $S$ that contains all Boolean constants that
  are primitive to the language (i.e., $\bot$) and is closed under all
  Boolean operations that are primitive to the language (i.e.,
  implication).  Since our language is Boolean complete (i.e., every
  Boolean constant and every Boolean connective is definable in terms
  of the Boolean constants and Boolean connectives primitive to the
  language), it follows that the Boolean closure ${\oplus}S$ of $S$ is
  the smallest extension of $S$ that contains all definable Boolean
  constants (i.e., $\bot$ and $\top$) and is closed under all
  definable Boolean connectives (e.g., implication, conjunction,
  disjunction, and negation).  For a set $S$ of formulas, we define:
  \begin{align*}
    {\pm}S 
    &\coloneqq S\cup\{\lnot\varphi\mid\varphi\in S\}
      \enspace,
    \\
    B_0S
    &\coloneqq S\enspace,
    \\
    B_{i+1}S
    &\coloneqq
      {\pm}\{
      B^\psi\varphi \mid 
      \psi\in S \text{ and } 
      \varphi\in B_iS
      \}
      \enspace,
    \\
    B_\omega S
    &\coloneqq\textstyle\bigcup_{0<i<\omega}B_iS
      \enspace,
    \\
    C_0
    &\coloneqq{\pm}\,\sub(\{\theta,\bot,\top\})
      \enspace,
    \\
    C_1
    &\coloneqq{\oplus}C_0
      \enspace,
    \\
    B
    &\coloneqq B_\omega C_1\enspace,
    \\
    C 
    &\coloneqq C_1\cup B
      \enspace.
  \end{align*}
  Notice that $0$ is excluded in the definition of $B_\omega S$.
  Further, $C_0$ is finite.  Since $\nvdash \lnot\theta$, we may
  extend $\{\lnot\theta\}$ to a set $w_\theta$ that is maxcons in $C$.
  We then define:
  \begin{align*}
    W
    &\coloneqq\{x\subseteq C\mid x \text{ is maxcons in } C\}
      \enspace,
    \\
    \bar x
    &\textstyle\coloneqq
      \bigwedge(x\cap C_0) \text{ for } x\in W
      \enspace,
    \\
    x^\psi
    &\textstyle\coloneqq
      \{\varphi\mid B^\psi\varphi\in x\}
      \text{ for } x\in W \text{ and } \psi\in C_1
      \enspace,
    \\
    {\leq }
    &\coloneqq
      \{(x,y)\in W\times W\mid
      \exists \psi\in(x\cap y\cap C_1),\,
      y^\psi\subseteq x
      \}
      \enspace,
    \\
    V(x)
    &\coloneqq\Prop\cap x \text{ for } x\in W
      \enspace,
    \\
    M
    &\coloneqq(W,\leq ,V)
      \enspace.
  \end{align*}

  Notice that for $x\in W$, we have $\bar x\in C_1$, from which it
  follows by the maximal consistency of $x$ in $C\supseteq C_1$ that
  $\bar x\in x$.  Further, it follows by the definition
  $\bar x=\bigwedge(x\cap C_0)$ and the fact that $x$ is maximal
  consistent in $C=C_1\cup B$ that for each $\psi\in C_1={\oplus}C_0$,
  we have $\vdash\bar x\to\psi$ or $\vdash\bar x\to\lnot\psi$.
  Finally, if $x\leq y$ and $x\neq y$, then it follows by the maximal
  consistency of $x$ and of $y$ in $C$ that $\bar x\notin y$ and
  $\bar y\notin x$ and therefore that $\vdash\bar x\to\lnot\bar y$ and
  $\vdash\bar y\to\lnot\bar x$.  We make tacit use of the facts
  mentioned in this paragraph in what follows.

  We prove that $W$ is finite. First a definition due to
  \cite{Segerberg71}: to say a set $S$ of formulas is \emph{logically
    finite} means that $S$ has a \emph{finite basis}, which is a
  finite $S'\subseteq S$ satisfying the property that for every
  $\chi\in S$, there exists $\chi'\in S'$ such that
  $\vdash \chi\leftrightarrow\chi'$.  It can be shown by a normal form
  argument that if $S$ is logically finite, then there can be only
  finitely many sets that are maximal consistent in $S$.  So to prove
  that $W$ is finite, it suffices to prove that $C$ is logically
  finite. Proceeding, since $C_1$ was obtained as the Boolean closure
  of the finite set $C_0$, it follows by a normal form argument that
  $C_1$ has a finite basis $C_1'$.  So for $B^\psi\varphi\in B_1C_1$,
  since we have that $\psi\in C_1$ and $\varphi\in B_0C_1=C_1$, there
  exists $\psi'\in C_1'$ and $\varphi'\in C_1'$ such that
  \[
  \textstyle \vdash \psi\leftrightarrow\psi'
  \quad\text{and}\quad
  \vdash \varphi\leftrightarrow\varphi'\enspace.
  \]
  Applying (LE) and modal reasoning, it follows that
  \[
  \vdash B^\psi\varphi\leftrightarrow B^{\psi'}\varphi'\enspace.
  \]
  Since $ B^{\psi'}\varphi'\in B_1C_1'$ and
  $\lnot B^{\psi'}\varphi'\in B_1C_1'$, it follows by classical
  reasoning that $B_1C_1'$ is a finite basis for $B_1C_1$.  Therefore,
  there exists a finite basis $B_1'$ for ${\oplus}B_1C_1'$.  We prove
  by induction on positive $i<\omega$ that the set $B_1'$ is also a
  finite basis for $B_iC_1$.  The induction base case $i=1$ (for
  $B_1C_1$) has already been handled. So let us proceed with the
  induction step: we assume $B_1'$ is a finite basis for $B_jC_1$ for
  each non-negative integer $j$ that does not exceed some fixed
  $i\geq 1$ (this is the ``induction hypothesis''), and we prove that
  $B_1'$ is a finite basis for $B_{i+1}C_1$.  Proceeding, take
  $\chi\in B_{i+1}C_1$.  Since $i\geq 1$, we have $i+1\geq 2$ and
  therefore $\chi$ has one of the forms $B^\psi B^\delta\varphi$,
  $B^\psi\lnot B^\delta\varphi$, $\lnot B^\psi B^\delta\varphi$, or
  $ B^\psi\lnot B^\delta\varphi$ for some $\varphi\in B_{i-1}C_1$. We
  have by (PR), (NR), and classical reasoning that the following
  ``reductive equivalences'' obtain:
  \begin{align*}
    &
      \vdash   B^\psi B^\delta\varphi\leftrightarrow
      ( B^\psi\bot\lor B^\delta\varphi)
      \enspace,
    \\
    &
      \vdash  \lnot B^\psi B^\delta\varphi\leftrightarrow
      (\lnot B^\psi\bot\land\lnot B^\delta\varphi)
      \enspace,
    \\
    &
      \vdash   B^\psi\lnot B^\delta\varphi\leftrightarrow
      ( B^\psi\bot\lor\lnot B^\delta\varphi)
      \enspace,
    \\
    &
      \vdash  
      \lnot B^\psi\lnot B^\delta\varphi\leftrightarrow
      (\lnot B^\psi\bot\land B^\delta\varphi)
      \enspace.
  \end{align*}
  $ B^\psi\bot$ and $\lnot B^\psi\bot$ are members of
  $ B_1C_1$, and $ B^\delta\varphi$ and
  $\lnot B^\delta\varphi$ are members of $ B_iC_1$.  Since
  $i+1>1$ and $i+1>i$, we may apply the induction hypothesis: there
  exist members $( B^\psi\bot)'$, $(\lnot B^\psi\bot)'$,
  $( B^\delta\varphi)'$, and $(\lnot B^\delta\varphi)'$ of
  $B_1'$ such that the following ``inductive equivalences'' obtain:
  \begin{align*}
    &
      \vdash  
       B^\psi\bot
      \leftrightarrow
      ( B^\psi\bot)'
      \enspace,
    &
      \vdash  
      \lnot B^\psi\bot
      \leftrightarrow
      (\lnot B^\psi\bot)'
      \enspace,
    \\
    &
      \vdash  
       B^\delta\varphi
      \leftrightarrow
      ( B^\delta\varphi)'
      \enspace,
    &
      \vdash  
      \lnot B^\delta\varphi
      \leftrightarrow
      (\lnot B^\delta\varphi)'
      \enspace.
  \end{align*}
  Let us call the four formulas appearing on the right sides of the
  inductive equivalences the ``reduced formulas.''  Each reduced
  formula is a member of the finite basis $B_1'$ for
  ${\oplus} B_1C_1$, and therefore each reduced formula is also a
  member of ${\oplus} B_1C_1$.  Since ${\oplus} B_1C_1$ is closed
  under all definable Boolean operations and $B_1'$ is a finite basis
  for ${\oplus} B_1C_1$, it follows that any Boolean combination of
  the reduced formulas is also a member of ${\oplus} B_1C_1$ and
  therefore that any such Boolean combination is provably equivalent
  to a formula in $B_1'$.  But then it follows by the inductive
  equivalences and classical reasoning that the right side of each
  reductive equivalence is provably equivalent to a Boolean
  combination of reduced formulas, and the latter combination is
  itself provably equivalent to a formula in $B_1'$.  It follows that
  our original formula $\chi\in B_{i+1}C$ must be provably equivalent
  to a formula in $B_1'$ as well.  Therefore, $B_1'$ is indeed a
  finite basis for $ B_{i+1}C_1$.  This completes the induction step.
  We conclude that $B_1'$ is a finite basis for $ B_iC_1$ for each
  $i\geq 1$.  As a result, it follows that $B_1'$ is a finite basis
  for $B= B_\omega C_1=\bigcup_{0<i<\omega} B_iC_1$.  But then
  $C_1'\cup B_1'$ is a finite basis for $C=C_1\cup B$.  Conclusion:
  $W$ is finite.

  Suppose we are given $x\in W$ and $\varphi\in x$.  If
  $\vdash\varphi\to\psi$ and $\psi\in C$, then it follows by the
  maximal consistency of $x$ in $C$ that $\psi\in C$.  It is tedious
  to repeatedly verify membership assertions in $C$ and state that the
  reason the result follows is by the fact that $x$ is maximal
  consistent in $C$.  Therefore, we adopt the convention that we shall
  generally only write that $\vdash\varphi\to\psi$ and $\varphi\in x$
  together imply $\psi\in x$.  In so doing, we tacitly indicate (and
  the reader should verify) that $\psi\in C$, $x$ is maximal
  consistent in $C$, and so the result follows by the maximal
  consistency of $x$ in $C$.  The reader will always know when such
  tacit use takes place (and requires verification), since this use
  occurs every time it is stated that a membership assertion
  $\psi\in x$ obtains as a logical consequence of some collection of
  assumptions that does not include the assumption $\psi\in x$ itself.
  Finally, for convenience in the remainder of the proof, we shall say
  that a set is ``maximal consistent'' to mean that it is maximal
  consistent in $C$.

  \emph{Agreement Lemma\/}: for each $\{x,y\}\subseteq W$, if
  $\psi\in C_1$ and $x^\psi\subseteq y$, then $x\cap B=y\cap B$.  We
  prove this now. Proceeding, assume $\psi\in C_1$ and
  $x^\psi\subseteq y$. If $B^\chi\varphi\in x$, then
  $B^\psi B^\chi\varphi\in x$ by (PI) and therefore
  $B^\chi\varphi\in y$ by $x^\psi\subseteq y$.  So
  $B^\chi\varphi\in x$ implies $B^\chi\varphi\in y$.  Now suppose
  $B^\chi\varphi\in y$. If we had $B^\chi\varphi\notin x$, it would
  follow by maximal consistency that $\lnot B^\chi\varphi\in x$, hence
  $B^\psi\lnot B^\chi\varphi\in x$ by (NI), and hence
  $\lnot B^\chi\varphi\in y$, contradicting the consistency of $y$
  because $B^\chi\varphi\in y$. So $B^\chi\varphi\in y$ implies
  $B^\chi\varphi\in x$. Conclusion: the Agreement Lemma obtains.  Note
  that we obtain from this lemma by the definition of $\leq$ that that
  $x\leq y$ implies $x\cap B=y\cap B$.

  We prove that $\leq$ is reflexive. By the definition of $\leq$, it
  suffices to prove that $x^{\bar x}\subseteq x$.  Proceeding, since
  $B^{\bar x}\varphi\in x$ implies
  $B^{\bar x}\varphi\in B=\bigcup_{0<i<\omega}B_iC_1$, all we need do
  is prove by induction on $i\geq 1$ that
  $B^{\bar x}\varphi\in x\cap B_iC_1$ implies $\varphi\in x$.
  \begin{itemize}
  \item Induction base: $B^{\bar x}\varphi\in x\cap B_1C_1$.

    We have $\varphi\in C_1$ and hence $\vdash\bar x\to\varphi$ or
    $\vdash\bar x\to\lnot\varphi$.  If $\vdash\bar x\to\lnot\varphi$,
    then it follows by (SC) that $B^{\bar x}\lnot\varphi\in x$.  Since
    we also have $B^{\bar x}\varphi\in x$, we obtain by modal
    reasoning that $B^{\bar x}\bot\in x$, from which it follows by
    (WCon) that $\lnot\bar x\in x$, a contradiction. Therefore it
    follows that $\vdash\bar x\to\varphi$ and hence $\varphi\in x$.

  \item Induction step: we assume the result for $i=1,\dots,k$ and we
    prove the result for $i=k+1$.
    
    Assume $B^{\bar x}\varphi\in x\cap B_{k+1}C_1$.  Since $k\geq 2$,
    it follows that $\varphi=B^\psi\chi$ or
    $\varphi=\lnot B^\psi\chi$. Therefore, by (PR) or (NR) we have
    $B^{\bar x}\bot\lor\varphi\in x$.  As in the induction base, it
    follows by (WCon) that $B^{\bar x}\bot\notin x$. Therefore,
    $\varphi\in x$.
  \end{itemize}
  Conclusion: $\leq$ is reflexive.

  We prove that $\leq $ is transitive.  Proceeding, assume
  $x\leq y\leq z$. This means there exists $a\in(x\cap y\cap C_1)$ and
  $b\in(y\cap z\cap C_1)$ such that $y^a\subseteq x$ and
  $z^b\subseteq x$. It follows that $a\lor b\in(x\cap z\cap C_1)$, and
  so to conclude $x\leq z$, it suffices for us to prove that
  $z^{a\lor b}\subseteq x$. Proceeding, we take an arbitrary
  $B^{a\lor b}\varphi\in x$ and we seek to prove that $\varphi\in
  z$. We consider two cases.
  \begin{itemize}
  \item Case: $B^{a\lor b}\lnot a\in z$.

    Since $B^{a\lor b}(a\lor b)\in z$ by (Succ), it follows by the
    assumption of this case and modal reasoning that
    $B^{a\lor b}b\in z$. Applying the latter and the assumption of the
    case again, we obtain by (CM) that
    $B^{(a\lor b)\land b}\lnot a\in z$, from which it follows by (LE)
    that $B^b\lnot a\in z$.  But $z^b\subseteq y$ and therefore
    $\lnot a\in y$. Since $a\in y$, we have reached a
    contradiction. So this case cannot obtain, and so there is nothing
    more to prove.

  \item Case: $\lnot B^{a\lor b}\lnot a\in z$.

    From the assumption of this case and $B^{a\lor b}\varphi\in z$ we
    obtain by (RM) that $B^{(a\lor b)\land a}\varphi\in z$.  Applying
    (LE), it follows that $B^a\varphi\in z$.  Hence
    $B^bB^a\varphi\in z$ by (PI).  Since $z^b\subseteq y$, it follows
    that $B^a\varphi\in y$.  Since $y^a\subseteq x$, we obtain
    $\varphi\in x$.
  \end{itemize}
  Conclusion: $\leq$ is transitive.

  We prove that $\leq$ is total on each connected component.
  Proceeding, suppose we have $w\in W$ and
  $(x,y)\in\cc(w)\times\cc(w)$.  Since $x\in\cc(w)$ and $y\in\cc(w)$,
  it follows that $v\leq w$ or $w\leq v$ for each
  $v\in\{x,y\}$. Applying the definition of $\leq$ and the Agreement
  Lemma, we obtain $x\cap B=y\cap B$.  We wish to prove that $x\leq y$
  or $y\leq x$.  We consider two cases.
  \begin{itemize}
  \item Case: $B^{\bar x\lor \bar y}\lnot\bar x\in x\cap y$.

    By (Succ), we have
    $B^{\bar x\lor\bar y}(\bar x\lor\bar y)\in x\cap y$.  Applying the
    assumption of this case and modal reasoning, we obtain
    $B^{\bar x\lor\bar y}\bar y\in x\cap y$.  We shall now prove that
    $y\leq x$.  Proceeding, take an arbitrary
    $B^{\bar x\lor y}\varphi\in x$. It follows from this by
    $B^{\bar x\lor\bar y}\bar y\in x$ and (CM) that
    $B^{(\bar x\lor\bar y)\land\bar y}\varphi\in x$. Applying (LE), we
    obtain $B^{\bar y}\varphi\in x$. Since $x\cap B=y\cap B$, we have
    $B^{\bar y}\varphi\in y$. But we saw in the argument for
    reflexivity that $y^{\bar y}\subseteq y$ and therefore
    $\varphi\in y$. So we have proved that for an arbitrary
    $B^{\bar y}\varphi\in y$, we obtain $\varphi\in y$.  That is, we
    have shown that $x^{\bar x\lor\bar y}\subseteq y$. Since
    $\bar x\lor\bar y\in(x\cap y\cap C_1)$, it follows that $y\leq x$.

  \item Case: $\lnot B^{\bar x\lor\bar y}\lnot\bar x\in x\cap y$.

    We prove that $x\leq y$. Proceeding, take an arbitrary
    $B^{\bar x\lor\bar y}\varphi\in y$.  It follows from this by the
    assumption of this case and (RM) that
    $B^{(\bar x\lor\bar y)\land\bar x}\varphi\in y$. Applying (LE),
    $B^{\bar x}\varphi\in y$.  Since $y\cap B=x\cap B$, we obtain
    $B^{\bar x}\varphi\in x$. But we saw in the argument for
    reflexivity that $x^{\bar x}\subseteq x$ and therefore
    $\varphi\in x$.  So we have proved that for an arbitrary
    $B^{\bar x\lor\bar y}\varphi\in y$, we obtain $\varphi\in x$. That
    is, $y^{\bar x\lor\bar y}\subseteq x$. Since
    $\bar x\lor\bar y\in(x\cap y\cap C_1)$, it follows that $x\leq y$.
  \end{itemize}
  Conclusion: $\leq$ is total on each connected component.

  Since $w_\theta\in W$, it follows that $W$ is nonempty.  Therefore,
  $\leq$ is a reflexive and transitive binary relation over the
  nonempty finite set $W$, and $\leq$ is total on each connected
  component.  It follows from the finiteness of $W$ that $\leq$ is
  well-founded.  Therefore, $M$ is a locally well-ordered plausibility
  model. We now prove a few lemmas that will be of assistance.
  
  \emph{Consistency Lemma\/}: for each $x\in W$, if $\psi\in x$, then
  $x^\psi$ is consistent.  We prove this now.  Proceeding, suppose
  $x^\psi$ is not consistent.  It follows that there exists a nonempty
  $\{\chi_1,\dots,\chi_n\}\subseteq x^\psi$ such that
  $\vdash(\bigwedge_{i\leq n}\chi_i)\to\bot$.  By modal reasoning,
  $\vdash(\bigwedge_{i\leq n}B^\psi\chi_i)\to B^\psi\bot$.  Since
  $\chi_i\in x^\psi$ and hence $B^\psi\chi_i\in x$ for each $i\leq n$,
  it follows that $B^\psi\bot\in x$.  Applying (WCon),
  $\lnot\psi\in x$, which contradicts the consistency of $x$ because
  we assumed $\psi\in x$.  Conclusion: $x^\psi$ is consistent.

  \emph{Minimality Lemma\/}: for each $\psi\in C_1$ and $x\in[\psi]$,
  where
  \[
  [\psi]\coloneqq\{x\in W\mid \psi\in x\}\enspace,
  \]
  we have $x\in\min[\psi]$ iff $\lnot B^\psi\lnot\bar x\in x$.
  \begin{itemize}
  \item Left to right: for $\psi\in C_1$ and $x\in[\psi]$, we prove
    $x\in\min[\psi]$ implies $\lnot B^\psi\lnot\bar x\in x$.

    Assume $\psi\in C_1$ and $x\in\min[\psi]$.  Toward a
    contradiction, suppose $B^\psi\lnot\bar x\in x$. Applying the
    Consistency Lemma, $x^\psi$ is consistent and so may be extended
    to a maximal consistent $y\in W$.  Since $B^\psi\psi\in x$ by
    (Succ), it follows that $\psi\in y$. And since
    $B^\psi\lnot\bar x\in x$, it follows that $\lnot\bar x\in y$ and
    therefore that $y\neq x$.  But then $\psi\in(x\cap y\cap C_1)$ and
    $y^\psi\subseteq x$, from which it follows that $y\leq x$.  Since
    $y\in[\psi]$, $y\leq x$, and $x\in\min[\psi]$, it follows that
    $x\leq y$. That is, there exists $\delta\in(x\cap y\cap C_1)$ such
    that $y^\delta\subseteq x$. Now if we had
    $B^\psi\lnot\delta\in x$, then it would follow that
    $\lnot\delta\in y$, contradicting the fact that $\delta\in y$.
    Hence $\lnot B^\psi\lnot\delta\in x$.  Since
    $B^\psi\lnot\bar x\in x$ as well, it follows by (RM) and (Comm)
    that $B^{\delta\land\psi}\lnot\bar x\in x$. Applying (Inc),
    $B^\delta(\psi\to\lnot\bar x)\in x$. Since $x\leq y$ implies
    $x\cap B=y\cap B$, it follows that
    $B^\delta(\psi\to\lnot\bar x)\in y$, from which it follows by
    $y^\delta\subseteq x$ that $\psi\to\lnot\bar x\in x$. Since
    $\psi\in x$, we obtain $\lnot\bar x\in x$, a contradiction. Our
    assumption $B^\psi\lnot\bar x\in x$ must have been incorrect, and
    so we must have $\lnot B^\psi\lnot\bar x\in x$ after all.

  \item Right to left: for $\psi\in C_1$ and $x\in[\psi]$, we prove
    $\lnot B^\psi\lnot\bar x\in x$ implies $x\in\min[\psi]$.

    Assume $\psi\in C_1$, $x\in[\psi]$ and,
    $\lnot B^\psi\lnot\bar x\in x$.  It suffices to show that for each
    $y\in[\psi]\cap\cc(x)$, we have $x\leq y$.  Proceeding, take an
    arbitrary $y\in[\psi]\cap\cc(x)$. It follows that
    $y\cap B=x\cap B$ and that $\psi\in(x\cap y\cap C_1)$.  So from
    $\lnot B^\psi\lnot\bar x\in x$ we obtain
    $\lnot B^\psi\lnot\bar x\in y$.  Now take an arbitrary
    $B^\psi\varphi\in y$.  It follows from this and
    $\lnot B^\psi\lnot\bar x\in y$ by (RM) that
    $B^{\psi\land\bar x}\varphi\in y$.  Since $\psi\in x\cap C_1$, we
    have $\vdash \bar x\to\psi$ and therefore it follows by (LE) that
    $B^{\bar x}\varphi\in y$.  Since $x\cap B=y\cap B$, we obtain
    $B^{\bar x}\varphi\in x$. However, we saw in the argument for
    reflexivity that $x^{\bar x}\subseteq x$, so it follows that
    $\varphi\in x$. That is, we have shown that
    $\psi\in(x\cap y\cap C_1)$ and $y^\psi\subseteq x$. Hence
    $x\leq y$. Since $y\in[\psi]\cap\cc(x)$ was chosen arbitrarily and
    $x\in[\psi]$, it follows that $x\in\min[\psi]$.
  \end{itemize}
  This completes the proof of the Minimality Lemma.

  \emph{Truth Lemma\/}: for each $\varphi\in C_1$, we have
  $[\varphi]=\sem{\varphi}_M$. We prove this now.  The proof is by
  induction on the construction of formulas in $C_1$.  The induction
  base case and Boolean induction step cases are standard, so we only
  consider the induction step case for formulas
  $B^\psi\varphi\in C_1$.  Note: by the definition of $C_1$, we have
  $B^\psi\varphi\in C_1$ iff $B^\psi\varphi\in C_0$. Further, from
  $B^\psi\varphi\in C_0$, it follows by the definition of $C_0$ that
  $\psi\in C_0$ and $\varphi\in C_0$.
  \begin{itemize}
  \item Induction step $B^\psi\varphi$ (left to right): if $x\in W$
    and $B^\psi\varphi\in x\cap C_1$, then $M,x\models B^\psi\varphi$.
    
    Assume $B^\psi\varphi\in x\cap C_1$.  If
    $\min\sem{\psi}\cap\cc(x)=\emptyset$, then, since $M$ is
    well-founded, the result follows immediately by
    Theorem~\ref{theorem:CDL-wf}.  So assume
    $\min\sem{\psi}\cap\cc(x)\neq\emptyset$ and take an arbitrary
    $y\in\min\sem{\psi}\cap\cc(x)$.  Applying the induction
    hypothesis, $y\in\min[\psi]$, from which it follows by the
    Minimality Lemma that $\lnot B^\psi\lnot\bar y\in y$. Since
    $y\in\cc(x)$ implies $x\cap B=y\cap B$, it follows that
    $\lnot B^\psi\lnot\bar y\in x$. It follows from this by the
    assumption and (RM) that $B^{\psi\land\bar y}\varphi\in x$.  Since
    $\psi\in y\cap C_1$, we have $\vdash\bar y\to\psi$ and therefore
    it follows from $B^{\psi\land\bar y}\varphi\in x$ by (LE) that
    $B^{\bar y}\varphi\in x$. Since $x\cap B=y\cap B$, we have
    $B^{\bar y}\varphi\in y$.  But as we saw that
    $y^{\bar y}\subseteq y$ in the argument for reflexivity, it
    follows that $\varphi\in y$ and hence that
    $y\in[\varphi]$. Applying the induction hypothesis,
    $y\in\sem{\varphi}$.  Since $y\in\min\sem{\psi}\cap\cc(x)$ was
    chosen arbitrarily, we have shown that
    $\min\sem{\psi}\cap\cc(x)\subseteq\sem{\varphi}$.  Since $M$ is
    well-founded, it follows by Theorem~\ref{theorem:CDL-wf} that
    $M,w\models B^\psi\varphi$.

  \item Induction step $B^\psi\varphi$ (right to left): if $x\in W$,
    $B^\psi\varphi\in C_1$, and $M,x\models B^\psi\varphi$, then
    $B^\psi\varphi\in x$.
    
    Assume $B^\psi\varphi\in C_1$ and $M,x\models B^\psi\varphi$.
    Since $M$ is well-founded, it follows by
    Theorem~\ref{theorem:CDL-wf} that
    $\min\sem{\psi}\cap\cc(x)\subseteq\sem{\varphi}$.  By the
    induction hypothesis, $\min[\psi]\cap\cc(x)\subseteq[\varphi]$.
    We wish to prove that $B^\psi\varphi\in x$.  We consider two
    cases.
    
    \emph{Case\/}: $\min[\psi]\cap\cc(x)=\emptyset$. Since $M$ is
    well-founded, it follows that $[\psi]\cap\cc(x)=\emptyset$.
    Toward a contradiction, assume $B^\psi\bot\notin x$.  Now if
    $x^\psi$ were not consistent, then there would exist
    $\{\chi_1,\dots,\chi_n\}\subseteq x^\psi$ such that
    $\vdash(\bigwedge_{i\leq n}\chi_i)\to\bot$, from which it would
    follow by modal reasoning and the fact that $B^\psi\chi_i\in x$
    for each $i\leq n$ that $B^\psi\bot\in x$, contradicting our
    assumption that $B^\psi\bot\notin x$. So $x^\psi$ is consistent
    after all, and we may extend this set to some $y\in W$.  Since
    $B^\psi\psi\in x$ by (Succ), it follows that $\psi\in y$ and
    therefore that $y\in[\psi]$.  But $[\psi]\cap\cc(x)=\emptyset$,
    and so we must have that $y\notin\cc(x)$. However,
    $x^\psi\subseteq y$, so it follows by the Agreement Lemma that
    $x\cap B=y\cap B$ and hence
    $x^{\bar x\lor\bar y}=y^{\bar x\lor\bar y}$.  If
    $x^{\bar x\lor\bar y}$ were not consistent, then there would exist
    $\{\chi_1,\dots,\chi_n\}\subseteq x^{\bar x\lor\bar y}$ such that
    $\vdash(\bigwedge_{i\leq n}\chi_i)\to\bot$, from which it would
    follow by modal reasoning and the fact that
    $B^{\bar x\lor\bar y}\chi_i\in x$ for each $i\leq n$ that
    $B^{\bar x\lor\bar y}\bot\in x$, from which we would obtain by
    (WCon) that $\lnot(\bar x\lor\bar y)\in x$ and hence that
    $\lnot\bar x\in x$, a contradiction. So
    $x^{\bar x\lor\bar y}=y^{\bar x\lor\bar y}$ is consistent and so
    may be extended to some $z\in W$.  By (Succ),
    $B^{\bar x\lor\bar y}(\bar x\lor\bar y)\in x\cap y$ and therefore
    $\bar x\lor\bar y\in z$.  But then
    $\bar x\lor\bar y\in(x\cap y\cap z\cap C_1)$,
    $x^{\bar x\lor\bar y}\subseteq z$, and
    $y^{\bar x\lor\bar y}\subseteq z$; that is, $z\leq x$ and
    $z\leq y$, from which it follows that $y\in\cc(x)$, a
    contradiction. Therefore our original assumption that
    $B^\psi\bot\notin x$ must have been incorrect, and so we must have
    $B^\psi\bot\in x$ after all. Applying modal reasoning, we obtain
    $B^\psi\varphi\in x$, as desired.

    \emph{Case\/}: $\min[\psi]\cap\cc(x)\neq\emptyset$. It follows
    that there exists $y\in\min[\psi]\cap\cc(x)$. Toward a
    contradiction, assume $B^\psi\varphi\notin y$.  If
    $y^\psi\cup\{\lnot\varphi\}$ were not consistent, then we would
    have a finite $\{\chi_1,\dots,\chi_n\}\subseteq y^\psi$ such that
    $\vdash(\bigwedge_{i\leq n}\chi_i)\to\varphi$, from which it would
    follow by modal reasoning and the fact that $B^\psi\chi_i\in y$
    for each $i\leq n$ that $B^\psi\varphi\in y$, contradicting the
    consistency of $y$ by our assumption that $B^\psi\varphi\notin y$.
    So $y^\psi\cup\{\lnot\varphi\}$ is indeed consistent and may be
    extended to some $z\in W$.  Since we have $B^\psi\psi\in y$ by
    (Succ), it follows that $\psi\in z\cap y\cap C_1$. So since
    $y^\psi\subseteq z$, it follows that $z\leq y$. But then
    $z\leq y$, $\psi\in z$, and $y\in\min[\psi]\cap\cc(x)$, so it
    follows that $z\in\min[\psi]\cap\cc(x)$.  Since
    $\min[\psi]\cap\cc(x)\subseteq[\varphi]$, it follows that
    $\varphi\in z$. However, by the construction of $z$ as a maximal
    consistent extension of $y^\psi\cup\{\lnot\varphi\}$, we also have
    $\lnot\varphi\in z$. So $z$ is inconsistent, a contradiction.  It
    follows that our assumption $\lnot B^\psi\varphi\in y$ must have
    been incorrect and therefore we must have $B^\psi\varphi\in y$
    after all.  Since $y\in\cc(x)$, we have $x\cap B=y\cap B$ and
    therefore that $B^\psi\varphi\in x$, as desired.
  \end{itemize}
  This completes the proof of the Truth Lemma.  Since
  $\lnot\theta\in w_\theta\cap C_1$, it follows by the Truth Lemma
  that $M,w_\theta\not\models\theta$.  But then we have shown that
  $\nvdash\theta$ implies $M\not\models\theta$ for our locally
  well-founded model $M\in\mathfrak{P}_L$.  Applying
  Theorem~\ref{theorem:CDL-lwo}, $\nvdash\theta$ implies
  $\not\models\varphi$. By contraposition, we have that
  $\models\theta$ implies $\vdash\theta$. So completeness obtains.
\end{proof}

\subsection{Results for AGM revision}

\begin{proof}[Proof of Theorem~\ref{theorem:AGM-wo}]
  For \eqref{i:cbs-base}, assume $T\in\mathfrak{B}_\CPL$ is
  consistent. To say a $S$ is \emph{maximal} $\CPL$\emph{-consistent}
  means that $S\subseteq\Lang_\CPL$, $S$ is consistent, and adding to
  $S$ any $\psi\in\Lang_\CPL$ not already present would result in a
  set that is inconsistent. We make tacit use of various well-known
  facts about maximal consistent sets.  Let $\mathfrak{S}$ be the
  collection of all maximal $\CPL$-consistent sets.  Let
  $\{\psi_i\}_{0<i<\omega}$ be an enumeration of $\Lang_\CPL$. Define
  $\omega^+\coloneqq\omega-\{0\}$ and
  $S+\psi\coloneqq\CPL(S\cup\{\psi\})$. Take
  \begin{align*}
    W \coloneqq\;
    &
      \{ (S,0)\in\mathfrak{S}\times\{0\}\mid
      T\subseteq S\}\;\cup
    \\
    &
      \{ (S,i)\in\mathfrak{S}\times\omega^+\mid
      \lnot\psi_i\in T-\CPL(\emptyset)
      \text{ and } 
      \CPL(\psi_i)\subseteq S\}
      \enspace,
    \\
    {\leq}\coloneqq\;
    &
      \{ ((S,i),(S',i'))\in W\times W \mid i\leq i' \}\enspace,
    \\
    V\coloneqq\;
    &
      \{((S,i),P)\in W\times\wp(\Prop)\mid P=S\cap\Prop\}\enspace,
    \\
    M\coloneqq\;
    &
      (W,\leq,V)\enspace.
  \end{align*}
  Since $T$ is consistent, it follows that there exists
  $S\in\mathfrak{S}$ such that $S\supseteq T$, which implies
  $(S,0)\in W$.  Therefore $W\neq\emptyset$.  Similarly, if
  $i\in\omega^+$ and $\lnot\psi_i\in\CPL(T)-\CPL(\emptyset)$, then it
  follows that $\psi_i$ is consistent and therefore there exists
  $S\in\mathfrak{S}$ such that $S\supseteq\CPL(\psi_i)$, which implies
  $(S,i)\in W$.  The relation $\leq$ on $W$ is a well-order because
  the relation $\leq$ on $\omega$ is a well-order. So $M$ is a
  well-ordered plausibility model.

  \emph{Truth Lemma\/}: for each $\varphi\in\Lang_\CPL$ and
  $(S,i)\in W$, we have $M,(S,i)\models\varphi$ iff $\varphi\in S$.
  The proof is by induction on the construction of
  $\varphi\in\Lang_\CPL$. Induction base: if $\varphi=\bot$, then
  $\bot\notin S$ since $S$ is consistent and $M,(S,i)\not\models\bot$
  by Definition~\ref{definition:CDL-truth}.  Induction base: if
  $\varphi=p\in\Prop$, then the result follows by the definition of
  $V$ and Definition~\ref{definition:CDL-truth}. Induction step: if
  $\varphi=\varphi_1\to\varphi_2$, then the result follows by the
  induction hypothesis and Definition~\ref{definition:CDL-truth}. The
  lemma follows.
  
  \emph{Theory Lemma\/}: for each $\psi\in\Lang_\CPL$, we have
  \[
  M^\down*_M\psi =
  \begin{cases}
    T+\psi & \text{if }\lnot\psi\notin T,
    \\
    \CPL(\psi) & \text{if }\lnot\psi\in T-\CPL(\emptyset),
    \\
    \Lang_\CPL & \text{if } \lnot\psi\in\CPL(\emptyset).
  \end{cases}
  \]
  We prove this now by a case distinction.  Proceeding, for each
  $W'\subseteq W$, define
  \[
  T(W')\coloneqq\{\varphi\in\Lang_\CDL\mid \forall w\in W':
  M,w\models\varphi\} \enspace.
  \]
  \begin{itemize}
  \item Case: $\lnot\psi\notin T$.
    
    Let
    $\mathfrak{S}_\psi\coloneqq\{S\in\mathfrak{S}\mid S\supseteq
    T\cup\{\psi\}\}$.
    It follows by the assumption of this case that
    $\mathfrak{S}_\psi\neq\emptyset$ and hence
    $\mathfrak{S}_\psi\times\{0\}\subseteq W$.  By the Truth Lemma and
    the definition of $\leq$, it follows that
    $\min\sem{\psi}_M=\mathfrak{S}_\psi\times\{0\}$.  Now notice that
    since $S\in\mathfrak{S}_\psi$ is maximal $\CPL$-consistent, it
    follows that for each $S'\subseteq S$ we have
    $\CPL(S')\subseteq S$.  Therefore, since
    $T\cup\{\psi\}\subseteq S$ for each $S\in\mathfrak{S}_\psi$, it
    follows that $T+\psi\subseteq S$ for each $S\in\mathfrak{S}_\psi$.
    So $T+\psi\subseteq\bigcap\mathfrak{S}_\psi$.  Now if we had some
    $\varphi\in\bigcap\mathfrak{S}_\psi$ such that
    $\varphi\notin T+\psi$, then it would follow that
    $\varphi\notin\CPL(T\cup\{\psi\})$ and therefore we could extend
    $T\cup\{\psi,\lnot\varphi\}$ to a maximal $\CPL$-consistent
    $S\in\mathfrak{S}_\psi$ having $\varphi\notin S$, contradicting
    the fact that $\varphi\in\bigcap\mathfrak{S}_\psi$. Therefore it
    must be the case that $\bigcap\mathfrak{S}_\psi\subseteq T+\psi$.
    Hence $T+\psi=\bigcap\mathfrak{S}_\psi$.  But then we have
    \[
    \textstyle
    T(\mathfrak{S}_\psi\times\{0\}) = \bigcap\mathfrak{S}_\psi =
    T+\psi\enspace,
    \]
    where the leftmost equality follows by the Truth Lemma.  Since
    $\min\sem{\psi}_M=\mathfrak{S}_\psi\times\{0\}$, it follows by the
    definition of $*_M$ that $M^\down*_M\psi=T+\psi$.

  \item Case: $\lnot\psi\in T-\CPL(\emptyset)$.

    Define
    $\mathfrak{S}_\psi\coloneqq\{S\in\mathfrak{S}\mid \psi\in S\}$.
    It follows by the definition of $W$ that there exists
    $j\in\omega^+$ and $S\in\mathfrak{S}_\psi$ such that $\psi_j=\psi$
    and $(S,j)\in W$.  So let $i\in\omega^+$ be the least positive
    integer satisfying the property that there exists $(S,i)\in W$
    with $\psi\in S$. By the Truth Lemma, the definition of $\leq$,
    and our choice of $i$, it follows that
    $\min\sem{\psi}_M=\mathfrak{S}_\psi\times\{i\}$. Since
    $S\in\mathfrak{S}_\psi$ is maximal $\CPL$-consistent, it follows
    that for each $S'\subseteq S$ we have $\CPL(S')\subseteq S$.
    Therefore, since $\psi\in S$ for each $S\in\mathfrak{S}_\psi$, it
    follows that $\CPL(\psi)\subseteq S$ for each
    $S\in\mathfrak{S}_\psi$.  So
    $\CPL(\psi)\subseteq\bigcap\mathfrak{S}_\psi$. If we had some
    $\varphi\in\bigcap\mathfrak{S}_\psi$ such that
    $\varphi\notin\CPL(\psi)$, then we could extend
    $\{\psi,\lnot\varphi\}$ to a maximal $\CPL$-consistent
    $S\in\mathfrak{S}_\psi$ having $\varphi\notin S$, contradicting
    the fact that $\varphi\in\bigcap\mathfrak{S}_\psi$. Therefore it
    must be the case that
    $\bigcap\mathfrak{S}_\psi\subseteq\CPL(\psi)$.  Hence
    $\CPL(\psi)=\bigcap\mathfrak{S}_\psi$. But then we have
    \[
    \textstyle T(\mathfrak{S}_\psi\times\{i\})=
    \bigcap\mathfrak{S}_\psi=\CPL(\psi)\enspace,
    \]
    where the leftmost equality follows by the Truth Lemma.  Since
    $\min\sem{\psi}_M=\mathfrak{S}_\psi\times\{i\}$, it follows by the
    definition of $*_M$ that $M^\down*_M\psi=\CPL(\psi)$.

  \item Case: $\lnot\psi\in\CPL(\emptyset)$.  

    It follows from $\lnot\psi\in\CPL(\emptyset)$ by the soundness of
    Classical Propositional Logic that $\not\models\psi$ and therefore
    $\min\sem{\psi}_M=\emptyset$. Hence
    $\min\sem{\psi}_M\subseteq\sem{\varphi}_M$ for each
    $\varphi\in\Lang_\CPL$.  Applying the definition of $*_M$, it
    follows that $M^\down*_M\psi=\Lang_\CPL$.
  \end{itemize}
  The lemma follows.

  
  We now prove that $M$ is a Grove system for $T$. First, since
  $M^\down=L_0$, it follows by the Theory Lemma that $M^\down=T$.
  Second, we have already seen that $M$ is a well-ordered plausibility
  model.  So all that remains is to prove that $M^\down*_M(-)$
  satisfies the AGM revision postulates.  So given an arbitrary
  $\psi\in\Lang_\CPL$, we check each postulate in turn.
  \begin{itemize}
  \item Closure: $M^\down*_M\psi=\CPL(M^\down*_M\psi)$.

    By the Theory Lemma, $M^\down*_M\psi$ is either $T+\psi$,
    $\CPL(\psi)$, or $\Lang_\CPL$.  However, for each of these sets
    $S$, we have $\CPL(S)=S$.

  \item Success:  $\psi\in M^\down*_M\psi$.

    By the Theory Lemma, $M^\down*_M\psi$ is either $T+\psi$,
    $\CPL(\psi)$, or $\Lang_\CPL$. However, each of these sets
    contains $\psi$.

  \item Inclusion: $M^\down*_M \psi\subseteq M^\down+\psi$.

    By the Theory Lemma, $M^\down*_M\psi$ is $T+\psi$ if
    $\lnot\psi\notin T$, is $\CPL(\psi)$ if
    $\lnot\psi\in T-\CPL(\emptyset)$, and is $\Lang_\CPL$ if
    $\lnot\psi\in\CPL(\emptyset)$.  Further $M^\down=T$.  Inclusion
    obviously follows for the case $\lnot\psi\notin T$.  for the case
    $\lnot\psi\in T-\CPL(\emptyset)$, Inclusion follows because
    $\CPL(\psi)\subseteq T+\psi=\CPL(T\cup\{\psi\})$. For the case
    $\lnot\psi\in\CPL(\emptyset)$, Inclusion follows because
    $T+\psi=\CPL(T\cup\{\psi\})=\Lang_\CPL$.

  \item Vacuity: if $\lnot\psi\notin M^\down$, then
    $M^\down*_M\psi=M^\down+\psi$.

    Since $M^\down=T$, the result follows by the Theory Lemma.
 
  \item Consistency: if $\lnot\psi\notin\CPL(\emptyset)$, then
    $\bot\notin\CPL(M^\down*_M\psi)$.

    If $\lnot\psi\notin\CPL(\emptyset)$, then either
    $\lnot\psi\notin T$ or $\lnot\psi\in T-\CPL(\emptyset)$.  Applying
    the Theory Lemma, $M^\down*_M\psi$ is either $T+\psi$ or
    $\CPL(\psi)$.  Since $T$ is consistent and we assumed $\psi$ is
    consistent, each of $T+\psi$ and $\CPL(\psi)$ is consistent.
    Hence $\bot\notin\CPL(M^\down*_M\psi)$.
 
  \item Extensionality: if
    $(\psi\leftrightarrow\psi')\in\CPL(\emptyset)$, then
    $M^\down*_M\psi=M^\down*_M\psi'$.

    Assume $(\psi\leftrightarrow\psi')\in\CPL(\emptyset)$.  It follows
    that $T+\psi=T+\psi'$ and $\CPL(\psi)=\CPL(\psi')$.  The result
    therefore follows by the Theory Lemma.

  \item Superexpansion:
    $M^\down*_M(\psi\land\varphi)\subseteq (M^\down*_M\psi)+\varphi$.

    \textbf{Case:} $\lnot(\psi\land\varphi)\notin T$.  By the Theory
    Lemma, $M^\down*_M(\psi\land\varphi)=T+(\psi\land\varphi)$. Since
    $T$ is a theory, it follows from the assumption of this case that
    $\lnot\psi\notin T$. Applying the Theory Lemma,
    $M^\down*_M\psi=T+\psi$.  The result follows because
    \[
    T+(\psi\land\varphi)= \CPL(T\cup\{\psi\land\varphi\})
    =\CPL(\CPL(T\cup\{\psi\})\cup\{\varphi\}) =(T+\psi)+\varphi\enspace.
    \]

    \textbf{Case:} $\lnot(\psi\land\varphi)\in T-\CPL(\emptyset)$.  By
    the Theory Lemma,
    $M^\down*_M(\psi\land\varphi)=\CPL(\psi\land\varphi)$.  If
    $\lnot\psi\notin T$, then it follows by the Theory Lemma that
    $M^\down*_M\psi=T+\psi$ and hence
    \[
    \CPL(\psi\land\varphi)\subseteq \CPL(T\cup\{\psi,\varphi\})=
    \CPL(\CPL(T\cup\{\psi\})\cup\{\varphi\})=(T+\psi)+\varphi
    \enspace.
    \]
    And if $\lnot\psi\in T-\CPL(\emptyset)$, then it follows by the
    Theory Lemma that $M^\down*_M\psi=\CPL(\psi)$ and hence
    $\CPL(\psi\land\varphi)= \CPL(\{\psi\}\cup\{\varphi\})=
    \CPL(\psi)+\varphi$.
    Finally, if $\lnot\psi\in\CPL(\emptyset)$, then it follows by the
    Theory Lemma that $M^\down*_M\psi=\Lang_\CPL$, from which we
    obtain
    \[
    \CPL(\psi\land\varphi)=\Lang_\CPL= \Lang_\CPL+\varphi\enspace.
    \]

    \textbf{Case:} $\lnot(\psi\land\varphi)\in\CPL(\emptyset)$.  By
    the Theory Lemma, $M^\down*_M(\psi\land\varphi)=\Lang_\CPL$.  If
    $\lnot\psi\notin T$, then it follows by the Theory Lemma that
    $M^\down*_M\psi=T+\psi$; however, since
    $\varphi\land\psi\in(T+\psi)+\varphi$, it follows from the
    assumption of the case that $(T+\psi)+\varphi=\Lang_\CPL$, which
    implies the result.  And if $\lnot\psi\in T-\CPL(\emptyset)$, then
    it follows by the Theory Lemma that $M^\down*_M\psi=\CPL(\psi)$;
    however, since $\varphi\land\psi\in\CPL(\psi)+\varphi$, it follows
    from the assumption of the case that
    $\CPL(\psi)+\varphi=\Lang_\CPL$, which implies the result.
    Finally, if $\lnot\psi\in\CPL(\emptyset)$, then it follows by the
    Theory Lemma that $M^\down*_M\psi=\Lang_\CPL$ and the result
    follows because $\Lang_\CPL+\varphi=\Lang_\CPL$.

  \item Subexpansion: if $\lnot \varphi\notin\CPL(M^\down*_M\psi)$, then
    $(M^\down*_M\psi)+\varphi\subseteq M^\down*_M(\psi\land\varphi)$.
    
    Suppose $\lnot \varphi\notin\CPL(M^\down*_M\psi)$. We consider a
    few cases.

    \textbf{Case:} $\lnot\psi\notin T$. By the Theory Lemma,
    $\CPL(M^\down*_M\psi)=\CPL(T+\psi)=T+\psi$.  If
    $\lnot(\psi\land\varphi)\notin T$, then it follows by the Theory
    Lemma that $M^\down*_M(\psi\land\varphi)=T+(\psi\land\varphi)$,
    from which we obtain the result because
    \[
    (T+\psi)+\varphi= \CPL(\CPL(T\cup\{\psi\})\cup\{\varphi\})=
    \CPL(T\cup\{\psi\}\cup\{\varphi\})=
    \CPL(T\cup\{\psi\land\varphi\})= T+(\psi\land\varphi)\enspace.
    \]
    And if $\lnot(\psi\land\varphi)\in T-\CPL(\emptyset)$, then we
    have $(T+\psi)+\varphi=\Lang_\CPL$ by the definition of $+$;
    however, we originally assumed that
    $\lnot\varphi\notin\CPL(M^\down*_M\psi)=T+\psi$, which implies
    $(T+\psi)+\varphi\neq\Lang_\CPL$, a contradiction that allows us
    to conclude that the hypothesis
    $\lnot(\psi\land\varphi)\in T-\CPL(\emptyset)$ cannot obtain under
    the assumption of this case.  Finally, if
    $\lnot(\psi\land\varphi)\in\CPL(\emptyset)$, then it follows by
    the Theory Lemma that $M^\down*_M(\psi\land\varphi)=\Lang_\CDL$,
    which trivially implies the result.

    \textbf{Case:} $\lnot\psi\in T-\CPL(\emptyset)$. By the Theory
    Lemma, $\CPL(M^\down*_M\psi)=\CPL(\CPL(\psi))=\CPL(\psi)$.  Since
    $T$ is a theory, it follows from the assumption of this case that
    $\lnot(\psi\land\varphi)\in T-\CPL(\emptyset)$. Applying the
    Theory Lemma,
    $M^\down*_M(\psi\land\varphi)=\CPL(\psi\land\varphi)$.  But then
    the result follows because
    \[
    \CPL(\psi)+\varphi=
    \CPL(\CPL(\psi)\cup\{\varphi\})=
    \CPL(\{\psi\}\cup\{\varphi\})=
    \CPL(\psi\land\varphi)
    \enspace.
    \]

    \textbf{Case:} $\lnot\psi\in\CPL(\emptyset)$. By the Theory Lemma
    $\CPL(M^\down*_M\psi)=\CPL(\Lang_\CPL)=\Lang_\CPL$.  It follows
    from the assumption of this case that
    $\lnot(\psi\land\varphi)\in\CPL(\emptyset)$.  Applying the Theory
    Lemma, $M^\down*_M\psi=\Lang_\CPL$. The result follows.
  \end{itemize}
  Conclusion: $M$ is a Grove system for $T$.

  For \eqref{i:*=>AGM}, if $T$ is inconsistent, then since $M_*$ is a
  system of spheres, it follows that $T*(-)=M_*^\down*_M(-)$ satisfies
  the AGM revision postulates.  And if $T$ is consistent, then since
  $T*\psi=T*_{M_T}\psi$ and $M_T$ is a Grove system for $T$, we have
  $M^\down=T$ and that $M$ is a system of spheres; therefore,
  $T*_{M_T}(-)=T*(-)$ satisfies the AGM revision
  postulates. Conclusion: $*$ is an AGM revision operator.
\end{proof}

\subsection{Results for \texorpdfstring{$\JCDL$}{JCDL}}

\begin{proof}[Proof of Lemma~\ref{lemma:elim-t-necess}]
  For this argument, a \emph{derivation} is a $\JCDL$-derivation.  We
  wish to establish the following result:
  \begin{equation}
    \pi\vdash_\JCDL^n\varphi
    \quad\Rightarrow\quad
    \exists \pi_*\supseteq\pi,\;\;
    \pi_*\vdash_\JCDL^0\varphi\enspace.
    \label{eq:N-elim}
  \end{equation}
  The proof of \eqref{eq:N-elim}, an adaptation of
  \cite[Lemma~4.6]{BalRenSme14:APAL} to the present setting, is by
  induction on $n$ with a sub-induction on the length $|\pi|$ of
  $\pi$.
  \begin{itemize}
  \item \textsl{Induction base:} $n=0$. Take $\pi_*=\pi$.

  \item \textsl{Induction step:} assume the result holds for all $m<n$
    (this is the ``induction hypothesis''), and prove it holds for $n$
    by a sub-induction on $|\pi|$.

    \medskip\textsl{Sub-induction base:} $|\pi|=1$. Since $\pi$
    contains a single line, $\varphi$ is an axiom and therefore
    $\pi\vdash_\JCDL^0\varphi$. Take $\pi_*=\pi$.

    \medskip\textsl{Sub-induction step:} assume the result holds for
    all derivations $\pi'$ having $|\pi'|<|\pi|$ (this is the
    ``sub-induction hypothesis''), and prove it holds for $\pi$.
    Proceeding, assume $\pi\vdash_\JCDL^n\varphi$.  If
    $\pi\vdash_\JCDL^0\varphi$, we are done. So let us assume further
    that $\pi\nvdash_\JCDL^0\varphi$. Therefore
    \[
    \pi = \theta_1,\dots,\theta_{|\pi|-1},\varphi
    \]
    contains at least one troublesome necessitation.  It follows that
    there exists a shortest prefix $\pi'$ of $\pi$ whose last line is
    a troublesome necessitation: we have
    \[
    \pi' = \theta_1,\dots,\theta_{|\pi'|-1},\dot B^\delta\theta_m
    \]
    with $|\pi'|$ the minimum value such that $\dot B^\delta\theta_m$
    neither is not a possibly necessitated axiom nor follows from a
    previous line by (MP).  We consider two cases.

    \medskip Case: $|\pi'|<|\pi|$.  By the sub-induction hypothesis,
    there exists $\pi'_*\supseteq\pi'$ such that
    $\pi'_*\vdash_\JCDL^0\dot B^\delta\theta_m$. Let $\sigma$ be the
    suffix of $\pi$ such that $\pi'\sigma=\pi$, where we have denoted
    sequence concatenation by juxtaposition.  Since $\pi'_*$ is a
    derivation and $\pi'_*\supseteq\pi'$, it follows that
    $\pi'_*\sigma$ is a derivation and this derivation has the same
    last line as $\pi$.  Therefore, since we chose $\pi'$ as the
    shortest prefix $\pi'$ of $\pi$ whose last line is a troublesome
    necessitation, we have $\pi'_*\sigma\vdash_\JCDL^{n-1}\varphi$.
    Applying the induction hypothesis, it follows that there exists a
    derivation $\pi_*\supseteq\pi'_*\sigma\supseteq\pi'\sigma=\pi$
    such that $\pi_*\vdash_\JCDL^0\varphi$.

    \medskip Case: $\pi'=\pi$.  Hence $\varphi=\dot B^\delta\theta_m$.
    Since the shortest prefix of $\pi$ whose last line is a
    troublesome necessitation is $\pi$ itself, it follows that
    $\pi\vdash_\JCDL^1\dot B^\delta\theta_m$. But then $\theta_m$ is not
    a troublesome necessitation. Moreover, $\theta_m$ is not a
    possibly necessitated axiom (since if it were,
    $\dot B^\delta\theta_m$ would not be a troublesome necessitation,
    contrary to our assumption implying that it is).  So $\theta_m$
    must follow by way of (MN) from lines $\theta_n\to\theta_m$ and
    $\theta_n$ appearing earlier in $\pi$ than line $m$.  That is,
    \[
    \pi = \begin{array}[t]{l}
            \sigma,\theta_m,\tau,X\theta_m \\[.3em]
            \text{with both $\theta_n\to\theta_m$ and $\theta_n$ in $\sigma$}
          \end{array}
    \]
    where $\sigma$ and $\tau$ denote sequences of formulas (and we
    note that $\tau$ may be empty). By (eMN), we have derivations
    \begin{align*}
      \pi^1 &= \sigma,\dot B^\delta\theta_n 
      & \text{with }|\pi^1|\leq m<|\pi| 
      \\
      \pi^2 &= \sigma,\dot B^\delta(\theta_n\to\theta_m) 
      & \text{with }|\pi^2|\leq m<|\pi|
    \end{align*}
    such that $\pi^1\vdash_\JCDL^1\dot B^\delta\theta_n$ and
    $\pi^2\vdash_\JCDL^1\dot B^\delta(\theta_n\to\theta_m)$.  Applying
    the sub-induction hypothesis to $\pi^1$ and to $\pi^2$, there
    exist $\pi^1_*\supseteq\pi^1$ and $\pi^2_*\supseteq\pi^2$ such
    that $\pi^1_*\vdash_\JCDL^0\dot B^\delta\theta_n$ and
    $\pi^2_*\vdash_\JCDL^0\dot B^\delta(\theta_n\to\theta_m)$.
    Recalling the abbreviation
    $\dot B^\delta\chi=c_\chi\colu\delta\chi$ and making use of an
    instance of (eK), the sequence
    \begin{align*}
      \pi_* =
      &\quad
        \pi^1_*,\pi^2_*,\theta_m,\tau,
      \\
      &\quad
        \dot B^\delta(\theta_n\to\theta_m)\to
        (\dot B^\delta\theta_n\to
        (c_{\theta_n\to\theta_m}\cdot{}c_{\theta_n})\colu\delta\theta_m),
      \\
      &\quad
        \dot B^\delta\theta_n\to
        (c_{\theta_n\to\theta_m}\cdot{}c_{\theta_n})\colu\delta\theta_m,
      \\
      &\quad
        (c_{\theta_n\to\theta_m}\cdot{}c_{\theta_n})\colu\delta\theta_m
    \end{align*}
    is a derivation satisfying $\pi_*\supseteq\pi$
    and $\pi_*\vdash_\JCDL^0\varphi$.
  \end{itemize}
  This completes the proof that \eqref{eq:N-elim} holds.
\end{proof}

\begin{proof}[Proof of Theorem~\ref{theorem:internalization}]
  In light of Lemma~\ref{lemma:elim-t-necess}, to prove the statement
  of the present theorem, it suffices to prove the following: if
  $\vdash_\JCDL^0\varphi$, then there exists a logical term $t$ such
  that $\vdash_\JCDL^0 t\colu\delta\varphi$.  This we prove by
  induction on the length of derivation.
  \begin{itemize}
  \item Induction base and induction step for (eMN): $\varphi$ is a
    possibly necessitated axiom.  It follows that
    $\vdash_\JCDL^0 c_\varphi\colu\delta\varphi$ for the logical term
    $c_\varphi$.

  \item Induction step for (MP): we assume
    $\vdash_\JCDL^0\varphi\to\psi$ and $\vdash_\JCDL^0\varphi$ along
    with the following ``induction hypothesis'': there exist logical
    terms $t$ and $s$ such that
    $\vdash_\JCDL^0 t\colu\delta(\varphi\to\psi)$ and
    $\vdash_\JCDL^0 s\colu\delta\varphi$. But then it follows by the
    induction hypothesis, (eK), and two applications of (MP) that
    $\vdash_\JCDL^0 (t\cdot s)\colu\delta\psi$ for the logical term
    $t\cdot s$.
  \end{itemize}
  The result follows.
\end{proof}

\begin{proof}[Proof of Theorem~\ref{theorem:realization}]
  For Item~\ref{i:projection}, we prove by induction on the length of
  derivation in $\JCDL$ from hypotheses $\Gamma$ that
  \[
  \Gamma\vdash_\JCDL\varphi \quad\text{implies}\quad
  \Gamma^\circ\vdash_\CDL\varphi^\circ\enspace.
  \]
  \begin{itemize}
  \item Induction base for hypotheses: if $\varphi\in\Gamma$, then
    $\varphi^\circ\in\Gamma^\circ$.  So
    $\Gamma^\circ\vdash_\CDL\varphi^\circ$.

  \item Induction base for axioms: $\varphi$ is an axiom. But for each
    axiom scheme (e$X$) of $\JCDL$ for which there is a ``matching''
    axiom scheme ($X$) of $\CDL$, the forgetful projection of an
    instance of the $\JCDL$ axiom is an instance of the matching
    $\CDL$ scheme. Regarding the three remaining $\JCDL$ schemes
    (eSum), (eCert), and (eA) that have no matching $\CDL$ scheme:
    each of (eSum) and (eCert) is mapped to an instance of the
    $\CDL$-theorem $p\imp p$, and (eA) is mapped to an instance of the
    $\CDL$-theorem $q\to(p\to p)$.  Conclusion:
    $\Gamma^\circ\vdash_\CDL\varphi^\circ$.

  \item Induction step for (MP): we assume the result holds for
    $\JCDL$-theorems $\varphi\to\psi$ and $\varphi$ derived from
    hypotheses $\Gamma$, and we prove that
    $\Gamma^\circ\vdash_\CDL\psi^\circ$. By our assumption,
    $\varphi^\circ\to\psi^\circ$ and $\varphi^\circ$ are
    $\CDL$-theorems derived from hypotheses $\Gamma^\circ$. Applying
    (MP), $\psi^\circ$ as a $\CDL$-theorem derived from hypotheses
    $\Gamma^\circ$.

  \item Induction step for (eMN): we assume the result holds for the
    $\JCDL$-theorem $\varphi$ derived from hypotheses $\Gamma$, and we
    prove that the result holds for the $\JCDL$-theorem
    $\dot B^\psi\varphi$ derived from hypotheses $\Gamma$; that is, we
    prove that $\Gamma^\circ\vdash_\CDL (\dot B^\psi\varphi)^\circ$.
    By our assumption, $\varphi^\circ$ is a $\CDL$-theorem derived
    from hypotheses $\Gamma^\circ$.  Applying (MN) and the definition
    of the forgetful projection,
    $B^{\psi^\circ}\varphi^\circ=(\dot B^\psi\varphi)^\circ$ is a
    $\CDL$-theorem derived from hypotheses in $\Gamma^\circ$.
  \end{itemize}
  This completes the proof of Item~\ref{i:projection}. For
  Item~\ref{i:t-realization}, we by induction on the length of
  derivation in $\CDL$ from hypotheses $\Delta$ that
  \[
  \Delta\vdash_\CDL\psi \quad\text{implies}\quad
  \Delta^t\vdash_\JCDL\psi^t\enspace.
  \]
  \begin{itemize}
  \item Induction base for hypotheses: If $\psi\in \Delta$, then
    $\psi^t\in\Delta^t$. So
    $\Delta^t\vdash_\CDL\psi^t$.

  \item Induction base for axiom scheme: $\psi$ is an axiom scheme of
    $\CDL$, so we consider each possibility.
    \begin{itemize}
    \item $\text{(CL)}^t$ is an instance of (CL) and therefore a
      $\JCDL$-derivable from hypotheses $\Delta^t$.

    \item $\text{(K)}^t$ is the scheme
      $\dot B^{\gamma^t}(\varphi_1^t\to\varphi_2^t)\to (\dot
      B^{\gamma^t}\varphi_1^t\to \dot B^{\gamma^t}\varphi_2^t)$.
      This has the form
      \[
      \dot B^\delta(\chi_1\to\chi_2)\to (\dot B^\delta\chi_1\to \dot
      B^\delta\chi_2)\enspace.
      \]
      Applying (eK), we obtain
      \[
      \Delta^t\vdash_\JCDL\dot B^\delta(\chi_1\to\chi_2)\to (\dot
      B^\delta\chi_1\to (c_{\chi_1\to\chi_2}\cdot
      c_{\chi_1})\colu\delta\chi_2)\enspace.
      \]
      Applying (eCert) and classical reasoning,
      \[
      \Delta^t\vdash_\JCDL\dot B^\delta(\chi_1\to\chi_2)\to (\dot
      B^\delta\chi_1\to \dot B^\delta\chi_2)\enspace.
      \]
      That is, $\Delta^t\vdash_\JCDL \text{(K)}^t$.

    \item For
      $X\in\{\text{Succ},\text{RM},\text{Inc},\text{Comm},\text{PI},\text{NI}\}$:
      $(X)^t$ is a schematic instance of $(\text{e}X)$.

    \item Similar to the argument for $\text{(K)}^t$:
      \begin{itemize}
      \item $\text{(KM)}^t$ follows by (eKM) and (eCert), and
        
      \item $\text{(WCon)}^t$ follows by (eWCon) and (eCert).
      \end{itemize}
    \end{itemize}
    
  \item Induction step for (MP): we assume the result holds for the
    $\CDL$-theorems $\varphi_1\to\varphi_2$ and $\varphi_1$ derived
    from hypotheses $\Delta$, and we prove the result holds for the
    $\CDL$-theorem $\varphi_2$ derived from hypothesis $\Delta$.  By
    our assumption,
    \[
    (\varphi_1\to\varphi_2)^t=\varphi_1^t\to\varphi_2^t
    \quad\text{and}\quad \varphi_1^t
    \]
    is a $\JCDL$-theorem derivable from hypotheses
    $\Delta^t$. So it follows by (MN) that $\varphi_2^t$
    is a $\JCDL$-theorem derivable from hypotheses $\Delta^t$ as
    well.

  \item Induction step for (MN): we assume the result holds for the
    $\CDL$-theorem $\varphi$ derived from hypotheses $\Delta$, and we
    prove the result holds for the $\CDL$-theorem $B^\gamma\varphi$
    derived from hypotheses $\Delta$; that is, we prove that
    $\Delta^t\vdash_\JCDL (B^\gamma\varphi)^t$.  By our
    assumption, the formula $\varphi^t$ is a $\JCDL$-theorem
    derived from hypotheses $\Delta^t$, but then it follows by
    (eMN) that
    $\dot B^{\gamma^t}\varphi^t=(\dot
    B^\gamma\varphi)^t$.
    is a $\JCDL$-theorem derived from hypotheses $\Delta^t$.
  \end{itemize}
  This completes the proof of Item~\ref{i:t-realization}.
\end{proof}

\begin{proof}[Proof of Theorem~\ref{theorem:JCDL-determinacy}]
  This proof uses much of the work from the proof of
  Theorem~\ref{theorem:CDL-completeness} (the ``old proof'').  In
  utilizing portions of the argument from the old proof in the present
  argument (the ``new proof''), we adopt the following conventions
  (the ``$\JCDL$-conventions''):
  \begin{itemize}
  \item the language (and set of formulas) is assumed to be
    $\Lang_\JCDL$;

  \item occurrences of ``$B\,$'' from the old proof are replaced by
    occurrences of ``$\dot B\,$'';
   
  \item truth or validity for $\CDL$ is replaced by truth or validity
    for $\JCDL$;

  \item use of Theorem~\ref{theorem:CDL-wf} or \ref{theorem:CDL-wo} is
    replaced by use of Theorem~\ref{theorem:JCDL-wf-B};

  \item use of Theorem~\ref{theorem:CDL-lwo} is replaced by
    Theorem~\ref{theorem:JCDL-lwo};

  \item derivability in $\CDL$ is replaced by derivability in $\JCDL$

  \item use of the non-subscripted turnstile $\vdash$ denotes
    $\vdash_\JCDL$;

  \item use of Theorem~\ref{theorem:CDL-theorems} is replaced by
    Theorem~\ref{theorem:JCDL-theorems}; and

  \item use of a derivable principle $(X)$ of $\CDL$ (perhaps by tacit
    use of Theorem~\ref{theorem:CDL-theorems}) is replaced by use of
    the corresponding derivable principle $(\text{e}X)$ of $\JCDL$
    (with corresponding tacit use of
    Theorem~\ref{theorem:JCDL-theorems} when appropriate).
  \end{itemize}
  Having established the above $\JCDL$-conventions and the ``old/new
  proof'' terminology, we proceed.

  For soundness, we proceed by induction on the length of
  derivation. In the induction base, we must show that each axiom
  scheme is valid. (CL) is straightforward, so we proceed with the
  remaining schemes.  Let $(M,w)$ be an arbitrary well-ordered pointed
  Fitting model. We make tacit use of
  Theorems~\ref{theorem:wf-smooth}\eqref{i:wo} and
  \ref{theorem:JCDL-wo}.
  
  \bigskip
  \noindent (eCert) is valid
  $\models t\colu\psi\varphi\to\dot B^\psi\varphi$.
  \begin{pquote}
    Assume $(M,w)$ satisfies $t\colu\psi\varphi$.  Then
    $w\in A(t,\varphi)$ and $\min\sem{\psi}\subseteq\sem{\varphi}$.
    Applying Theorem~\ref{theorem:JCDL-wf-B}, $(M,w)$ satisfies
    $B^\psi\varphi$.
  \end{pquote}
  
  \bigskip\noindent (eK) is valid:
  $\models t\colu\psi(\varphi_1\imp\varphi_2)\imp
  (s\colu\psi\varphi_1\imp(t\cdot s)\pcolu\psi\varphi_2)$.
  \begin{pquote}
    Assume $(M,w)$ satisfies $t\colu\psi(\varphi_1\to\varphi_2)$ and
    $s\colu\psi\varphi_1$. Then
    $w\in A(t,\varphi_1\to\varphi_2)\cap A(s,\varphi_1)$,
    $\min\sem{\psi}\subseteq\sem{\varphi_1\to\varphi_2}$, and
    $\min\sem{\psi}\subseteq\sem{\varphi_1}$.  Hence
    $\min\sem{\psi}\subseteq\sem{\varphi_2}$ and, by Application
    (Definition~\ref{definition:F-model}),
    $w\in A(t\cdot s,\varphi_2)$.  Conclusion: $(M,w)$ satisfies
    $(t\cdot s)\pcolu\psi\varphi_2$.
    
  \end{pquote}
  
  \bigskip\noindent (eSum) is valid:
  $\models (t\colu\psi\varphi\lor s\colu\psi\varphi)\imp
  (t+s)\colu\psi\varphi$.
  \begin{pquote}
    Assume $(M,w)$ satisfies
    $t\colu\psi\varphi\lor s\colu\psi\varphi$.  Then
    $w\in A(t,\varphi)\cup A(s,\varphi)$ and
    $\min\sem{\psi}\subseteq\sem{\varphi}$.  By Sum
    (Definition~\ref{definition:F-model}), $w\in A(t+s,\varphi)$.
    Conclusion: $(M,w)$ satisfies $(t+s)\colu\psi\varphi$.
  \end{pquote}

  \bigskip\noindent (eSucc) is valid: $\models \dot B^\psi\psi$.
  \begin{pquote}
    We have $\min\sem{\psi}\subseteq\sem{\psi}$.  Applying
    Theorem~\ref{theorem:JCDL-wf-B}, $(M,w)$ satisfies
    $\dot B^\psi\psi$.
  \end{pquote}

  \bigskip\noindent (eKM) is valid:
  $\models t\colu\psi\bot\to t\colu{\psi\land\varphi}\bot$.
  \begin{pquote}
    Assume $(M,w)$ satisfies $t\colu\psi\bot$.  Then $w\in A(t,\bot)$
    and $\min\sem{\psi}\subseteq\sem{\bot}$. It follows from the
    latter by the old proof that
    $\min\sem{\psi\land\varphi}\subseteq\sem{\bot}$. Conclusion:
    $(M,w)$ satisfies $t\colu{\psi\land\varphi}\bot$.
  \end{pquote}

  \bigskip\noindent (eRM) is valid:
  $\models \lnot\dot B^\psi\lnot\varphi\to (t\colu\psi\chi\to
  t\colu{\psi\land\varphi}\chi)$.
  \begin{pquote}
    Assume $(M,w)$ satisfies $\lnot\dot B^\psi\lnot\varphi$ and
    $t\colu\psi\chi$. 
    
    Since
    $\lnot\dot B^\psi\lnot\varphi=\lnot
    c_{\lnot\varphi}\colu\psi\lnot\varphi$,
    it follows by Certification (Definition~\ref{definition:F-model})
    that $\min\sem{\psi}\nsubseteq\sem{\varphi}$.  And it follows from
    $t\colu\psi\chi$ that $w\in A(t,\chi)$ and
    $\min\sem{\psi}\subseteq\sem{\chi}$.  By the old proof, we have by
    $\min\sem{\psi}\nsubseteq\sem{\varphi}$ and
    $\min\sem{\psi}\subseteq\sem{\chi}$ that
    $\min\sem{\psi\land\varphi}\subseteq\sem{\chi}$.  Since
    $w\in A(t,\chi)$, we conclude that $(M,w)$ satisfies
    $t\colu{\psi\land\varphi}\chi$.
  \end{pquote}

  \bigskip\noindent (eInc) is valid:
  $\models t\colu{\psi\land\varphi}\chi\to \dot
  B^\psi(\varphi\to\chi)$.
  \begin{pquote}
    Suppose $(M,w)$ satisfies $t\colu{\psi\land\varphi}\chi$.  Then
    $w\in A(t,\chi)$ and
    $\min\sem{\psi\land\varphi}\subseteq\sem{\chi}$.  It follows from
    the latter by the old proof that
    $\min\sem{\psi}\subseteq\sem{\varphi\to\chi}$.  Conclusion:
    $(M,w)$ satisfies $\dot B^\psi(\varphi\to\chi)$.
  \end{pquote}

  \bigskip\noindent (eComm) is valid:
  $\models t\colu{\psi\land\varphi}\chi\to
  t\colu{\varphi\land\psi}\chi$.
  \begin{pquote}
    Suppose $(M,w)$ satisfies $t\colu{\psi\land\varphi}\chi$.  Then
    $w\in A(t,\chi)$ and
    $\min\sem{\psi\land\varphi}\subseteq\sem{\chi}$.  Hence
    $\min\sem{\varphi\land\psi}\subseteq\sem{\chi}$.  Conclusion:
    $(M,w)$ satisfies $t\colu{\varphi\land\psi}\chi$.
  \end{pquote}

  \bigskip\noindent (ePI) is valid:
  $\models t\colu\psi\chi\to\dot B^\varphi(t\colu\psi\chi)$.
  \begin{pquote}
    Suppose $(M,w)$ satisfies $t\colu\psi\chi$.  Then $w\in A(t,\chi)$
    and $\min\sem{\psi}\subseteq\sem{\chi}$.  The latter implies that
    $(M,v)$ satisfies $\dot B^\psi\chi$ for any given $v\in W$.  Since
    $M$ is well-ordered, we have $v\in\cc(w)$ for each $v\in W$, and
    therefore it follows from $w\in A(t,\chi)$ by Admissibility
    Indefeasibility (Definition~\ref{definition:F-model}) that
    $A(t,\chi)=W$.  But then for each $v\in W$, we have
    $v\in A(t,\chi)$ and $M,v\models\dot B^\psi\chi$.  Applying
    Theorem~\ref{theorem:JCDL-truth-B}, we have for each $v\in W$ that
    $M,v\models t\colu\psi\chi$.  Therefore, $\sem{t\colu\psi\chi}=W$,
    from which it follows that
    $\min\sem{\varphi}\subseteq\sem{t\colu\psi\chi}$.  Applying
    Theorem~\ref{theorem:JCDL-wf-B}, we conclude that $(M,w)$
    satisfies $\dot B^\varphi(t\colu\psi\chi)$.
  \end{pquote}

  \bigskip\noindent (eNI) is valid:
  $\models \lnot t\colu\psi\chi\to\dot B^\varphi(\lnot
  t\colu\psi\chi)$.
  \begin{pquote}
    Suppose $(M,w)$ satisfies $\lnot t\colu\psi\chi$. Then
    $w\notin A(t,\varphi)$ or $\min\sem{\psi}\nsubseteq\sem{\chi}$.
    
    Case: $w\notin A(t,\varphi)$.  Since $M$ is well-ordered, we have
    $v\in\cc(w)$ for each $v\in W$, and so it follows from
    $w\notin A(t,\varphi)$ by Admissibility Indefeasibility
    (Definition~\ref{definition:F-model}) that
    $A(t,\varphi)=\emptyset$.  Therefore,
    $\sem{\lnot t\colu\psi\chi}=W$, from which it follows that
    $\min\sem{\varphi}\subseteq\sem{\lnot t\colu\psi\chi}$.  Applying
    Theorem~\ref{theorem:JCDL-wf-B}, we conclude that $(M,w)$
    satisfies $\dot B^\varphi(\lnot t\colu\psi\chi)$.

    Case: $\min\sem{\psi}\nsubseteq\sem{\chi}$.  It follows that
    $M,v\models\lnot t\colu\psi\chi$ for each $v\in\cc(w)=W$.
    Therefore, $\sem{\lnot t\colu\psi\chi}=W$, from which it follows
    that $\min\sem{\varphi}\subseteq\sem{\lnot t\colu\psi\chi}$.
    Applying Theorem~\ref{theorem:JCDL-wf-B}, we conclude that $(M,w)$
    satisfies $\dot B^\varphi(\lnot t\colu\psi\chi)$.
  \end{pquote}

  \bigskip\noindent (eWCon) is valid: $\models t\colu\psi\bot\to\lnot\psi$.
  \begin{pquote}
    Suppose $(M,w)$ satisfies $t\colu\psi\bot$.  By the old proof,
    $\sem{\lnot\psi}=W$. So $(M,w)$ satisfies $\lnot\psi$.
  \end{pquote}

  \bigskip\noindent (eA) is valid:
  $\models t\colu\psi\varphi\to(\dot B^\chi\varphi\to
  t\colu\chi\varphi)$.
  \begin{pquote}
    Suppose $(M,w)$ satisfies $t\colu\psi\varphi$ and
    $\dot B^\chi\varphi$.  Then $w\in A(t,\varphi)$ and
    $\min\sem{\chi}\subseteq\sem{\varphi}$.  Conclusion: $(M,w)$
    satisfies $t\colu\chi\varphi$.
  \end{pquote}

  \bigskip\noindent This completes the induction base. For the
  induction step, we must show that validity is preserved under the
  rules of (MP) and (eMN).  The argument for (MP) is standard, so let
  us focus on (eMN).  We assume $\models\varphi$ for the
  $\JCDL$-derivable $\varphi$ (this is the ``induction hypothesis''),
  and we prove that $\models\dot B^\psi\varphi$. Proceeding, since
  $\models\varphi$, it follows that $\sem{\varphi}=W$ and hence that
  $\min\sem{\psi}\subseteq\sem{\varphi}$.  By Certification
  (Definition~\ref{definition:F-model}, $A(c_\varphi,\varphi)=W$.  But
  then $(M,w)$ satisfies
  $c_\varphi\colu\psi\varphi=\dot B^\psi\varphi$.  Soundness has been
  proved.

  Since $\JCDL$ is sound with respect to the class of well-ordered
  Fitting models we note that $\JCDL$ is consistent (i.e.,
  $\nvdash_\JCDL\bot$).  In particular, take any pointed Fitting model
  $(M,w)$ containing only the single world $w$.  It is simple to
  construct such a model: take
  \[
  M\coloneqq(\{w\},\{(w,w)\},\{(w,\emptyset)\},A)\enspace,
  \]
  where $A$ is the ``total'' admissibility function defined by setting
  $A(t,\varphi)\coloneqq W$ for all
  $(t,\varphi)\in\Term_\JCDL\times\Lang_\JCDL$. It is obvious that the
  requisite properties from Definition~\ref{definition:F-model}
  obtain. Since there is only one world, $M$ is well-ordered.
  Further, by soundness, we have that $\vdash_\JCDL\varphi$ implies
  $M,w\models\varphi$. Therefore, since $M,w\not\models\bot$ by
  Definition~\ref{definition:JCDL-truth}, it follows that
  $\nvdash_\JCDL\bot$. That is, $\JCDL$ is consistent.  We make use of
  this fact tacitly in what follows.

  For completeness, take a formula $\theta$ such that
  $\nvdash\lnot\theta$.  Define consistency, inconsistency, maximal
  consistency in (``maxcons in'') a set of formulas, the set
  $\sub(\varphi)$ of subformulas of $\varphi$ (including $\varphi$
  itself), the set $\sub(S)$ containing all subformulas of each
  formula in the set $S$ (including the formulas themselves), and the
  Boolean closure ${\oplus}S$ as in the old proof (but of course using
  the $\JCDL$-conventions).  For a set $S$ of formulas, we define:
  \begin{align*}
    {\vec B}S
    &\coloneqq
      S\cup\{\dot B^\psi\varphi\mid t\colu\psi\varphi\in S\}
      \enspace,
    \\
    {\pm}S 
    &\coloneqq 
      S\cup\{\lnot\varphi\mid\varphi\in S\}
      \enspace,
    \\
    \dot B_0S
    &\coloneqq S\enspace,
    \\
    \dot B_{i+1}S
    &\coloneqq
      {\pm}\{
      \dot B^\psi\varphi \mid 
      \psi\in S \text{ and } 
      \varphi\in \dot B_iS
      \}
      \enspace,
    \\
    \dot B_\omega S
    &\coloneqq\textstyle\bigcup_{0<i<\omega}\dot B_iS
      \enspace,
    \\
    C_0
    &\coloneqq{\pm}\vec{B}\,\sub(\{\theta,\bot,\top\})
      \enspace,
    \\
    C_1
    &\coloneqq{\oplus}C_0
      \enspace,
    \\
    \dot B
    &\coloneqq\dot B_\omega C_1\enspace,
    \\
    C 
    &\coloneqq C_1\cup\dot B
      \enspace,
    \\
    T_0
    &\coloneqq 
      {\pm}\{ t\colu\psi\varphi\in\Lang_\JCDL \mid
      t\colu\psi\varphi\in C_0\}
      \enspace.
  \end{align*}
  Key differences from the old proof:
  \begin{itemize}
  \item the new operator $\vec{B}S$ adds the ``certified version''
    $\dot B^\psi\varphi=c_\varphi\colu\psi\varphi$ of each formula
    $t\colu\psi\varphi$ in $S$,

  \item $C_0$ has been changed by adding the operator $\vec{B}$,

  \item every ``$B\,$'' in the old proof has been replaced by
    ``$\dot B\,$'', 

  \item we use the $\JCDL$-conventions (e.g., we are working in the
    language of $\JCDL$), and
    
  \item the set $T_0$ (used later) is new.
  \end{itemize}
  We then define:
  \begin{align*}
    W
    &\coloneqq\{x\subseteq C\mid x \text{ is maxcons in } C\}
      \enspace,
    \\
    \bar x
    &\textstyle\coloneqq
      \bigwedge(x\cap C_0) \text{ for } x\in W
      \enspace,
    \\
    x^\psi
    &\textstyle\coloneqq
      \{\varphi\mid\dot B^\psi\varphi\in x\}
      \text{ for } x\in W \text{ and } \psi\in C_1
      \enspace,
    \\
    {\leq }
    &\coloneqq
      \{(x,y)\in W\times W\mid
      \exists \psi\in(x\cap y\cap C_1),\,
      y^\psi\subseteq x
      \}
      \enspace,
    \\
    V(x)
    &\coloneqq\Prop\cap x \text{ for } x\in W
      \enspace,
    \\
    \hat x
    &\coloneqq
      \textstyle\bigwedge(x\cap T_0)
      \text{ for } x\in W \enspace,
    \\
    A(t,\varphi)
    &\coloneqq
      \{
      x\in W\mid
      \exists\psi\in\Lang_\JCDL,\,
      {}\vdash\hat x\to t\colu\psi\varphi
      \}
      \enspace,
    \\
    M
    &\coloneqq(W,\leq ,V,A)
      \enspace.
  \end{align*}
  Key differences from the old proof:
  \begin{itemize}
  \item every ``$B\,$'' in the old proof has been replaced by
    ``$\dot B\,$'',

  \item the mapping taking a world $x\in W$ to a formula $\hat x$ is
    new,

  \item the function $A$ of type
    $(\Term_\JCDL\times\Lang_\JCDL)\to\wp(W)$ is new, and

  \item $M$ has been expanded to contain $A$.
  \end{itemize} 
  
  The arguments in the old proof (modulo the $\JCDL$-conventions) are
  used to prove $W$ is finite and nonempty, $\leq$ is reflexive and
  transitive, $\leq$ is total on each connected component, and $\leq$
  is well-founded.  We prove the following result particular to the
  new proof:
  \begin{equation}
    y\in\cc(x)
    \quad\Rightarrow \quad
    (\forall t\colu\psi\varphi\in C_1,\;\;
    t\colu\psi\varphi\in x \;\;\text{iff}\;\;
    t\colu\psi\varphi\in y)
    \enspace.
    \label{eq:JCDL<=}
  \end{equation}
  Proceeding, suppose $y\in\cc(x)$. Since $\leq$ is total on each
  connected component, we may assume without loss of generality that
  $x\leq y$.  (The case where $y\leq x$ is argued similarly.) Now
  $x\leq y$ implies there exists $\delta\in(x\cap y\cap C_1)$ such
  that $y^\delta\subseteq x$.  For the left-to-right direction:
  suppose $t\colu\psi\varphi\in x\cap C_1$.  If we had
  $t\colu\psi\varphi\notin y$, then it would follow by maximal
  consistency that $\lnot t\colu\psi\varphi\in y$, hence
  $\dot B^\delta(\lnot t\colu\psi\varphi)\in y$ by (eNI), hence
  $\lnot t\colu\psi\varphi\in x$ by $y^\delta\subseteq x$, thereby
  contradicting the consistency of $x$. So it must be the case that
  $t\colu\psi\varphi\in y$ after all.  Now for the right-to-left
  direction: suppose $t\colu\psi\varphi\in y\cap C_1$.  It follows by
  (ePI) that $\dot B^\delta(t\colu\psi\varphi)\in y$.  Since
  $y^\delta\subseteq x$, we obtain $t\colu\psi\varphi\in x$.  So
  \eqref{eq:JCDL<=} indeed obtains.

  In order to conclude that $M$ is a Fitting model, we must prove that
  $A$ satisfies the properties required of an admissibility function
  (Definition~\ref{definition:F-model}). We state and prove these in
  turn.
  \begin{itemize}
  \item Certification: $A(c_\varphi,\varphi)=W$.
    
    Take an arbitrary $x\in W$. We have
    $\vdash\hat x\to\dot B^\varphi\varphi$ by (Succ).  Since
    $\dot B^\varphi\varphi=c_\varphi\colu\varphi\varphi$, it follows
    that $\vdash\hat x\to c_\varphi\colu\varphi\varphi$.  Applying the
    definition of $A$, we obtain $x\in A(c_\varphi,\varphi)$. Since
    $x\in W$ was chosen arbitrarily, we conclude that
    $A(c_\varphi,\varphi)=W$.

  \item Application:
    $A(t,\varphi_1\to\varphi_2)\cap A(s,\varphi_1) \subseteq
    A(t\cdot s,\varphi_2)$.

    Assume $x\in A(t,\varphi_1\to\varphi_2)\cap A(s,\varphi_1)$.
    Applying the definition of $A$, this means there exists
    $\psi_1\in\Lang_\JCDL$ and $\psi_2\in\Lang_\JCDL$ such that
    $\vdash\hat x\to t\colu{\psi_1}(\varphi_1\to\varphi_2)$ and
    $\vdash\hat x\to s\colu{\psi_2}\varphi_1$.  Now we have each of
    $\vdash B^{\varphi_1\land\varphi_2}(\varphi_1\to\varphi_2)$ and
    $\vdash B^{\varphi_1\land\varphi_2}\varphi_1$ by (Succ) and modal
    reasoning.  So it follows that
    \[
    \vdash\hat x\to
    t\colu{\psi_1}(\varphi_1\to\varphi_2)\land
    B^{\varphi_1\land\varphi_2}(\varphi_1\to\varphi_2)\land
    s\colu{\psi_2}\varphi_1\land
    B^{\varphi_1\land\varphi_2}\varphi_1\enspace.
    \]
    Applying (eA), we obtain
    \[
    \vdash\hat x\to
    t\colu{\varphi_1\land\varphi_2}(\varphi_1\to\varphi_2) \land
    s\colu{\varphi_1\land\varphi_2}\varphi_1\enspace,
    \]
    from which it follows by (eK) that
    $\vdash\hat x\to (t\cdot
    s)\colu{\varphi_1\land\varphi_2}\varphi_2$.
    Applying the definition of $A$, it follows that
    $x\in A(t\cdot s,\varphi_2)$.
    
  \item Sum:
    $A(t,\varphi)\cup A(s,\varphi)\subseteq
    A(t+s,\varphi)$.

    Suppose $x\in A(t,\varphi)\cup A(s,\varphi)$. Applying the
    definition of $A$, this means there exists $\psi_t\in\Lang_\JCDL$
    such that $\vdash\hat x\to t\colu{\psi_t}\varphi$ or there exists
    $\psi_s\in\Lang_\JCDL$ such that
    $\vdash\hat x\to s\colu{\psi_s}\varphi$. Applying (eSum), it
    follows that $\vdash\hat x\to(t+s)\colu\psi\varphi$ for some
    $\psi\in\{\psi_t,\psi_s\}$. But then we obtain by the definition
    of $A$ that $x\in A(t+s,\varphi)$.

  \item Admissibility Indefeasibility: if $x\in A(t,\varphi)$ and
    $y\in\cc(x)$, then $y\in A(t,\varphi)$.
    
    Suppose $x\in A(t,\varphi)$ and $y\in\cc(x)$.  It follows from
    $x\in A(t,\varphi)$ by the definition of $A$ that there exists
    $\psi\in\Lang_\JCDL$ such that
    $\vdash\hat x\to t\colu\psi\varphi$.  But $y\in\cc(x)$, and so it
    follows by \eqref{eq:JCDL<=} that $\hat y=\hat x$.  Therefore
    $\vdash\hat y\to t\colu\psi\varphi$, from which it follows by the
    definition of $A$ that $y\in A(t,\varphi)$.
  \end{itemize}
  So $A$ satisfies the properties required of an admissibility
  function.
  
  Since $W$ is nonempty, $A$ satisfies the properties of an
  admissibility function, and $\leq$ is a locally well-ordered, it
  follows that $M$ is a locally well-ordered Fitting model.  The
  arguments in the old proof (modulo the $\JCDL$-conventions) are then
  used to prove that the Consistency Lemma holds and that the
  Minimality Lemma holds.
  
  By the old proof (modulo the $\JCDL$-conventions), all of the
  induction base and step cases of the Truth Lemma hold, except for
  the induction step case for formulas $t\colu\psi\varphi$. What
  remains is to check this remaining case.  Before we do this, note:
  by the definition of $C_1$ as the Boolean closure of $C_0$, it
  follows from $t\colu\psi\varphi\in C_1$ that
  $t\colu\psi\varphi\in C_0$ and therefore $t\colu\psi\varphi\in T_0$.
  We now proceed with the argument.
  \begin{itemize}
  \item Induction step $t\colu\psi\varphi$ (left to right):
    if $x\in W$ and $t\colu\psi\varphi\in x\cap C_1$, then
    $M,x\models t\colu\psi\varphi$.

    Assume $t\colu\psi\varphi\in x\cap C_1$. By (eCert) and the
    definition of $C_1$, it follows that
    $\dot B^\psi\varphi\in x\cap C_1$.  Using the old proof (modulo
    the $\JCDL$-conventions and using
    Theorem~\ref{theorem:JCDL-wf-B} in place of
    Theorem~\ref{theorem:CDL-wf}), we obtain
    $M,x\models\dot B^\psi\varphi$.  Also, since
    $t\colu\psi\varphi\in x\cap C_1$ implies
    $t\colu\psi\varphi\in x\cap T_0$, it follows that
    $\vdash\hat x\to t\colu\psi\varphi$, from which we obtain
    $x\in A(t,\varphi)$ by the definition of $A$.  So since
    $x\in A(t,\varphi)$ and $M,x\models\dot B^\psi\varphi$, it follows
    by Theorem~\ref{theorem:JCDL-truth-B} that
    $M,x\models t\colu\psi\varphi$.

  \item Induction step $t\colu\psi\varphi$ (right to left): if
    $x\in W$, $t\colu\psi\varphi\in C_1$, and
    $M,x\models t\colu\psi\varphi$, then $t\colu\psi\varphi\in x$.
    
    Assume $t\colu\psi\varphi\in C_1$ and
    $M,x\models t\colu\psi\varphi$.  By
    Theorem~\ref{theorem:JCDL-truth-B}, it follows that
    $M,x\models\dot B^\psi\varphi$ and $x\in A(t,\varphi)$.  Using the
    old proof (modulo the $\JCDL$-conventions and using
    Theorem~\ref{theorem:JCDL-wf-B} in place of
    Theorem~\ref{theorem:CDL-wf}), it follows from
    $M,x\models\dot B^\psi\varphi$ that $\dot B^\psi\varphi\in x$.
    Applying the definition of $A$, it follows from
    $x\in A(t,\varphi)$ that there exists $\chi\in\Lang_\JCDL$ such
    that $\vdash\hat x\to t\colu\chi\varphi$.  But then we have
    $\vdash\hat x\to t\colu\chi\varphi \land \dot B^\psi\varphi$.
    Applying (eA), we obtain $\vdash\hat x\to t\colu\psi\varphi$.
    Since $t\colu\psi\varphi\in C_1$ implies
    $t\colu\psi\varphi\in T_0$, it follows from
    $\vdash\hat x\to t\colu\psi\varphi$ by the maximal consistency of
    $x$ in $C\supseteq T_0$ that $t\colu\psi\varphi\in x$.
  \end{itemize}
  This completes the proof of the Truth Lemma.  We then apply the
  argument as in the old proof (modulo the $\JCDL$-conventions and
  using Theorem~\ref{theorem:JCDL-lwo} in place of
  Theorem~\ref{theorem:CDL-lwo}) to conclude that $\JCDL$ is
  complete with respect to the class of well-ordered Fitting models.
\end{proof}


\bibliographystyle{plain}
\bibliography{brs-JCDL}

\begin{thebibliography}{1}

\bibitem{AGM}
C.E. Alchourr\'{o}n, P.~G\"{a}rdenfors, and D.~Makinson.
\newblock On the logic of theory change: partial meet contraction and revision
  functions.
\newblock {\em Journal of Symbolic Logic}, 50:510--530, 1985.

\bibitem{ArtFit12:SEP}
Sergei~[N.] Artemov and Melvin Fitting.
\newblock Justification logic.
\newblock In Edward~N. Zalta, editor, {\em The {S}tanford {E}ncyclopedia of
  {P}hilosophy}, 2012.

\bibitem{BalRenSme14:APAL}
Alexandru Baltag, Bryan Renne, and Sonja Smets.
\newblock The logic of justified belief, explicit knowledge, and conclusive
  evidence.
\newblock {\em Annals of Pure and Applied Logic}, 165(1):49--81, 2014.

\bibitem{BalSme08:LOFT}
Alexandru Baltag and Sonja Smets.
\newblock A qualitative theory of dynamic interactive belief revision.
\newblock In Giacomo Bonanno, Wiebe van~der Hoek, and Michael Wooldridge,
  editors, {\em TLG 3: Logic and the Foundations of Game and Decision Theory
  ({LOFT~7})}, volume~3 of {\em Texts in logic and games}, pages 11--58.
  Amsterdam University Press, 2008.

\bibitem{Boa04:GEB}
Oliver Board.
\newblock Dynamic interactive epistemology.
\newblock {\em Games and Economic Behavior}, 49(1):49--80, 2004.

\bibitem{Gro88:JPL}
Adam Grove.
\newblock Two modellings for theory change.
\newblock {\em Journal of Philosophical Logic}, 17(2):157--170, 1988.

\bibitem{Ove11:SEP}
Sven Ove~Hansson.
\newblock Logic of belief revision.
\newblock In Edward~N. Zalta, editor, {\em The {S}tanford {E}ncyclopedia of
  {P}hilosophy}, 2011.

\bibitem{Segerberg71}
Krister Segerberg.
\newblock Qualitative probability in a modal setting.
\newblock In J.~E. Fenstad, editor, {\em Proceedings of the Second Scandinavian
  Logic Symposium}, volume~63 of {\em Studies in Logic and the Foundations of
  Mathematics}, pages 341--352. Elsevier, 1971.

\end{thebibliography}

\end{document}